\documentclass[a4paper,11pt]{article}
\usepackage[margin = 1in]{geometry}
\usepackage[utf8]{inputenc}
\usepackage[T1]{fontenc}
\usepackage[dvipsnames]{xcolor}
\usepackage{graphicx}
\usepackage{amsmath, amssymb, amstext, mathtools}
\usepackage{amsthm}

\usepackage{tikz}
\usepackage{pgfplots}

\usepackage{float}
\usepackage{hyperref}
\usepackage{paralist}
\usepackage{enumitem}
\usepackage[titlenumbered, linesnumbered, ruled]{algorithm2e}

\usepackage{subcaption}

\allowdisplaybreaks

\definecolor[named]{urlblue}{cmyk}{1,0.58,0,0.21}

\hypersetup{
 breaklinks=true,
 colorlinks=true,
 citecolor=purple!70!blue!60!black,
 linkcolor=purple!70!blue!90!black,
 urlcolor=urlblue,
 pdflang={en},
 pdftitle={Isomorphism Testing for Graphs Excluding Small Topological Subgraphs},
 pdfauthor={Daniel Neuen}
}

\usetikzlibrary{patterns,arrows,decorations.pathreplacing,decorations.pathmorphing,calc}
\usetikzlibrary{hobby}

\newcommand{\gurke}{
   \draw[draw=black,thick,fill=green,fill opacity=0.3,closed]
  (-1,0) .. (0,1) .. (1,0) .. (0.5,-2) .. (1,-4) .. (0,-5) .. (-1,-4) .. (-0.5,-2);
}

\newcommand{\randB}{
   \draw[draw=red,thick,dashed,closed]
  (-1.5,0) .. (0,1.5) .. (1.5,0) .. (1,-2) .. (1.5,-4.5) .. (0,-5.5) .. (-1.5,-4.5) .. (-1,-2);
}

\newcommand{\randA}{
   \draw[draw=red,thick,dashed,closed]
  (0,1.5) .. (2.5,-0.5) .. (4,-3) .. (0,-5.5) .. (-4,-3) .. (-2.5,-0.5);
}

\tikzstyle{vertex}=[circle,fill=white,draw=black,minimum size=7pt,inner sep=0pt]
\tikzstyle{smallvertex}=[circle,fill=white,draw=black,minimum size=4pt,inner sep=0pt]

\definecolor{darkpastelgreen}{rgb}{0.01, 0.75, 0.24}

\newtheorem{theorem}{Theorem}[section]
\newtheorem{lemma}[theorem]{Lemma}
\newtheorem{corollary}[theorem]{Corollary}
\newtheorem{fact}[theorem]{Fact}

\newtheorem{observation}[theorem]{Observation}

\theoremstyle{definition}
\newtheorem{definition}[theorem]{Definition}

\theoremstyle{remark}
\newtheorem{remark}[theorem]{Remark}
\newtheorem{claim}[theorem]{Claim}

\newenvironment{claimproof}{\begin{proof}}{\end{proof}}

\newcommand{\ColRef}[1]{\chi_{{\sf CR}}[#1]}
\newcommand{\ColRefIt}[2]{\chi_{(#1)}[#2]}

\newcommand{\WL}[2]{\chi_{\sf WL}^{#1}\![#2]}
\newcommand{\WLit}[3]{\chi_{(#2)}^{#1}[#3]}

\newcommand{\tColRef}[2]{\chi_{#1\text{-}{\sf CR}}[#2]}

\newcommand{\CD}{{\mathcal D}}
\newcommand{\CE}{{\mathcal E}}

\newcommand{\CH}{{\mathcal H}}

\newcommand{\CM}{{\mathcal M}}

\newcommand{\CO}{{\mathcal O}}
\newcommand{\CP}{{\mathcal P}}
\newcommand{\CQ}{{\mathcal Q}}

\newcommand{\CX}{{\mathcal X}}

\newcommand{\CZ}{{\mathcal Z}}

\newcommand{\NN}{{\mathbb N}}

\newcommand{\Fp}{{\mathfrak p}}

\DeclareMathOperator{\id}{id}

\DeclareMathOperator{\Iso}{Iso}
\DeclareMathOperator{\Aut}{Aut}
\DeclareMathOperator{\Sym}{Sym}

\DeclareMathOperator{\mgamma}{\widehat{\Gamma}}

\DeclareMathOperator{\dist}{dist}
\DeclareMathOperator{\cl}{cl}

\DeclareMathOperator{\fw}{fw}
\DeclareMathOperator{\bw}{bw}

\DeclareMathOperator{\Exp}{Exp}

\newcommand{\adeg}{a_{{\sf deg}}}

\newcommand{\polylog}{\operatorname{polylog}}

\newcommand{\orcid}[1]{\href{https://orcid.org/#1}{\includegraphics[height=1.8ex]{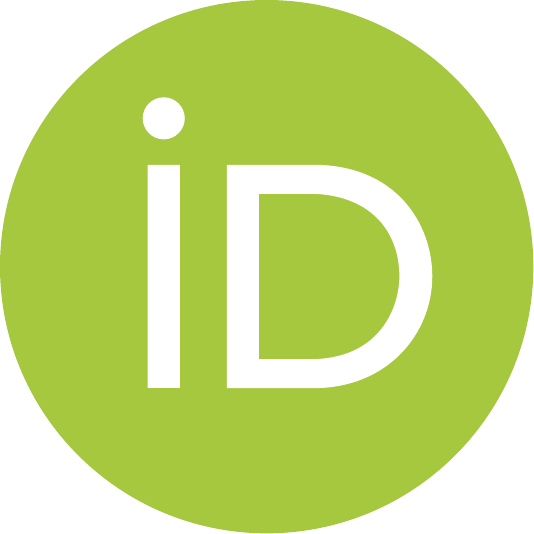}}}
\newcommand{\email}[1]{\href{mailto:#1}{\texttt{#1}}}

\title{Isomorphism Testing for Graphs Excluding Small Topological Subgraphs}
\author{
Daniel Neuen \orcid{0000-0002-4940-0318}\\
University of Bremen\\
\email{dneuen@uni-bremen.de}
}
\date{}

\begin{document}

\maketitle

\begin{abstract}
 We give an isomorphism test that runs in time $n^{\polylog(h)}$ on all $n$-vertex graphs excluding some $h$-vertex graph as a topological subgraph.
 Previous results state that isomorphism for such graphs can be tested in time $n^{\polylog(n)}$ (Babai, STOC 2016) and $n^{f(h)}$ for some function $f$ (Grohe and Marx, SIAM J.\ Comp., 2015).
 
 Our result also unifies and extends previous isomorphism tests for graphs of maximum degree $d$ running in time $n^{\polylog(d)}$ (SIAM J.\ Comp., 2023) and for graphs of Hadwiger number $h$ running in time $n^{\polylog(h)}$ (SIAM J.\ Comp., 2023).
\end{abstract}

\section{Introduction}

Determining the computational complexity of the Graph Isomorphism Problem (GI) is a long-standing open question in theoretical computer science (see, e.g., \cite{Karp72}).
The problem is easily seen to be contained in NP, but it is neither known to be in PTIME nor known to be NP-complete.
In a breakthrough result, Babai \cite{Babai16} obtained a quasipolynomial-time algorithm for testing isomorphism of graphs (i.e., an algorithm running in time $n^{\polylog(n)}$ where $n$ denotes the number of vertices of the input graphs), achieving the first improvement over the previous best algorithm running in time $n^{\CO(\sqrt{n / \log n})}$ \cite{BabaiKL83} in over three decades.
The algorithm advances both group-theoretic techniques for GI, dating back to Luks \cite{Luks82}, as well as our understanding of combinatorial approaches such as the Weisfeiler-Leman algorithm (see, e.g., \cite{CaiFI92,ImmermanL90}).
However, it remains wide open whether GI can be solved in polynomial time.

Polynomial-time algorithms for testing isomorphism are known for various restricted classes of graphs (see, e.g., \cite{BabaiGM82,Bodlaender90,GroheM15,GroheS15,HopcroftT71,Luks82,Neuen16,Miller83a,Ponomarenko91}).
One of the most general results in this direction due to Grohe and Marx \cite{GroheM15} states that isomorphism can be tested in polynomial time for all graph classes that exclude a fixed graph as a topological subgraph.
In particular, this includes previous results solving GI in polynomial time on all graphs excluding a fixed graph as a minor \cite{Ponomarenko91} and all graphs of bounded maximum degree \cite{Luks82}.

A common feature of all these algorithms is that the exponent of the running time depends at least linearly on the parameter in question.
In light of Babai's quasipolynomial-time algorithm it seems natural to ask for which graph parameters $k$ the Graph Isomorphism Problem can be solved in time $n^{\polylog(k)}$.
The first step towards answering this question was achieved by Grohe, Schweitzer and the present author \cite{GroheNS23} presenting an isomorphism algorithm running in time $n^{\polylog(d)}$ where $n$ denotes the number of vertices and $d$ the maximum degree of the input graphs.
Further extending the group-theoretic advances of Babai's algorithm, subsequent work resulted in algorithms testing isomorphism in time $n^{\polylog(k)}$ for graphs of tree-width $k$ \cite{Wiebking20} and time $n^{\polylog(g)}$ for graphs of Euler genus $g$ \cite{Neuen22}.
Both of these results were generalized in \cite{GroheNW23} obtaining an isomorphism test running in time $n^{\polylog(h)}$ for $n$-vertex graphs excluding an arbitrary $h$-vertex graph as a minor.
For a recent survey, we also refer to \cite{GroheN21}.
In this work we further generalize these results to all graphs that exclude an arbitrary $h$-vertex graph as a topological subgraph.

Recall that a graph $H$ is a topological subgraph of a graph $G$ if $H$ can be obtained from $G$ by deleting vertices and edges as well as dissolving degree two vertices (i.e., deleting a vertex of degree two and connecting its two neighbors).
A graph $G$ excludes $H$ as a topological subgraph if $G$ has no topological subgraph that is isomorphic to $H$.
For example, by Kuratowski's Theorem, planar graphs can be characterized by excluding $K_5$ and $K_{3,3}$ as topological subgraphs.
Note that, whenever a graph $H$ is a topological subgraph of $G$, then $H$ is also a minor of $G$.
Hence, any graph class that excludes some graph $H$ as a minor in particular excludes $H$ as a topological subgraph.
As another observation, the maximum degree of a topological subgraph $H$ of $G$ is at most the maximum of degree of $G$.
Thus, any graph of maximum degree $d$ excludes the complete graph $K_{d+2}$ on $d+2$ vertices as a topological subgraph.

The main result of this work is a new isomorphism algorithm for graphs that exclude some $h$-vertex graph as a topological subgraph which significantly improves the previous best algorithm due to Grohe and Marx \cite{GroheM15}.

\begin{theorem}
 The Graph Isomorphism Problem for graphs excluding some $h$-vertex graph as a topological subgraph can be solved in time $n^{\polylog(h)}$.
\end{theorem}

By the comments above, this result also unifies and extends previous isomorphism tests for graphs of maximum degree $d$ running in time $n^{\polylog(d)}$ \cite{GroheNS23} and for graphs of Hadwiger number $h$ (i.e., the maximum $h$ such that $K_h$ is a minor of the input graph) running in time $n^{\polylog(h)}$ \cite{GroheNW23}.

Observe that a graph $G$ excludes some $h$-vertex graph as a topological subgraph if and only if it excludes $K_h$, the complete graph on $h$ vertices, as a topological subgraph.
Hence, in the following, we restrict our attention to graphs that exclude $K_h$ as a topological subgraph.

On a high level, the algorithm follows the same decomposition strategy that is already used by Grohe et al.\ in \cite{GroheNW23} for testing isomorphism of graphs excluding $K_h$ as a minor.
The main idea is to decompose an input graph $G$ into parts $D \subseteq V(G)$ such that the interplay between the parts is simple, and $G$ restricted to $D$ is \emph{$t$-CR-bounded} for some number $t$ that is polynomially bounded in $h$. 
Intuitively speaking, a graph $G$ is \emph{$t$-CR-bounded} (where $t \in \NN$) if an initially uniform coloring can be transformed into a discrete coloring (i.e., a coloring where every vertex has its own color) by repeatedly applying the standard Color Refinement algorithm and splitting color classes of size at most $t$.
Building on the group-theoretic isomorphism machinery dating back to Luks \cite{Luks82} as well as recent extensions \cite{Neuen22}, isomorphism of $t$-CR-bounded graphs can be decided in time $n^{\polylog(t)}$ (where $n$ denotes the number of vertices of the input graphs).

In order to decompose the input graph $G$ into suitable parts $D \subseteq V(G)$, Grohe et al.\ \cite{GroheNW23} introduce the \emph{$t$-closure} $\cl_t^G(X)$ of a set $X \subseteq V(G)$ as the set of all uniquely colored vertices after artificially individualizing all vertices from $X$ and applying the $t$-CR procedure (i.e., repeatedly applying the standard Color Refinement algorithm and splitting color classes of size at most $t$).
Now, the central idea is to define $D \coloneqq \cl_t^{G}(X)$ for a suitable set $X$ (we will refer to $X$ as the \emph{initial set}).
In order for this approach to work out, the following two statements are crucial where $t$ is some number polynomially bounded in $h$.
\begin{enumerate}[label = (\Alph*)]
 \item\label{item:main-tool-1} For every $X \subseteq V(G)$ it holds that $|N_G(Z)| < h$ for every connected component $Z$ of $G - D$ where $D \coloneqq \cl_t^{G}(X)$.
 \item\label{item:main-tool-2} There is a polynomial-time algorithm that computes an isomorphism-invariant initial set $\emptyset \neq X \subseteq V(G)$ such that $X \subseteq \cl_t^G(v)$ for every $v \in X$.
\end{enumerate}

Assuming both statements are true, one can build an isomorphism test for graphs $G_1$ and $G_2$ as follows.
First, we compute sets $X_1 \subseteq V(G_1)$ and $X_2 \subseteq V(G_2)$ using Property \ref{item:main-tool-2} and define $D_i \coloneqq \cl_t^{G_i}(X_i)$ for both $i \in \{1,2\}$.
Afterwards, we recursively compute isomorphisms between all pairs of connected components of $G_1 - D_1$ and $G_2 - D_2$.
By Property \ref{item:main-tool-2}, $G_i$ is $t$-CR-bounded on $D_i$ after individualizing a single vertex $v \in X_i$.
Hence, isomorphism between $G_1[D_1]$ and $G_2[D_2]$ can be decided in time $n^{\polylog(t)} = n^{\polylog(h)}$.
Also, using the techniques from \cite{GroheNSW20,Wiebking20} and Property \ref{item:main-tool-1}, it is possible to incorporate the partial isomorphisms between the connected components of $G_1 - D_1$ and $G_2 - D_2$, overall resulting in an isomorphism test between $G_1$ and $G_2$ running in time $n^{\polylog(h)}$.
We remark that, for this strategy to work out, it is crucial to define $D_i$ in an isomorphism-invariant manner as to not compare two graphs that are decomposed in structurally different ways.
Observe that $D_i$ is indeed defined in an isomorphism-invariant way since the initial set $X_i$ is by Property \ref{item:main-tool-2}.

While Grohe et al.\ \cite{GroheNW23} already prove Property \ref{item:main-tool-1} for graphs excluding $K_h$ as a topological subgraph, their proof of Property \ref{item:main-tool-2} crucially requires closure under taking minors.
More precisely, for Property \ref{item:main-tool-2}, Grohe et al.\ provide a complicated algorithm that essentially contracts certain parts of the input graph $G$, and then builds the set $X$ from a solution $X'$ computed by recursion for the contracted graph $G'$.
Unfortunately, such a strategy is infeasible for graphs excluding $K_h$ as a topological subgraph since the corresponding class of graphs is not closed under taking minors.

The main technical contribution of this work is to provide an alternative algorithm for Property \ref{item:main-tool-2} that works for all graphs $G$ excluding $K_h$ as topological subgraph.
Indeed, our algorithm is much simpler than the one from \cite{GroheNW23} and solely relies on the well-known Weisfeiler-Leman algorithm (see, e.g., \cite{CaiFI92,ImmermanL90}).
The Weisfeiler-Leman algorithm is a standard tool in the context of isomorphism testing and computes an isomorphism-invariant coloring of $k$-tuples of vertices of a graph $G$.
Our main technical result is that the $3$-dimensional Weisfeiler-Leman algorithm is able to provide a suitable initial set $X$.

\begin{theorem}
 \label{thm:main-tool}
 Let $G$ be a connected graph that excludes $K_h$ as a topological subgraph and let $t = \Omega(h^5)$.
 Then there is a color $c_0 \in \{\WL{3}{G}(v,v,v) \mid v \in V(G)\}$ such that, for $\chi(v,w) \coloneqq \WL{3}{G}(v,w,w)$ for all $v,w \in V(G)$ and $X \coloneqq \{v \in V(G) \mid \WL{3}{G}(v,v,v) = c_0\}$, it holds that
 \[X \subseteq \cl_t^{(G,\chi)}(v)\]
 for all $v \in X$.
\end{theorem}

Here, $\WL{3}{G} \colon (V(G))^3 \rightarrow C$ denotes the coloring of $3$-tuples computed by the $3$-dimensional Weisfeiler-Leman algorithm.
We remark that the $t$-closure needs to be taken over a colored version of $G$ (where every pair $(v,w)$ of vertices is colored by $\chi(v,w)$) which, however, does not pose any problems for the final algorithm.

Observe that Theorem \ref{thm:main-tool} implies Property \ref{item:main-tool-2} since the coloring $\WL{3}{G}$ can be computed in polynomial time and we can simply try all possible colors $c_0 \in \{\WL{3}{G}(v,v,v) \mid v \in V(G)\}$ to find a good set $X$ (if more than one color yields a good set $X$, we simply take the smallest color according to some fixed linear order on the colors).

The proof of Theorem \ref{thm:main-tool} builds on a lengthy and technical analysis of the coloring $\chi$.
As a main tool for the proof, we introduce the \emph{$t$-closure graph of $(G,\chi)$} which is a directed graph $H$ on the same vertex set as $G$ and with an edge $(v,w)$ for all $v,w \in V(G)$ such that $w \in \cl_t^{(G,\chi)}(v)$.
Building on properties of the $3$-dimensional Weisfeiler-Leman algorithm, we show that it is possible to choose $X \coloneqq \{v \in V(G) \mid \WL{3}{G}(v,v,v) = c_0\}$ in such a way that $X$ only contains vertices appearing in maximal strongly connected components of $H$ (a strongly connected component of $H$ is maximal if it has no outgoing edges).
This implies that $\cl_t^{(G,\chi)}(v)$ and $\cl_t^{(G,\chi)}(v')$ are either disjoint or equal for all $v,v' \in X$, since $\cl_t^{(G,\chi)}(v)$ is exactly the strongly connected component of $H$ that contains $v$.
Assuming there are $v,v' \in X$ such that $\cl_t^{(G,\chi)}(v)$ and $\cl_t^{(G,\chi)}(v')$ are disjoint, we proceed by constructing a large number of pairwise internally vertex-disjoint paths between such sets leading to a topological subgraph of $G$ with high edge-density.
Eventually, this contradicts the fact that the average degree of every topological subgraph of $G$ is bounded by a polynomial function in $h$ \cite{BollobasT98,KomlosS96}.

Besides the isomorphism test for graphs excluding $K_h$ as a topological subgraph, the algorithmic approach taken in this paper also provides some structural insights into the automorphism groups of graphs without a topological subgraph isomorphic to $K_h$.
We show that every such graph $G$ admits a tree decomposition of adhesion width at most $h-1$ (i.e., the intersection between any two bags has size at most $h-1$) such that the automorphism group of $G$ restricted to a single bag is similar to those of graphs of bounded maximum degree.
More precisely, for $d \geq 1$, let $\mgamma_d$ denote the class of groups $\Gamma$ such that every composition factor of $\Gamma$ is isomorphic to a subgroup of $S_d$ (the symmetric group on $d$ points).
It is well-known that the automorphism group of every connected graph $G$ of maximum degree $d$ is in the class $\mgamma_d$ after individualizing a single vertex of $G$ \cite{Luks82}.
We obtain the following structural insights on the automorphism group $\Aut(G)$ of a graph $G$ excluding $K_h$ as a topological subgraph.

\begin{theorem}
 Let $G$ be a graph that excludes $K_h$ as a topological subgraph.
 Then there is an iso\-mor\-phism-invariant tree decomposition $(T,\beta)$ of $G$ such that
 \begin{enumerate}
  \item the adhesion-width of $(T,\beta)$ is at most $h-1$, and
  \item for every $t \in V(T)$ there is some $v \in \beta(t)$ such that $(\Aut(G))_v[\beta(t)] \in \mgamma_d$ for some $d = \CO(h^5)$.
 \end{enumerate}
\end{theorem}

Here, $(\Aut(G))_v[\beta(t)]$ denotes the restriction of $\Aut(G)$ to the bag $\beta(t)$ after individualizing the vertex $v$.

The remainder of this work is structured as follows.
In the next section we give the necessary preliminaries.
Afterwards, we define $t$-CR-bounded graphs and the corresponding closure operator in Section \ref{sec:t-cr} and provide a more detailed overview on the main algorithm.
In Section \ref{sec:initial-color} we prove the main technical theorem of this work which provides a suitable initial set $X$.
Based on this result, we assemble the main algorithm in Sections \ref{sec:group-machinery} and \ref{sec:algorithm-isomorphism}.

\section{Preliminaries}
\label{sec:preliminaries}

\subsection{Graphs}

A \emph{graph} is a pair $G = (V(G),E(G))$ consisting of a \emph{vertex set} $V(G)$ and an \emph{edge set} $E(G)$.
All graphs considered in this paper are finite and simple (i.e., they contain no loops or multiple edges).
Moreover, unless explicitly other stated, all graphs are undirected.
For an undirected graph $G$ and $v,w \in V(G)$, we write $vw$ as a shorthand for $\{v,w\} \in E(G)$.
The \emph{neighborhood} of a vertex $v \in V(G)$ is denoted by $N_G(v)$.
The \emph{degree} of $v$, denoted by $\deg_G(v)$, is the number of edges incident with $v$, i.e., $\deg_G(v)=|N_G(v)|$.
For $X \subseteq V(G)$, we define $N_G[X] \coloneqq X \cup \bigcup_{v \in X}N_G(v)$ and $N_G(X) \coloneqq N_G[X] \setminus X$.
If the graph $G$ is clear from context, we usually omit the index and simply write $N(v)$, $\deg(v)$, $N[X]$ and $N(X)$.

We write $K_n$ to denote the complete graph on $n$ vertices.
A graph is \emph{regular} if every vertex has the same degree.
A bipartite graph $G=(V_1,V_2,E)$ is called \emph{$(d_1,d_2)$-biregular} if all vertices $v_i \in V_i$ have degree $d_i$ for both $i \in \{1,2\}$.
In this case $d_1 \cdot |V_1| = d_2 \cdot |V_2| = |E|$.
By a double edge counting argument, for each subset $S \subseteq V_i$, $i\in\{1,2\}$, it holds that $|S| \cdot d_i \leq |N_G(S)| \cdot d_{3-i}$.
A bipartite graph is \emph{biregular}, if there are $d_1,d_2 \in \NN$ such that $G$ is $(d_1,d_2)$-biregular.
Each biregular graph satisfies the Hall condition, i.e., for all $S \subseteq V_1$ it holds $|S| \leq |N_G(S)|$ (assuming $|V_1| \leq |V_2|$).
Thus, by Hall's Marriage Theorem, each biregular graph contains a matching of size $\min(|V_1|,|V_2|)$.

A \emph{path of length $k$} from $v$ to $w$ is a sequence of distinct vertices $v = u_0,u_1,\dots,u_k = w$ such that $u_{i-1}u_i \in E(G)$ for all $i \in [k] \coloneqq \{1,\dots,k\}$.
The \emph{distance} between two vertices $v,w \in V(G)$, denoted by $\dist_G(v,w)$, is the length of a shortest path between $v$ and $w$.
As before, we omit the index $G$ it it is clear from context.
For two sets $A,B\subseteq V(G)$, we denote by $E_G(A,B) \coloneqq \{vw\in E(G)\mid v\in A,w\in B\}$.
Also, $G[A,B]$ denotes the graph with vertex set $A\cup B$ and edge set $E_G(A,B)$.
We write $G[A] \coloneqq G[A,A]$ to denote the \emph{induced subgraph} on vertex set $A$.
Also, we denote by $G - A$ the subgraph induced by the complement of $A$, that is, the graph $G - A \coloneqq G[V(G) \setminus A]$.
A graph $H$ is a \emph{subgraph} of $G$, denoted by $H \subseteq G$, if $V(H) \subseteq V(G)$ and $E(H) \subseteq E(G)$. 
A set $S \subseteq V(G)$ is a \emph{separator} of $G$ if $G - S$ has more connected components than $G$.
A \emph{$k$-separator} of $G$ is a separator of $G$ of size $k$.

An \emph{isomorphism} from $G$ to a graph $H$ is a bijection $\varphi\colon V(G) \rightarrow V(H)$ that respects the edge relation, that is, for all~$v,w \in V(G)$, it holds that~$vw \in E(G)$ if and only if $\varphi(v)\varphi(w) \in E(H)$.
Two graphs $G$ and $H$ are \emph{isomorphic}, written $G \cong H$, if there is an isomorphism from~$G$ to~$H$.
We write $\varphi\colon G\cong H$ to denote that $\varphi$ is an isomorphism from $G$ to $H$.
Also, $\Iso(G,H)$ denotes the set of all isomorphisms from $G$ to $H$.
The automorphism group of $G$ is $\Aut(G) \coloneqq \Iso(G,G)$.
Observe that, if $\Iso(G,H) \neq \emptyset$, it holds that $\Iso(G,H) = \Aut(G)\varphi \coloneqq \{\gamma\varphi \mid \gamma \in \Aut(G)\}$ for every isomorphism $\varphi \in \Iso(G,H)$.

A \emph{vertex-colored graph} is a tuple $(G,\chi_V)$ where $G$ is a graph and $\chi_V\colon V(G) \rightarrow C$ is a mapping into some set $C$ of colors, called \emph{vertex-coloring}.
Similarly, an \emph{arc-colored graph} is a tuple $(G,\chi_E)$, where $G$ is a graph and $\chi_E\colon\{(u,v) \mid \{u,v\} \in E(G)\} \rightarrow C$ is a mapping into some color set $C$, called \emph{arc-coloring}.
Observe that colors are assigned to directed edges, i.e., the directed edge $(v,w)$ may obtain a different color than $(w,v)$.
We also consider vertex- and arc-colored graphs $(G,\chi_V,\chi_E)$ where $\chi_V$ is a vertex-coloring and $\chi_E$ is an arc-coloring.
Also, a \emph{pair-colored graph} is a tuple $(G,\chi_P)$, where $G$ is a graph and $\chi_P\colon (V(G))^{2} \rightarrow C$ is a mapping into some color set $C$.
Typically, $C$ is chosen to be an initial segment $[n]$ of the natural numbers.
To be more precise, we generally assume that there is a total order on the set of all potential colors which, for example, allows us to identify a minimal color appearing in a graph in a unique way.
Isomorphisms between vertex-, arc- and pair-colored graphs have to respect the colors of the vertices, arcs and pairs.

\subsection{Topological Subgraphs}

A graph $H$ is a \emph{topological subgraph} of $G$ if $H$ can be obtained from $G$ by deleting edges, deleting vertices and dissolving degree $2$ vertices (which means deleting the vertex and making its two neighbors adjacent).
More formally, we say that $H$ is a topological subgraph of $G$ if a subdivision of $H$ is a subgraph of $G$
(a subdivision of a graph $H$ is obtained by replacing each edge of $H$ by a path of length at least 1).
The following theorem states the well-known fact that graphs excluding a topological subgraph have bounded average degree.

\begin{theorem}[\cite{BollobasT98,KomlosS96}]
 \label{thm:average-degree-excluded-topological}
 There is an absolute constant $\adeg \geq 1$ such that for every $h \geq 1$ and every graph $G$ that excludes $K_h$ as a topological subgraph, it holds that
 \[\frac{1}{|V(G)|}\sum_{v \in V(G)} \deg_G(v) \leq \adeg h^2.\]
\end{theorem}

In the remainder of this work, we always use $\adeg$ to refer to the constant from the theorem above.

\subsection{Weisfeiler-Leman Algorithm}

The Weisfeiler-Leman algorithm, originally introduced by Weisfeiler and Leman in its $2$-di\-men\-sio\-nal version \cite{WeisfeilerL68} (see also \cite{Weisfeiler76}), forms one of the most fundamental subroutines in the context of isomorphism testing.

Let~$\chi_1,\chi_2\colon V^k \rightarrow C$ be colorings of the~$k$-tuples of vertices of~$G$, where~$C$ is some finite set of colors. 
We say $\chi_1$ \emph{refines} $\chi_2$, denoted $\chi_1 \preceq \chi_2$, if $\chi_1(\bar v) = \chi_1(\bar w)$ implies $\chi_2(\bar v) = \chi_2(\bar w)$ for all $\bar v,\bar w \in V^{k}$.
The colorings $\chi_1$ and $\chi_2$ are \emph{equivalent}, denoted $\chi_1 \equiv \chi_2$,  if $\chi_1 \preceq \chi_2$ and $\chi_2 \preceq \chi_1$.

The \emph{Color Refinement algorithm} (i.e., the $1$-dimensional Weisfeiler-Leman algorithm) is a procedure that, given a graph $G$, iteratively computes an isomorphism-invariant coloring of the vertices of $G$.
In this work, we actually require an extension of the Color Refinement algorithm that apart from vertex-colors also takes arc-colors into account.
For a vertex- and arc-colored graph $(G,\chi_V,\chi_E)$ define $\ColRefIt{0}{G} \coloneqq \chi_V$ to be the initial coloring for the algorithm.
This coloring is iteratively refined by defining $\ColRefIt{i+1}{G}(v) \coloneqq (\ColRefIt{i}{G}(v), \CM_i(v))$ where
\[\CM_i(v) \coloneqq\big\{\!\!\big\{\big(\ColRefIt{i}{G}(w),\chi_E(v,w),\chi_E(w,v)\big) \;\big\vert\; w \in N_G(v) \big\}\!\!\big\}\]
(and $\{\!\!\{ \dots \}\!\!\}$ denotes a multiset).
By definition, $\ColRefIt{i+1}{G} \preceq \ColRefIt{i}{G}$ for all $i \geq 0$.
Hence, there is a minimal value $i_\infty$ such that $\ColRefIt{i_\infty}{G} \equiv \ColRefIt{i_\infty+1}{G}$.
We define $\ColRef{G} \coloneqq \ColRefIt{i_\infty}{G}$.
The Color Refinement algorithm takes as input a vertex- and arc-colored graph $(G,\chi_V,\chi_E)$ and returns (a coloring that is equivalent to) $\ColRef{G}$.
The procedure can be implemented in time $\CO((m+n)\log n)$ (see, e.g., \cite{BerkholzBG17}) where $n$ and $m$ denote the number of vertices and edges of $G$, respectively.

Next, we describe the \emph{$k$-dimensional Weisfeiler-Leman algorithm} ($k$-WL) for all $k \geq 2$.
For an input graph $G$ let $\WLit{k}{0}{G}\colon (V(G))^{k} \rightarrow C$ be the coloring where each tuple is colored with the isomorphism type of its underlying ordered subgraph.
More precisely, $\WLit{k}{0}{G}(v_1,\dots,v_k) = \WLit{k}{0}{G}(v_1',\dots,v_k')$ if and only if, for all $i,j \in [k]$, it holds that
$v_i = v_j \Leftrightarrow v_i'= v_j'$ and $v_iv_j \in E(G) \Leftrightarrow v_i'v_j' \in E(G)$.
If the graph is equipped with a coloring the initial coloring $\WLit{k}{0}{G}$ also takes the input coloring into account.
More precisely, for a vertex-coloring $\chi_V$, it additionally holds that $\chi_V(v_i) = \chi_V(v_i')$ for all $i \in [k]$.
For an arc-coloring $\chi_E$, it is the case that $\chi_E(v_i,v_j) = \chi_E(v_i',v_j')$ for all $i,j \in [k]$ such that $v_iv_j \in E(G)$.
Finally, for a pair coloring $\chi_P$, it holds that $\chi_P(v_i,v_j) = \chi_P(v_i',v_j')$ is additionally satisfied for all $i,j \in [k]$.

We then recursively define the coloring $\WLit{k}{i}{G}$ obtained after $i$ rounds of the algorithm.
For $\bar v = (v_1,\dots,v_k) \in (V(G))^k$ let
\[\WLit{k}{i+1}{G}(\bar v) \coloneqq \Big(\WLit{k}{i}{G}(\bar v), \CM_i(\bar v)\Big)\]
where
\[\CM_i(\bar v) \coloneqq \Big\{\!\!\Big\{\big(\WLit{k}{i}{G}(\bar v[w/1]),\dots,\WLit{k}{i}{G}(\bar v[w/k])\big) \;\Big\vert\; w \in V(G) \Big\}\!\!\Big\}\]
and $\bar v[w/i] \coloneqq (v_1,\dots,v_{i-1},w,v_{i+1},\dots,v_k)$ is the tuple obtained from $\bar v$ by replacing the $i$-th entry by $w$.
Again, there is a minimal~$i_\infty$ such that $\WLit{k}{i_{\infty}}{G} \equiv \WLit{k}{i_{\infty}+1}{G}$ and for this $i_\infty$ we define $\WL{k}{G} \coloneqq \WLit{k}{i_\infty+1}{G}$.

The $k$-dimensional Weisfeiler-Leman algorithm takes as input a (vertex-, arc- or pair-)colored graph $G$ and returns (a coloring that is equivalent to) $\WL{k}{G}$.
This can be implemented in time $\CO(n^{k+1}\log n)$ (see \cite{ImmermanL90}).

Let $G$ be a graph.
Let $k \geq 1$ and let $\chi\colon(V(G))^{k} \rightarrow C$ be a coloring of $k$-tuples.
We say that $\chi$ is \emph{$k$-stable on $G$} if $\chi \preceq \WLit{k}{0}{G}$ and $\chi$ is not refined by applying one round of the $k$-dimensional Weisfeiler-Leman algorithm (the $1$-dimensional Weisfeiler-Leman algorithm is defined as the Color Refinement algorithm).
In particular, $\WL{k}{G}$ is $k$-stable on $G$.
The following facts are well-known (for example, they follow easily from known characterizations of $k$-WL \cite{CaiFI92,ImmermanL90}).

\begin{fact}
 \label{fact:k-wl-is-stable}
 Let $G$ be a graph and let $1 \leq \ell \leq k$.
 Also, define 
 \[\chi(v_1,\dots,v_\ell) \coloneqq \WL{k}{G}(v_1,\dots,v_\ell,v_\ell,\dots,v_\ell)\]
 for all $v_1,\dots,v_\ell \in V(G)$.
 Then $\chi$ is $\ell$-stable on $G$.
\end{fact}

\begin{fact}
 \label{fact:k-wl-is-stable-after-individualization}
 Let $G$ be a graph and let $1 \leq \ell < k$.
 Let $w \in V(G)$ and define 
 \[\chi(v_1,\dots,v_\ell) \coloneqq \WL{k}{G}(w,v_1,\dots,v_\ell,v_\ell,\dots,v_\ell)\]
 for all $v_1,\dots,v_\ell \in V(G)$.
 Then $\chi$ is $\ell$-stable on $(G,\chi_w)$ where $\chi_w$ is the vertex-coloring defined via $\chi_w(w) \coloneqq 1$ and $\chi_w(v) \coloneqq 0$ for all $v \in V(G) \setminus \{w\}$.
\end{fact}

\section{Allowing Color Refinement to Split Small Color Classes}
\label{sec:t-cr}

In the following, we provide a more detailed overview for the main algorithm testing isomorphism of graphs excluding $K_h$ as a topological subgraph.
On a high-level, the algorithm builds on a decomposition strategy.
Let $G_1$ and $G_2$ denote the two input graphs.
By testing isomorphisms of connected components separately, we may assume without loss of generality that $G_1$ and $G_2$ are connected.
The algorithm aims at finding isomorphism-invariant sets $D_1$ and $D_2$ and recursively computes isomorphisms between all pairs of components of $G_1 - D_1$ and $G_2 - D_2$.
In order to combine the partial isomorphisms between connected components of $G_1 - D_1$ and $G_2 - D_2$ into full isomorphisms between $G_1$ and $G_2$ the algorithm builds on the various group-theoretic tools (to be discussed in Section \ref{sec:group-machinery}).
Here, one of the crucial properties is that $|N_{G_i}(Z)|$ is polynomially bounded in $h$ for every connected component $Z$ of $G_i - D_i$.
To test isomorphisms between $G_1[D_1]$ and $G_2[D_2]$, the algorithms builds on the notion of $t$-CR-bounded graphs originally introduced by Ponomarenko \cite{Ponomarenko89}.

Intuitively speaking, a (vertex- and arc-colored) graph $G$ is $t$-CR-bounded, $t \in \NN$, if the vertex coloring of $G$ can be turned into the discrete coloring (i.e., each vertex has its own color) by repeatedly
\begin{itemize}
 \item performing the Color Refinement algorithm (expressed by the letters `CR'), and
 \item taking a color class $[v]_\chi \coloneqq \{w \in V(G) \mid \chi(w) = \chi(v)\}$ of some vertex $v \in V(G)$ of size $|[v]_\chi| \leq t$ and assigning each vertex from the class its own color.
\end{itemize}
An example is given Figure \ref{fig:t-cr-example}. The next definition formalizes this concept.

\begin{figure}
 \centering
 \scalebox{0.9}{
 \begin{tikzpicture}
  \node at (0.2,2.0) {$\chi_1$:};
  \node[vertex,fill=blue] (v1-1) at (1.5,2.0) {};
  \node[vertex,fill=blue] (v1-2) at (1.5,-1.0) {};
  \node[vertex,fill=darkpastelgreen] (v1-3) at (0.2,0.5) {};
  \node[vertex,fill=darkpastelgreen] (v1-4) at (2.8,0.5) {};
  \node[vertex,fill=red] (v1-5) at (1.1,1.0) {};
  \node[vertex,fill=red] (v1-6) at (1.1,0.0) {};
  \node[vertex,fill=red] (v1-7) at (1.9,0.0) {};
  \node[vertex,fill=red] (v1-8) at (1.9,1.0) {};

  \node at (4.2,2.0) {$\chi_2$:};
  \node[vertex,fill=violet] (v2-1) at (5.5,2.0) {};
  \node[vertex,fill=blue] (v2-2) at (5.5,-1.0) {};
  \node[vertex,fill=darkpastelgreen] (v2-3) at (4.2,0.5) {};
  \node[vertex,fill=orange] (v2-4) at (6.8,0.5) {};
  \node[vertex,fill=red] (v2-5) at (5.1,1.0) {};
  \node[vertex,fill=red] (v2-6) at (5.1,0.0) {};
  \node[vertex,fill=red] (v2-7) at (5.9,0.0) {};
  \node[vertex,fill=red] (v2-8) at (5.9,1.0) {};

  \node at (8.2,2.0) {$\chi_3$:};
  \node[vertex,fill=violet] (v3-1) at (9.5,2.0) {};
  \node[vertex,fill=blue] (v3-2) at (9.5,-1.0) {};
  \node[vertex,fill=darkpastelgreen] (v3-3) at (8.2,0.5) {};
  \node[vertex,fill=orange] (v3-4) at (10.8,0.5) {};
  \node[vertex,fill=red] (v3-5) at (9.1,1.0) {};
  \node[vertex,fill=red] (v3-6) at (9.1,0.0) {};
  \node[vertex,fill=cyan] (v3-7) at (9.9,0.0) {};
  \node[vertex,fill=cyan] (v3-8) at (9.9,1.0) {};

  \node at (12.2,2.0) {$\chi_4$:};
  \node[vertex,fill=violet] (v4-1) at (13.5,2.0) {};
  \node[vertex,fill=blue] (v4-2) at (13.5,-1.0) {};
  \node[vertex,fill=darkpastelgreen] (v4-3) at (12.2,0.5) {};
  \node[vertex,fill=orange] (v4-4) at (14.8,0.5) {};
  \node[vertex,fill=red] (v4-5) at (13.1,1.0) {};
  \node[vertex,fill=yellow] (v4-6) at (13.1,0.0) {};
  \node[vertex,fill=magenta] (v4-7) at (13.9,0.0) {};
  \node[vertex,fill=cyan] (v4-8) at (13.9,1.0) {};

  \foreach \i in {1,...,4}{
   \foreach \v/\w in {1/3,1/4,2/3,2/4,3/5,3/6,4/7,4/8,5/6,5/8,6/7,7/8}{
    \draw[thick] (v\i-\v) edge (v\i-\w);
   }
  }
 \end{tikzpicture}
 }
 \caption{Visualization of a graph $G$ and the sequence of colorings described in Definition \ref{def:t-cr-bounded} for $t = 2$.
  The coloring $\chi_4$ is discrete, so $G$ is $2$-CR-bounded.}
 \label{fig:t-cr-example}
\end{figure}

\begin{definition}
 \label{def:t-cr-bounded}
 Let $t \geq 1$ and let $G = (V,E,\chi_V,\chi_E)$ be a vertex- and arc-colored graph.
 We define the sequence $(\chi_i)_{i \geq 0}$ of colorings where $\chi_0 \coloneqq \chi_V$,
 \[\chi_{2i+1} \coloneqq \WL{1}{V,E,\chi_{2i},\chi_E}\]
 and
 \[\chi_{2i+2}(v) \coloneqq \begin{cases}
                     (v,1)              & \text{if } |[v]_{\chi_{2i+1}}| \leq t\\
                     (\chi_{2i+1}(v),0) & \text{otherwise}
                    \end{cases}\]
 for all $i \geq 0$.
 
 For the minimal $i_\infty \geq 0$ such that $\chi_{i_\infty} \equiv \chi_{i_\infty+1}$, we refer to $\chi_{i_\infty}$ as the \emph{$t$-CR-stable} coloring of $G$ and denote it by $\tColRef{t}{G}$.

 The graph $G$ is \emph{$t$-CR-bounded} if $\tColRef{t}{G}$ is discrete.
\end{definition}

\begin{remark}
 Recall that we generally assume the set of colors be to linearly ordered.
 For the colorings defined in Definition \ref{def:t-cr-bounded}, we do not require this property as we are only interested in the partition into color classes.
\end{remark}

Assuming $G_1[D_1]$ and $G_2[D_2]$ are $t$-CR-bounded, isomorphisms between the subgraphs can be found in time $n^{\polylog(t)}$ building on the group-theoretic graph isomorphism machinery \cite{Neuen22}.
Also, using the tools from \cite{Wiebking20}, one can incorporate the partial isomorphisms between connected components of $G_1 - D_1$ and $G_2 - D_2$, assuming $|N_{G_i}(Z)|$ is polynomially bounded in $h$ for every connected component $Z$ of $G_i - D_i$.
Hence, the main task is to find suitable sets $D_1$ and $D_2$.
Here, we follow the same strategy as in \cite{GroheNW23} and rely on a closure operator associated with $t$-CR-bounded graphs.
Let $G$ be a graph and let $X \subseteq V(G)$ be a set of vertices.
Let $\chi_V^*\colon V(G) \rightarrow C$ be the vertex-coloring obtained by individualizing all vertices in the set $X$, i.e., $\chi_V^*(v) \coloneqq (v,1)$ for $v \in X$ and $\chi_V^*(v) \coloneqq (0,0)$ for $v \in V(G) \setminus X$.
Let $\chi \coloneqq \tColRef{t}{G,\chi_V^*}$ denote the $t$-CR-stable coloring with respect to the input graph $(G,\chi_V^*)$.
We define the \emph{$t$-closure} of the set $X$ (with respect to $G$) to be the set
\[\cl_t^G(X) \coloneqq \left\{v \in V(G) \mid |[v]_{\chi}| = 1\right\}.\]
Observe that $X \subseteq \cl_t^G(X)$.
For $v_1,\dots,v_\ell \in V(G)$ we also use $\cl_t^G(v_1,\dots,v_\ell)$ as a shorthand for $\cl_t^G(\{v_1,\dots,v_\ell\})$.
If the input graph is equipped with a vertex- or arc-coloring, all definitions are extended in the natural way. 

We also build the $t$-closure for vertex sets over pair-colored graphs $(G,\chi)$.
Let $n$ denote the number of vertices of $G$.
We define a vertex- and arc-coloring $\widetilde{\chi}_V$ and $\widetilde{\chi}_E$ of the complete graph $K_n$ by $\widetilde{\chi}_V(v) \coloneqq \chi(v,v)$ and $\widetilde{\chi}_E(v,w) \coloneqq (1,\chi(v,w))$ for all $vw \in E(G)$ and $\widetilde{\chi}_E(v,w) \coloneqq (0,\chi(v,w))$ for all $v,w \in V(G)$ such that $v \neq w$ and $vw \notin E(G)$.
Now, for $X \subseteq V(G)$ let
\[\cl_t^{(G,\chi)}(X) \coloneqq \cl_t^{(K_n,\widetilde{\chi}_V,\widetilde{\chi}_E)}(X).\]
As before, for vertices $v_1,\dots,v_\ell \in V(G)$, we again write $\cl_t^{(G,\chi)}(v_1,\dots,v_\ell)$ as a shorthand for $\cl_t^{(G,\chi)}(\{v_1,\dots,v_\ell\})$.

Following the general strategy outlined above, the goal is to compute suitable isomorphism-invariant sets $X_i \subseteq V(G_i)$ and define $D_i \coloneqq \cl_t^{G_i}(X_i)$ for both $i \in \{1,2\}$ (for some suitable choice of $t$).
We refer to the sets $X_1$ and $X_2$ as the \emph{initial sets}.
More precisely, the algorithms aims at finding isomorphism-invariant initial sets $X_i \subseteq V(G_i)$ and pair-colorings $\chi_i \colon (V(G_i))^2 \rightarrow C$ such that the following properties hold for the some $t \in \NN$ that is polynomially bounded in $h$.
\begin{enumerate}[label = (\Alph*)]
 \item\label{item:main-tool-1-formal} $|N_{G_i}(Z)| < h$ for every connected component $Z$ of $G_i - D_i$ where $D_i \coloneqq \cl_t^{(G_i,\chi_i)}(X_i)$, and
 \item\label{item:main-tool-2-formal} $X_i \subseteq \cl_t^{(G_i,\chi_i)}(v)$ for every $v \in X_i$.
\end{enumerate}
The second property implies that $D_i = \cl_t^{(G_i,\chi_i)}(v)$ for every $v \in X_i$.
This allows us to test isomorphisms between $G_1[D_1]$ and $G_2[D_2]$ using the group-theoretic graph isomorphism machinery after individualizing a single vertex.
Also, as already discussed above, the first property allows to incorporate the partial isomorphisms between connected components of $G_1 - D_1$ and $G_2 - D_2$ using tools from \cite{GroheNSW20, Wiebking20}.

Now, Property \ref{item:main-tool-1-formal} has already been proved in \cite{GroheNW23} for graphs excluding $K_h$ as a topological subgraph.

\begin{theorem}[{\cite[Theorem 4.1]{GroheNW23}}]
 \label{thm:small-separator-for-t-cr-bounded-closure}
 Let $G$ be a graph that excludes $K_h$ as a topological subgraph and let $X \subseteq V(G)$.
 Let $t \geq 3h^3$ and define $D \coloneqq \cl_t^{G}(X)$.
 Let $Z$ be the vertex set of a connected component of $G - D$.
 Then $|N_G(Z)| < h$.
\end{theorem}

Hence, the remaining task is to find suitable isomorphism-invariant initial sets $X_1$ and $X_2$ that satisfy Property \ref{item:main-tool-2-formal}.
This is the main technical contribution of this paper.
We show that, after applying the $3$-dimensional Weisfeiler-Leman algorithm, there is some color $c_0$ such that, setting $\chi_i(v,w) \coloneqq \WL{3}{G_i}(v,w,w)$ for all $v,w \in V(G_i)$ and $X_i \coloneqq \{v \in V(G_i) \mid \WL{3}{G_i}(v,v,v) = c_0\}$, Property \ref{item:main-tool-2-formal} is satisfied.
Observe that $X_i$ and $\chi_i$ are clearly defined in an isomorphism-invariant manner since the coloring computed the $3$-dimensional Weisfeiler-Leman algorithm is preserved by isomorphisms.
This completes the high-level description of the algorithm.

In the following sections, these ideas are formalized.
In the next section, we first show that the $3$-dimensional Weisfeiler-Leman algorithm is able to provide suitable initial sets $X_1$ and $X_2$.
Afterwards, we review the group-theoretic graph isomorphism machinery in Section \ref{sec:group-machinery} and assemble the main algorithm in Section \ref{sec:algorithm-isomorphism}.

\section{Finding the Initial Set}
\label{sec:initial-color}

In this section, we argue how to compute the initial sets $X_1$ and $X_2$ with the desired properties as discussed in the last section.
Recall the definition of the constant $\adeg$ from Theorem \ref{thm:average-degree-excluded-topological}.
Without loss of generality assume $\adeg \geq 2$.
We define
\begin{equation}
 \label{eq:def-t}
 t(h) \coloneqq \max\{3h^3,\adeg h^2,\adeg^2h^4,12\adeg h^4,36\adeg^2h^5,144\adeg^2h^5\} = 144\adeg^2h^5.
\end{equation}
Here, the term $t(h)$ provides a lower bound on the parameter $t$ for the $t$-closure of a set that is required to achieve the desired properties.
While it is clear that the last term achieves the maximum, the other terms are also stated for later reference.
Indeed, each bound on $t$ will allow us to derive a specific property when building the $t$-closure of a set.
For example, $t \geq 3h^3$ allows the application of Theorem \ref{thm:small-separator-for-t-cr-bounded-closure}.
Combining all these properties eventually allows to prove that the computed initial sets $X_1$ and $X_2$ have the desired properties.
The next theorem forms the key technical contribution of this paper.

\begin{theorem}
 \label{thm:initial-color-via-wl}
 Let $G$ be a connected graph that excludes $K_h$ as a topological subgraph and let $t \geq t(h)$.
 Then there is a color $c_0 \in \{\WL{3}{G}(v,v,v) \mid v \in V(G)\}$ such that, for $\chi(v,w) \coloneqq \WL{3}{G}(v,w,w)$ for all $v,w \in V(G)$ and $X \coloneqq \{v \in V(G) \mid \WL{3}{G}(v,v,v) = c_0\}$, it holds that
 \[X \subseteq \cl_t^{(G,\chi)}(v)\]
 for all $v \in X$.
\end{theorem}

Before diving into the proof of Theorem \ref{thm:initial-color-via-wl}, let us state the main corollary that is used for our isomorphism test for graphs excluding $K_h$ as a topological subgraph.

\begin{corollary}
 \label{cor:initial-color}
 There is a polynomial-time algorithm that, given a connected vertex- and arc-colored graph $G$ and a number $t \geq t(h)$, either concludes that $G$ has a topological subgraph isomorphic to $K_h$
 or computes a pair-coloring $\chi\colon (V(G))^2 \rightarrow C$ and a non-empty set $X \subseteq V(G)$ such that
 \begin{enumerate}
  \item\label{item:initial-color-1} $X = \{v \in V(G) \mid \chi(v,v) = c\}$ for some color $c \in C$, and
  \item\label{item:initial-color-2} $X \subseteq \cl_t^{(G,\chi)}(v)$ for every $v \in X$.
 \end{enumerate}
 Moreover, the output of the algorithm is isomorphism-invariant with respect to $G$ and $t$.
\end{corollary}

Here, the output is isomorphism-invariant if it depends only on the isomorphism type of $G$ and the number $t$.
More formally, let $(G_1,t_1)$ and $(G_2,t_2)$ be two input pairs such that $G_1 \cong G_2$ and $t_1 = t_2$.
Then the algorithm either concludes in both cases that $G_i$ contains a topological subgraph isomorphic to $K_h$,
or it computes colorings $\chi_i\colon (V(G_i))^2 \rightarrow C$ and sets $X_i \subseteq V(G_i)$, $i \in \{1,2\}$, such that $\chi_1(v,w) = \chi_2(\varphi(v),\varphi(w))$ and $X_1 = X_2^\varphi$ for all $v,w \in V(G_1)$ and all $\varphi \in \Iso(G_1,G_2)$.

\begin{proof}
 The algorithms sets $\chi(v,w) \coloneqq \WL{3}{G}(v,w,w)$ for all $v,w \in V(G)$.
 Also, it computes the set $C_0$ of all colors $c_0 \in \{\WL{3}{G}(v,v,v) \mid v \in V(G)\}$ such that, for $X_{c_0} \coloneqq \{v \in V(G) \mid \WL{3}{G}(v,v,v) = c_0\}$, it holds that
 \[X_{c_0} \subseteq \cl_t^{(G,\chi)}(v)\]
 for all $v \in X_{c_0}$.
 If $C_0 = \emptyset$ the algorithm outputs that $G$ contains a topological subgraph isomorphic to $K_h$.
 This is correct by Theorem \ref{thm:initial-color-via-wl}.
 
 Otherwise, let $c_0$ be the minimal color contained in $C_0$ (recall that we always assume colors to be linearly ordered).
 The algorithm outputs $\chi$ and $X$ where $X \coloneqq X_{c_0}$.
 Clearly, both requirements are satisfied by definition.
 
 Also, the algorithm runs in polynomial time since $\WL{3}{G}$ and $\cl_t^{(G,\chi)}(v)$ for every $v \in V(G)$ can be computed in polynomial time.
\end{proof}

Now, let us return to Theorem \ref{thm:initial-color-via-wl}.
Its lengthy and technical proof covers the remainder of this section.
For better readability it is split into several steps.
Let us start by giving a brief outline.
Further intuition is provided throughout the proof when the single steps can be formulated more clearly.

The central tool for the analysis of the closure sets is the \emph{$t$-closure graph of $(G,\chi)$} which is defined as the directed graph $H$ with vertex set $V(H) \coloneqq V(G)$ and edge set
\[E(H) \coloneqq \{(v,w) \mid w \in \cl_t^{(G,\chi)}(v)\}.\]
A key property of the $3$-dimensional Weisfeiler-Leman algorithm is that it detects the edge relation of $H$, i.e., there is some set of colors $C_H \subseteq \{\chi(v,w) \mid v \neq w \in V(H)\}$ such that $(v,w) \in E(H)$ if and only if $\chi(v,w) \in C_H$ for all $v,w \in V(G)$.
Actually, this is the only part of the proof that requires us to use the $3$-dimensional Weisfeiler-Leman algorithm.
For all remaining parts, it turns out to be sufficient to use the $2$-dimensional Weisfeiler-Leman algorithm.

We define $c_0 \coloneqq \WL{3}{G}(v_0,v_0,v_0)$ for some $v_0 \in V(G)$ that appears in a maximal strongly connected component of $H$ (a strongly connected component of $H$ is \emph{maximal} if it has no outgoing edges).
Also, we set $X \coloneqq \{v \in V(G) \mid \WL{3}{G}(v,v,v) = c_0\}$ (as in the statement of Theorem \ref{thm:initial-color-via-wl}).
Let $D(v) \coloneqq \cl_t^{(G,\chi)}(v)$ for all $v \in X$.
In order to show that $X \subseteq \cl_t^{(G,\chi)}(v)$ it suffices to show that $D(v) = D(w)$ for all $v,w \in X$.
Building on the fact that the $3$-dimensional Weisfeiler-Leman algorithm detects the edge relation of $H$, we first show that $D(v) = D(w)$ or $D(v) \cap D(w) = \emptyset$ for all $v,w \in X$.
This allows us to partition the set $D \coloneqq \bigcup_{v \in X} D(v)$ and into classes $D_1,\dots,D_k$ where $\{D_1,\dots,D_k\} = \{D(v) \mid v \in X\}$, i.e., $D_1,\dots,D_k$ is an arbitrary enumeration of all distinct sets $D(v)$, $v \in X$.
Assuming $k \geq 2$, the basic strategy is to construct a large number of internally vertex-disjoint paths connecting vertices of different partition classes (see Lemma \ref{la:find-many-disjoint-paths}).
The construction of these paths (see Section \ref{subsec:finding-paths}) turns out to the most technical and complicated part of the proof.
Given a large number of such paths then results in a topological subgraph of $G$ that has high edge density which eventually contradicts Theorem \ref{thm:average-degree-excluded-topological}.

The remainder of this section is structured as follows.
In Section \ref{subsec:wl-tools} we provide additional notation and basic tools for the proof of Theorem \ref{thm:initial-color-via-wl}.
The closure graph and its basic properties are covered in Section \ref{subsec:compute-initial-color}.
In Section \ref{subsec:interaction-closure-sets}, we investigate the interaction between the closure sets $D(v) \coloneqq \cl_t^{(G,\chi)}(v)$, $v \in X$, in the graph $G$.
In particular, we reduce the task of proving Theorem \ref{thm:initial-color-via-wl} to proving the existence of certain vertex-disjoint paths, as formulated in Lemma \ref{la:find-many-disjoint-paths}.
To construct the paths, they are split into a ``middle part'' and two ``end parts''.
The main challenge turns out to be the construction of the middle part.
We further split this task depending on the parity of the length $p$ of the paths we are aiming to construct (all constructed paths shall have the same length).
We first cover the case in which $p$ is odd in Section \ref{subsec:paths-odd}.
While the basic strategy for both cases is similar, the details are much simpler for $p$ odd.
The more complicated case in which $p$ is even is then covered Section \ref{subsec:paths-even}.

\subsection{Basic Tools}
\label{subsec:wl-tools}

Among other things, the proof builds on various properties of the $2$-dimensional Weisfeiler-Leman algorithm.
Towards this end, we first introduce additional notation as well as some basic tools.

Let $G$ be a graph and let $\chi$ be a pair coloring that is $2$-stable on $G$.
We refer to $C_V \coloneqq C_V(G,\chi)\coloneqq\{\chi(v,v) \mid v \in V(G)\}$ as the set of \emph{vertex colors} and
$C_E\coloneqq C_E(G,\chi) \coloneqq \{\chi(v,w) \mid vw \in E(G)\}$ as the set of \emph{edge colors}.
For a vertex color $c\in C_V(G,\chi)$, we define $V_c\coloneqq V_c(G,\chi)\coloneqq\{v\in V(G)\mid \chi(v,v)=c\}$ as the set of all vertices with color $c$.
Similar, for an edge color $c\in C_E(G,\chi)$ we define $E_c\coloneqq E_c(G,\chi)\coloneqq\{v_1v_2\in E(G) \mid \chi(v_1,v_2) = c\}$.
We say a set $U \subseteq V(G)$ is \emph{$\chi$-invariant} if there is a set of vertex colors $C_U \subseteq C_V$ such that $U = \bigcup_{c \in C_U} V_c$.

Let $C \subseteq \{\chi(v,w) \mid v,w \in V(G), v \neq w\}$ be a set of colors.
We define the graph $G[C]$ with edge set
\[E(G[C]) \coloneqq \{vw \mid \chi(v,w) \in C\}\]
and vertex set
\[V(G[C]) \coloneqq \bigcup_{vw \in E(G[C])}\{v,w\}.\]

Let $A_1,\dots,A_\ell$ be the vertex sets of the connected components of $G[C]$.
We also define the graph $G/C$ as the graph obtained from contracting every set $A_i$ to a single vertex.
Formally,
\[V(G/C) \coloneqq \{\{v\} \mid v \in V(G) \setminus V(G[C])\} \cup \{A_1,\dots,A_\ell\}\]
and edge set
\[E(G/C) \coloneqq \{X_1X_2 \mid \exists v_1 \in X_1,v_2 \in X_2\colon v_1v_2 \in E(G)\}.\]

\begin{lemma}[see {\cite[Theorem 3.1.11]{ChenP19}}]
 \label{la:factor-graph-2-wl}
 Let $G$ be a graph and $C \subseteq \{\chi(v,w) \mid v,w \in V(G), v \neq w\}$ be a set of colors.
 Define
 \[(\chi/C)(X_1,X_2) \coloneqq \{\!\!\{\chi(v_1,v_2) \mid v_1 \in X_1, v_2 \in X_2\}\!\!\}\]
 for all $X_1,X_2 \in V(G/C)$.
 Then $\chi/C$ is $2$-stable on $G/C$.
 
 Moreover, for all $X_1,X_2,X_1',X_2' \in V(G/C)$, either $(\chi/C)(X_1,X_2) = (\chi/C)(X_1',X_2')$ or $(\chi/C)(X_1,X_2) \cap (\chi/C)(X_1',X_2') = \emptyset$.
\end{lemma}

For every edge color $c$, the endvertices of all $c$-colored edges have the same vertex colors, that is,
for all edges $vw,v'w'\in E(G)$ with $\chi(v,w) = \chi(v',w') = c$ we have $\chi(v,v) = \chi(v',v')$ and $\chi(w,w) = \chi(w',w')$.
This implies $1 \leq |C_V(G[\{c\}],\chi)| \leq 2$.
We say that $G[c] \coloneqq G[\{c\}]$ is \emph{unicolored} if $|C_V(G[c],\chi)| = 1$.
Otherwise $G[c]$ is called \emph{bicolored}. 
The next two lemmas investigate properties of connected components of bicolored graphs $G[c]$ for an edge color $c$.
Again, recall the definition of the constant $\adeg$ from Theorem \ref{thm:average-degree-excluded-topological}.

\begin{lemma}
 \label{la:hat-graph-excluded-topological}
 Let $G = (V_1,V_2,E)$ be a connected, bipartite graph that excludes $K_h$ as a topological subgraph and let $\chi$ be a pair-coloring that is $2$-stable on $G$.
 Suppose that $\chi(v_1,v_2) = \chi(v_1',v_2')$ for all $(v_1,v_2),(v_1',v_2') \in V_1 \times V_2$ with $v_1v_2,v_1'v_2'\in E$.
 Also assume that $|V_2| > \adeg h^2|V_1|$.
 Let 
 \[E^{*} \coloneqq \left\{v_1v_2 \in \binom{V_1}{2} \mid \exists w \in V_2 \colon v_1w,v_2w \in E(G)\right\}.\]
 Then there are colors $c_1,\dots,c_r \in \chi(V_1^{2})$ such that
 \begin{enumerate}
  \item $E^{*} = \bigcup_{i \in [r]} E_{c_i}$ where $E_{c_i} \coloneqq \{v_1v_2\in V(G)^2 \mid \chi(v_1,v_2) = c_i\}$,
  \item $H \coloneqq (V_1,E^{*})$ is connected, and
  \item $H_i$ is a topological subgraph of $G$ for all $i \in [r]$ where $H_i = (V_1,E_{c_i})$.
 \end{enumerate}
\end{lemma}

\begin{proof}
 Clearly, $H$ is connected (since $G$ is connected) and there are colors $c_1,\dots,c_r \in \chi(V_1^{2})$ such that $E^{*} = \bigcup_{i \in [r]} E_{c_i}$.
 
 So let $i \in [r]$ and consider the bipartite graph $B = (V_2,E_{c_i},E(B))$ where $E(B) \coloneqq \{(u,v_1v_2) \mid u \in N_G(v_1) \cap N_G(v_2)\}$.
 By the properties of the $2$-dimensional Weisfeiler-Leman algorithm
 the graph $B$ is biregular. It follows from Hall's Marriage Theorem
 that $B$ contains a matching $M$
 of size $\min(|V_2|,|E_{c_i}|)$ as explained in the preliminaries.
 If $|V_2| \geq |E_{c_i}|$ then $H_i$ is a topological subgraph of $G$.
 
 So suppose that $|V_2| < |E_{c_i}|$.
 Let $F_i \subseteq E_{c_i}$ be those vertices that are matched by the matching $M$ in the graph $B$.
 Then $H_i' \coloneqq (V_1,F_i)$ is a topological subgraph of $G$, and thus it excludes $K_h$ as a topological subgraph.
 However, $|F_i| = |V_2| > \adeg h^2 |V_1|$ which contradicts Theorem \ref{thm:average-degree-excluded-topological}.
\end{proof}

\begin{lemma}
 \label{la:closure-bipartite}
 Let $t \geq \adeg^2h^4$.
 Let $G = (V_1,V_2,E)$ be a connected bipartite graph that excludes $K_h$ as a topological subgraph and let $\chi$ be a pair-coloring that is $2$-stable on $G$.
 Suppose that $\chi(v_1,v_2) = \chi(v_1',v_2')$ for all $(v_1,v_2),(v_1',v_2') \in V_1 \times V_2$ with $v_1v_2,v_1'v_2'\in E$.
 Also assume that $|V_1| \leq |V_2|$.
 Then $V_1 \subseteq \cl_t^{(G,\chi)}(v)$ for all $v \in V_1 \cup V_2$.
\end{lemma}

\begin{proof}
 The graph $G$ is biregular and it holds that $\deg(v_1)\cdot|V_1| = \deg(v_2)\cdot|V_2|$ for all $v_1 \in V_1$ and $v_2 \in V_2$.
 Hence, $\deg(v_2) \leq \adeg h^2$ for all $v_2 \in V_2$ by Theorem \ref{thm:average-degree-excluded-topological}.
 This means $\cl_t^{(G,\chi)}(v) \cap V_1 \neq \emptyset$, because either $v \in \cl_t^{(G,\chi)}(v)\cap V_1$ or $v\in V_2$ and $N_G(v)\subseteq \cl_t^{(G,\chi)}(v) \cap V_1$.

 First suppose that $|V_2| \leq \adeg h^2 |V_1|$.
 Then $\deg(v_1)=\deg(v_2)\frac{|V_2|}{|V_1|} \leq t$ and $\deg(v_2)\leq t$ for all $v_1 \in V_1,v_2\in V_2$.
 It follows that $\cl_t^{(G,\chi)}(v)=V(G)$.
 
 So assume that $|V_2| > \adeg h^2 |V_1|$.
 By Lemma \ref{la:hat-graph-excluded-topological}, there are colors $c_1,\dots,c_r \in \chi(V_1^{2})$ such that
 \begin{enumerate}
  \item $H_i$ excludes $K_h$ as a topological subgraph for all $i \in [r]$ where $H_i = (V_1,E_{c_i})$ and $E_{c_i}\coloneqq\{v_1v_2\in V(G)^2\mid\chi(v_1,v_2)=c_i\}$, and
  \item $H = (V_1,\bigcup_{i \in [r]} E_{c_i})$ is connected.
 \end{enumerate}
 For all $i\in[r]$, the graph $H_i$ is $d$-regular for some $d$,
 and by Theorem \ref{thm:average-degree-excluded-topological} it holds that $d \leq \adeg h^2 \leq t$.
 This implies that $V_1 \subseteq \cl_t^{(G,\chi)}(v_1)$ for all $v_1 \in V_1$, and since $V_1 \cap
 \cl_t^{(G,\chi)}(v)\neq\emptyset$ for all $v\in V_1\cup V_2$, it follows that $V_1 \subseteq  \cl_t^{(G,\chi)}(v)$.
\end{proof}

\subsection{The Closure Graph}
\label{subsec:compute-initial-color}

A key tool in the proof of Theorem \ref{thm:initial-color-via-wl} is the closure graph.
Let $G$ be a graph and let $\chi$ be the pair-coloring defined via $\chi(v,w) \coloneqq \WL{3}{G}(v,w,w)$ for all $v,w \in V(G)$.
For $t \geq 1$, we define the \emph{$t$-closure graph of $(G,\chi)$} to be the directed graph $H$ defined via $V(H) \coloneqq V(G)$ and
\[E(H) \coloneqq \{(v,w) \mid w \in \cl_t^{(G,\chi)}(v), v \neq w\}.\]

\begin{theorem}
 \label{thm:chi-is-stable-on-h}
 Let $t \geq 1$.
 Let $G$ be a graph and let $\chi$ be the pair-coloring defined via $\chi(v,w) \coloneqq \WL{3}{G}(v,w,w)$ for all $v,w \in V(G)$.
 Also let $H$ be the $t$-closure graph of $(G,\chi)$.
 Then the coloring $\chi$ is $2$-stable on $H$.
\end{theorem}

\begin{proof}
 Let $\lambda$ be a pair-coloring that is $2$-stable on $G$ such that $\lambda \preceq \chi$.
 We say a partition $\CP$ of the set of vertices of $G$ is \emph{$\lambda$-definable} if there is a set $C_{\CP} \subseteq \{\lambda(v,w) \mid v \neq w \in V(G)\}$ such that $\CP$ is the partition into connected components of $G[C_\CP]$.
 To prove the lemma we argue that all partitions into color classes of colorings computed by the $t$-CR algorithm are $\lambda$-definable.
 \begin{claim}
  Let $\CP$ be a $\lambda$-definable partition of the vertex set of $G$.
  Also define
  \[\CP' \coloneqq \{P \in \CP \mid |P| > t\} \cup \{\{v\} \mid v \in P \text{ for some } P \in \CP \text{ with } |P| \leq t\}.\]
  Then $\CP'$ is $\lambda$-definable.
 \end{claim}
 \begin{claimproof}
  Let $P_1,P_2 \in \CP$ such that $|P_1| \neq |P_2|$.
  Also let $v_1 \in P_1$ and $v_2 \in P_2$.
  Since the $2$-dimensional Weisfeiler-Leman algorithm detects which vertices are reachable from $v_1$ and $v_2$ in the graph $G[C_\CP]$ (see, e.g., \cite[Theorem 2.6.7]{ChenP19}), it follows that $\lambda(v_1,v_1) \neq \lambda(v_2,v_2)$.
  Hence, defining
  \[C_{\CP'} \coloneqq C_{\CP} \setminus \{\lambda(v,w) \mid \exists P \in \CP\colon |P| \leq t \wedge v,w \in P\},\]
  it follows that $\CP'$ is $\lambda$-definable.
 \end{claimproof}

 \begin{claim}
  Let $\CP$ be a $\lambda$-definable partition of the vertex set of $G$.
  Also define $\CP' \preceq \CP$ to be the coarsest partition which is stable with respect to the Color Refinement algorithm. 
  Then $\CP'$ is $\lambda$-definable.
 \end{claim}
 \begin{claimproof}
  For $v \in V(G)$, define $P_v$ to be the unique set $P \in \CP$ such that $v \in P$.
  Define $v \sim_\CP w$ if
  \[\Big\{\!\!\Big\{\big(P_u,\chi(v,u),\chi(u,v)\big) \;\Big\vert\; u \in V(G) \Big\}\!\!\Big\} = \Big\{\!\!\Big\{\big(P_u,\chi(w,u),\chi(u,w)\big) \;\Big\vert\; u \in V(G) \Big\}\!\!\Big\}.\]
  Let $\CP''$ be the partition into equivalence classes of $\sim_\CP$.
  Then $\CP''$ is $\lambda$-definable using Lemma \ref{la:factor-graph-2-wl}.
  So in other words, applying a single round of the Color Refinement algorithm does not effect $\lambda$-definability.
  Hence, the claim follows.
 \end{claimproof}
 
 Now, let $u \in V(G)$ and define $\lambda_u(v,w) \coloneqq \WL{3}{G}(u,v,w)$ for all $v,w \in V(G)$.
 Also, let $\chi_u(u,u) \coloneqq (1,1)$ and $\chi_u(v,w) \coloneqq (0,\chi(v,w))$ for all $(u,u) \neq (v,w) \in (V(G))^2$.
 Then $\lambda_u$ is $2$-stable on $(G,\chi_u)$ (i.e., the graph obtained from $(G,\chi)$ by individualizing $u$) by Fact \ref{fact:k-wl-is-stable-after-individualization}.
 Now, the previous two claims imply that the partition $\CP$ into color classes of $\tColRef{t}{G,\chi_u}$ is $\lambda_u$-definable.
 Moreover, the choice of $C_\CP$ does not depend on $u$.
 It follows that there is some set of colors $C^* \subseteq \{\WL{3}{G}(u,v,w) \mid u,v,w \in V(G)\}$ such that
 \[\WL{3}{G}(u,v,w) \in C^* \;\;\Leftrightarrow\;\; \tColRef{t}{G,\chi_u}(v) = \tColRef{t}{G,\chi_u}(w).\]
 This implies the statement of the lemma since $(u,v) \in E(H)$ if and only if there is no $w \neq v$ such that $\WL{3}{G}(u,v,w) \in C^*$.
\end{proof}

We now turn to the proof of Theorem \ref{thm:initial-color-via-wl}.
For the remainder of this section, let us fix a connected graph $G$ and a number $t \geq t(h)$ as the input for Theorem \ref{thm:initial-color-via-wl}.
Also, fix the pair-coloring $\chi$ defined via $\chi(v,w) \coloneqq \WL{3}{G}(v,w,w)$ for all $v,w \in V(G)$.
Observe that $\chi$ is $2$-stable on $G$ by Fact \ref{fact:k-wl-is-stable}.
Moreover, let $H$ be the $t$-closure graph of $(G,\chi)$.
We have that $\chi$ is also $2$-stable on $H$ by Theorem \ref{thm:chi-is-stable-on-h}.

\begin{observation}
 \label{obs:h-is-transitive}
 Let $(u,v),(v,w) \in E(H)$ such that $u \neq w$. Then $(u,w) \in E(H)$.
\end{observation}

\begin{proof}
 By definition of the closure graph, $v \in \cl_t^{(G,\chi)}(u)$ and $w \in \cl_t^{(G,\chi)}(v)$.
 So $w \in \cl_t^{(G,\chi)}(u)$ which implies $(u,w) \in E(H)$.
\end{proof}

Now let $A$ be a strongly connected component of $H$.
Then $(v,w) \in E(H)$ for all distinct $v,w \in A$ by Observation \ref{obs:h-is-transitive}.
We say that a vertex $v \in V(G)$ is \emph{maximal} if it appears in a strongly connected component of $H$ without outgoing edges, i.e., $(u,v) \in E(H)$ for all $u \in V(G)$ such that $(v,u) \in E(H)$.

\begin{corollary}
 The set of maximal vertices is $\chi$-invariant.
\end{corollary}

\begin{proof}
 This follows from the fact that $\chi$ is $2$-stable on $H$.
\end{proof}

We say that a vertex color $c \in C_V \coloneqq C_V(G,\chi)$ is \emph{maximal} if there is some vertex $v \in V_c$ which is maximal.
By the corollary, for a maximal color $c \in C_V$, every vertex $v \in V_c$ is maximal.
We fix $c_0 \in C_V$ to be a maximal color and define
\begin{equation}
 \label{eq:def-x}
 X \coloneqq V_{c_0}.
\end{equation}
Also, for $v \in X$, let
\begin{equation}
 \label{eq:def-d-v}
 D(v) \coloneqq \cl_t^{(G,\chi)}(v)
\end{equation}
denote the $t$-closure of $v$.

\begin{observation}
 \label{obs:equal-or-disjoint}
 Let $v,w \in X$. Then $D(v) = D(w)$ or $D(v) \cap D(w) = \emptyset$.
\end{observation}

\begin{proof}
 Since $v$ is maximal it holds that $D(v)$ is precisely the strongly connected component of $H$ that contains $v$.
 Two such components are either equal or disjoint.
\end{proof}

Then next lemma forms the main step in proving Theorem \ref{thm:initial-color-via-wl}.

\begin{lemma}
 \label{la:initial-color-main}
 Suppose there are $v,w \in X$ such that $D(v) \cap D(w) = \emptyset$.
 Then $G$ has a topological subgraph isomorphic to $K_h$.
\end{lemma}

\begin{remark}
 \label{rem:pair-coloring-requirements}
 We remark that the proof of Lemma \ref{la:initial-color-main} only exploits that the coloring $\chi$ is $2$-stable on the graph $G$ and on the $t$-closure graph of $(G,\chi)$.
 In other words, Theorem \ref{thm:chi-is-stable-on-h} is the only part of the proof that actually requires the $3$-dimensional Weisfeiler-Leman algorithm.
\end{remark}

Before diving into the proof of the lemma, let us first provide a proof for Theorem \ref{thm:initial-color-via-wl} assuming Lemma \ref{la:initial-color-main} holds true.

\begin{proof}[Proof of Theorem \ref{thm:initial-color-via-wl}]
 Let $H$ denote the $t$-closure graph of $(G,\chi)$ and pick $c_0 \in C_V$ to be a maximal color.
 Also, let $X \coloneqq V_{c_0}$ and define $D(v) \coloneqq \cl_t^{(G,\chi)}(v)$ for all $v \in X$.
 
 If there are $v,w \in X$ such that $D(v) \neq D(w)$ then $G$ has a topological subgraph isomorphic to $K_h$ by Observation \ref{obs:equal-or-disjoint} and Lemma \ref{la:initial-color-main}.
 So $D(v) = D(w)$ for all $v,w \in X$.
 Since $v \in D(v)$ for all $v \in X$ this implies that $X \subseteq D(v)$ for all $v \in X$.
\end{proof}

Now, let us turn to the proof of Lemma \ref{la:initial-color-main} which covers the rest of this section.
Assume there are $v,w \in X$ such that $D(v) \cap D(w) = \emptyset$.
We define $D \coloneqq \bigcup_{v \in X} D(v)$ and let $k \coloneqq |\{D(v) \mid v \in X\}|$.
Also let $\{D_1,\dots,D_k\} = \{D(v) \mid v \in X\}$, i.e., $D_1,\dots,D_k$ is an arbitrary enumeration of all distinct sets $D(v)$, $v \in X$.

\begin{corollary}
 \label{cor:d-sets-invariant}
 The set $D$ is $\chi$-invariant.
 Moreover, there is a set of colors
 \begin{equation}
  C_\sim \subseteq \{\chi(v,w) \mid v,w \in V(G), v \neq w\}
 \end{equation}
 such that $D_1,\dots,D_k$ are precisely the connected components of $G[C_\sim]$.
 Also,
 \begin{equation}
  (\chi/C_\sim)(D_i,D_i) = (\chi/C_\sim)(D_j,D_j)
 \end{equation}
 for all $i,j \in [k]$.
\end{corollary}

\begin{proof}
 By definition, the set $X$ is $\chi$-invariant.
 Recall that $H$ denotes the $t$-closure graph of $(G,\chi)$.
 By Theorem \ref{thm:chi-is-stable-on-h} $\chi$ is $2$-stable on $H$.
 By definition of the closure graph $D = \{v \in V(H) \mid \exists w \in X \colon (w,v) \in E(H)\}$.
 It follows that $D$ is $\chi$-invariant.
 
 Moreover, the sets $D_i$, $i \in [k]$, are precisely the strongly connected components of $G[D]$.
 Hence, there is a set $C_\sim \subseteq \{\chi(v,w) \mid v,w \in V(G), v \neq w\}$ such that $D_1,\dots,D_k$ are precisely the connected components of $G[C_\sim]$.
 
 Finally, by Lemma \ref{la:factor-graph-2-wl}, either $(\chi/C_\sim)(D_i,D_i) = (\chi/C_\sim)(D_j,D_j)$ or $(\chi/C_\sim)(D_i,D_i) \cap (\chi/C_\sim)(D_j,D_j) = \emptyset$ for all $i,j \in [k]$.
 Since $X \cap D_i \neq \emptyset$ for all $i \in [k]$ and $X = V_{c_0}$ by definition, it follows that $c_0 \in (\chi/C_\sim)(D_i,D_i)$ for all $i \in [k]$.
 So $(\chi/C_\sim)(D_i,D_i) = (\chi/C_\sim)(D_j,D_j)$ for all $i,j \in [k]$.
\end{proof}

On a high level, the main target for the proof of Lemma \ref{la:initial-color-main} is to construct a topological subgraph of $G$ which violates the bound on the average degree from Theorem \ref{thm:average-degree-excluded-topological}.
Towards this end, we shall consider paths of minimum length connecting sets $D_i$ and $D_j$ for distinct $i,j \in [k]$.
We start by covering some simple cases using the tools from Section \ref{subsec:wl-tools}.
This simplifies the analysis later on since we can exclude certain corner cases.

\begin{lemma}
 \label{la:no-edges-between-d-sets}
 Let $v,w \in X$ such that $D(v) \cap D(w) = \emptyset$ and $E_G(D(v),D(w)) \neq \emptyset$.
 Then $G$ has a topological subgraph isomorphic to $K_h$.
\end{lemma}

\begin{proof}
 Suppose that $G$ has no topological subgraph isomorphic to $K_h$ and there are $v' \in D(v)$ and $w' \in D(w)$ such that $v'w' \in E(G)$.
 We argue that $D(v) \cap D(w) \neq \emptyset$.
 
 Let $c_E \coloneqq \chi(v',w')$ and consider graph $F \coloneqq G[c_E]$.
 Let $A$ be the vertex set of the connected component of $F$ such that $v',w' \in A$.
 If $F$ is unicolored then $G[A]$ is $d$-regular for some $d \leq \adeg h^2 \leq t$ by Theorem \ref{thm:average-degree-excluded-topological}.
 Hence, $A \subseteq D(v)$ and $D(v) \cap D(w) \neq \emptyset$.
 
 Otherwise $F$ is bicolored with color classes $V_1$ and $V_2$.
 Without loss of generality suppose that $|V_1| \leq |V_2|$ and $w' \in V_1$.
 Hence, $v' \in V_2$ which implies that $w' \in D(v)$ using Lemma \ref{la:closure-bipartite}.
\end{proof}

\begin{lemma}
 \label{la:no-common-neighbors-between-d-sets}
 Let $v,w \in X$ such that $D(v) \cap D(w) = \emptyset$ and $N_G(D(v)) \cap N_G(D(w)) \neq \emptyset$.
 Then $G$ has a topological subgraph isomorphic to $K_h$.
\end{lemma}

\begin{proof}
 Suppose that $G$ has no topological subgraph isomorphic to $K_h$ and there are $v' \in D(v)$, $w' \in D(w)$, and $u' \in V(G)$ such that $v'u',w'u' \in E(G)$.
 We argue that $D(v) \cap D(w) \neq \emptyset$.
 By Lemma \ref{la:no-edges-between-d-sets}, we may assume that $u' \notin D$.
 In particular $\chi(v',v') \neq \chi(u',u') \neq \chi(w',w')$ by Corollary \ref{cor:d-sets-invariant}.
 
 Let $c_1 \coloneqq \chi(v',u')$ and consider graph $F_1 \coloneqq G[c_1]$.
 Also, define $c_2 \coloneqq \chi(w',u')$ and $F_2 \coloneqq G[c_2]$.
 First suppose that $c_1 = c_2$.
 Then $\{w',u'\} \cap D(v) \neq \emptyset$ by Lemma \ref{la:closure-bipartite}.
 Since $u' \notin D$ we conclude that $w' \in D(v)$.
 
 So suppose that $c_1 \neq c_2$.
 Since $\chi(v',v') \neq \chi(u',u') \neq \chi(w',w')$ we conclude that $F_1$ and $F_2$ are bicolored.
 Let $V_1^i$ and $V_2^i$ be the color classes of $F_i$, $i \in \{1,2\}$, and suppose without loss of generality that $|V_1^i| \leq |V_2^i|$.
 Since $u' \notin D$, Lemma \ref{la:closure-bipartite} implies that $u' \in V_2^i$ for both $i \in \{1,2\}$ and thus, $V_2 \coloneqq V_2^1 = V_2^2$.
 Moreover, $v' \in V_1^1$ and $w' \in V_1^2$.
 
 Without loss of generality assume $|V_1^2| \leq |V_1^1|$.
 We argue that $w' \in D(v)$ which implies the lemma.
 Towards this end, consider the graph $F \coloneqq G[c]$ where $c \coloneqq \chi(v',w')$.
 \begin{claim}
  \label{cl:topological-subgraph-f}
  $F$ is a topological subgraph of $G$.
 \end{claim}
 \begin{claimproof}
  Consider the bipartite graph $B = (V_2,E(F),E(B))$ where
  \[E(B) \coloneqq \{ue \mid u \in V_2, e = xy \in E(F), \chi(x,u) = c_1, \chi(y,u) = c_2\}.\]
  By the properties of the $2$-dimensional Weisfeiler-Leman algorithm we conclude that $B$ is biregular.
  Hence, by Hall's Marriage Theorem, there is a matching $M \subseteq E(B)$ of $B$ of size $\min(|V_2|,|E(F)|)$.
  
  If $|V_2| \geq |E(F)|$ then $F$ is a topological subgraph of $G$ where an edge $e \in E(F)$ is realized by a path of length $2$ via $u$ for the unique $u$ such that $ue \in M$.
  
  So suppose that $|V_2| < |E(F)|$.
  Let $\widehat{F}$ be the subgraph of $F$ defined via $V(\widehat{F}) \coloneqq V(F)$ and
  \[E(\widehat{F}) \coloneqq \{e \in E(F) \mid \exists u \in V_2 \colon ue \in M\}.\]
  Then $\widehat{F}$ is topological subgraph of $G$ and hence, $2|E(\widehat{F})| \leq \adeg h^2|V(F)|$ by Theorem \ref{thm:average-degree-excluded-topological}.
  Since $|E(\widehat{F})| = |V_2|$ we conclude that
  \[|V_2| \leq \frac{1}{2} \adeg h^2|V(F)| \leq \frac{1}{2} \adeg h^2 (|V_1^1| + |V_1^2|) \leq \adeg h^2 |V_1^1|.\]
  Now consider the graph $F_1$ which is biregular and a topological subgraph of $G$.
  We have that
  \[|E(F_1)| \leq \frac{1}{2}\adeg h^2(|V_1^1| + |V_2|) \leq \frac{1}{2}\adeg h^2(|V_1^1| + \adeg h^2 |V_1^1|) \leq \adeg^2h^4|V_1^1|\]
  using again Theorem \ref{thm:average-degree-excluded-topological}.
  On the other hand,
  \[|E(F_1)| \geq |V_1^1| \cdot \deg_{F_1}(v')\]
  since  $F_1$ is biregular.
  Together, this means that $\deg_{F_1}(v') \leq \adeg^2h^4 \leq t$.
  So $N_{F_1}(v') \subseteq D(v)$.
  But this is a contradiction since $u' \in N_{F_1}(v')$ and $u' \notin D$.
  So this case does not occur.
 \end{claimproof}

 Let $A$ be the vertex set of the connected component of $F$ such that $v',w' \in A$.
 If $F$ is unicolored then $F[A]$ is $d$-regular for some $d \leq \adeg h^2 \leq t$ by Theorem \ref{thm:average-degree-excluded-topological}.
 Hence, $A \subseteq D(v)$.
 
 Otherwise $F$ is bicolored with color classes $V_1^1$ and $V_1^2$.
 By the assumptions, $|V_1^2| \leq |V_1^1|$, $w' \in V_1^2$ and $v' \in V_1^1$.
 More strongly, we actually get that $|V_1^2 \cap A| \leq |V_1^1 \cap A|$, $w' \in V_1^2 \cap A$ and $v' \in V_1^1 \cap A$.
 Also, the restriction $\chi|_{A^2}$ is $2$-stable on $F[A]$.
 So
 \[w' \in \cl_t^{(F[A],\chi|_{A^2})}(v') \subseteq \cl_t^{(G,\chi)}(v') = D(v') = D(v)\]
 using Lemma \ref{la:closure-bipartite} and Observation \ref{obs:equal-or-disjoint}.
\end{proof}

Hence, for the remainder of this section, we assume that the distance between distinct sets $D_i$ and $D_j$ is least three.
In order to cover the case of larger distances, we need to understand in more detail how the \emph{closure sets} $D_i$, $i \in [k]$, interact with one another.

\subsection{Interaction between Closure Sets}
\label{subsec:interaction-closure-sets}

In order to analyze the interaction between closure sets, it turns out to be more convenient to assume that $G$ has no topological subgraph isomorphic to $K_h$.
Hence, for the remainder of this section, we make the following assumptions and eventually derive a contradiction:
\begin{enumerate}[label=(A.\arabic*)]
 \item\label{item:assumption-1} $k \geq 2$, i.e., there are $v,w \in X$ such that $D(v) \cap D(w) = \emptyset$, and
 \item\label{item:assumption-2} $G$ has no topological subgraph isomorphic to $K_h$.
\end{enumerate}
The first goal is to argue that each set $D_i$ interacts with the other sets $D_j$ only via a small set of vertices.
To be more precise, we argue that there are sets $S_i \subseteq D_i$ of size $|S_i| < h$ such that every shortest path between $D_i$ and $D_j$ for distinct $i,j \in [k]$ starts in $S_i$ and ends in $S_j$.
The following auxiliary lemma turns out to be useful for this task.

\begin{lemma}
 \label{la:disjoint-paths-from-stable-coloring}
 Let $G$ be a graph and let $X_1,\dots,X_\ell \subseteq V(G)$ be pairwise disjoint sets such that $G[X_i,X_{i+1}]$ is a non-empty, biregular graph for all $i \in [\ell-1]$.
 Let $k = \min_{i \in [\ell]} |X_i|$.
 
 Then there exist $k$ vertex-disjoint paths from $X_1$ to $X_\ell$, i.e., there are distinct vertices $v_{i,j} \in X_i$ for all $i \in [\ell]$ and $j \in [k]$ such that $v_{i,j}v_{i+1,j} \in E(G)$ for all $i \in [\ell-1]$ and $j \in [k]$.
\end{lemma}

\begin{proof}
 Without loss of generality assume that $E(G) = \bigcup_{i \in [\ell-1]} E(G[X_i,X_{i+1}])$. Let $S$ be a $(X_1,X_\ell)$-separator.
 By Menger's Theorem it suffices to prove that $|S| \geq k$.
 
 For $i \in [\ell]$ let $k_i = |S \cap X_i|$. Also define $f_i = |X_i| \sum_{j \leq i} \frac{k_j}{|X_j|}$.
 We prove by induction on $i$ that
 \[|\{v_i \in X_i \setminus S \mid \text{there is a path from $X_1 \setminus S$ to $v_i$ in $G - S$ }\}| \geq |X_i| - f_i.\]
 The base step $i = 1$ is immediately clear from the definition.
 So let
 \[B_i = \{v_i \in X_i \setminus S \mid \text{there is a path from $X_1 \setminus S$ to $v_i$ in $G - S$ }\}\] and suppose $|B_i| \geq |X_i| - f_i$.
 
 We first argue that $|N(B_i) \cap X_{i+1}| \geq \frac{B_i}{|X_i|} \cdot |X_{i+1}|$.
 Since $G[X_i,X_{i+1}]$ is biregular there exist $d \in \mathbb{N}$ such that $|N(v) \cap X_{i+1}| = d$ for all $v \in X_i$.
 Moreover, $|N(w) \cap X_{i}| = d \cdot \frac{|X_{i}|}{|X_{i+1}|}$ for all $w \in X_{i+1}$ by a simple counting argument.
 Let $E(B_i,X_{i+1}) = \{vw \in E(G) \mid v \in B_i,w \in X_{i+1}\}$.
 Then $|E(B_i,X_{i+1})| = d \cdot |B_i|$ and $|E(B_i,X_{i+1})| \leq |N(B_i) \cap X_{i+1}| \cdot d \cdot \frac{|X_{i}|}{|X_{i+1}|}$.
 In combination, this gives the desired bound $|N(B_i) \cap X_{i+1}| \geq \frac{B_i}{|X_i|} \cdot |X_{i+1}|$.
 
 Now $B_{i+1}$ contains exactly those vertices in the set $N(B_i) \cap X_{i+1}$ that are not contained in $S$. Hence,
 \begin{align*}
  |B_{i+1}| &\geq |N(B_i) \cap X_{i+1}| - k_{i+1} \\
            &\geq \frac{B_i}{|X_i|} \cdot |X_{i+1}| - k_{i+1} \\
            &\geq \left(1 - \sum_{j \leq i} \frac{k_j}{|X_j|}\right) \cdot |X_{i+1}| - k_{i+1} \\
            &= \left(1 - \sum_{j \leq i+1} \frac{k_j}{|X_j|}\right) \cdot |X_{i+1}| \\
            &= |X_{i+1}| - f_{i+1}.
 \end{align*}
 By the induction principle it follows that $|B_\ell| \geq |X_\ell| - f_\ell$.
 Since $S$ is an $(X_1,X_\ell)$-separator it follows that $B_\ell = \emptyset$ and thus, $f_\ell \geq |X_\ell|$.
 In other words, $\sum_{i \leq \ell} \frac{k_i}{|X_i|} \geq 1$.
 This implies that $|S| = \sum_{i \leq \ell} k_i \geq \min_{i \in [\ell]} |X_i| = k$.
\end{proof}

\begin{lemma}
 \label{la:d-sets-do-not-split}
 Let $v,w \in X$ such that $D(v) \cap D(w) = \emptyset$.
 Then there is a connected component $Z$ of $G - D(v)$ such that $D(w) \subseteq Z$.
\end{lemma}

\begin{proof}
 Suppose towards a contraction that there are distinct components of $Z_1,Z_2$ of the graph $G - D(v)$ such that $D(w) \cap Z_i \neq \emptyset$ for both $i \in \{1,2\}$.
 Pick vertices $w_i \in D(w) \cap Z_i$ for both $i \in \{1,2\}$.
 Since $G$ is connected there is a shortest path $w_1 = u_0,u_1,\dots,u_m,u_{m+1} = w_2$ from $w_1$ to $w_2$.
 Without loss of generality assume that $u_1,\dots,u_m \notin D(w)$.
 
 Let $\lambda \coloneqq \tColRef{t}{G,\chi,w}$ be the $t$-CR-stable coloring after individualizing $w$ and let $X_i \coloneqq [u_i]_{\lambda}$ be the color class of $u_i$ with respect to the coloring $\lambda$.
 Note that $\lambda(u_i) \neq \lambda(u_j)$ for all distinct $i,j \in [m]$, because $\dist_G(u_i,w_1) \neq \dist_G(u_j,w_1)$ and $|[w_1]_{\lambda}| = 1$.
 So $X_i \cap X_j = \emptyset$ for all distinct $i,j \in [m]$.
 Also $|X_i| \geq t$ for all $i \in [m]$.
 So there are $t$ internally vertex-disjoint paths from $w_1$ to $w_2$ by Lemma \ref{la:disjoint-paths-from-stable-coloring}.
 
 On the other hand, $|N_G(Z_1)| < h \leq t$ by Theorem \ref{thm:small-separator-for-t-cr-bounded-closure}.
 Since $w_1 \in Z_1$ and $w_2 \notin Z_1$ this is a contradiction.
\end{proof}

The lemma builds a main step for understanding the interaction between closure sets.
Recall that we currently aim to prove that there are sets $S_i \subseteq D_i$ of size $|S_i| < h$ such that every shortest path between $D_i$ and $D_j$ for distinct $i,j \in [k]$ starts in $S_i$ and ends in $S_j$.
If $G - D_i$ is connected for all $i \in [k]$ this statement follows from directly from Theorem \ref{thm:small-separator-for-t-cr-bounded-closure} setting $S_i \coloneqq N_G(Z_i)$ where $Z_i$ is the unique connected component of $G - D_i$.
So suppose there is some $i \in [k]$ such that $G - D_i$ is not connected, i.e., the set $D_i$ forms a separator.
Now, Lemma \ref{la:d-sets-do-not-split} implies that these separators do not ``cross''.
In particular, we may contract all sets $D_i$ to a single vertex without effectively changing the connected components of $G - D_i$.
This way, we can strengthen the last lemma and prove that, indeed, all sets $D_j$, $i \neq j \in [k]$, appear in the same connected component of $G - D_i$.
Here, we exploit the fact that the $2$-dimensional Weisfeiler-Leman algorithm detects cut vertices (i.e., $1$-separators) as well as the structure of the block-cut tree (see \cite{KieferN22,KieferPS19}).

\begin{lemma}[{\cite[Corollary 7]{KieferPS19}}]
 \label{la:2-wl-cut-vertices}
 Let $G_1$ be a graph and suppose $v_1 \in V(G_1)$ is a cut vertex of $G_1$.
 Also let $G_2$ be a second graph and let $v_2 \in V(G_2)$ such that $\WL{2}{G_1}(v_1,v_1) = \WL{2}{G_2}(v_2,v_2)$.
 Then $v_2$ is a cut vertex of $G_2$.
\end{lemma}

\begin{lemma}
 \label{la:root-set}
 For each $v \in X$ there is a connected component $Z(v)$ of the graph $G - D(v)$ such that
 \begin{enumerate}
  \item $D(w) \subseteq Z(v)$, or
  \item $D(w) = D(v)$
 \end{enumerate}
 for all $w \in X$. Moreover, the set
 \begin{equation}
  R \coloneqq \bigcap_{v \in X} Z(v).
 \end{equation}
 is $\chi$-invariant.
\end{lemma}

\begin{proof}
 We first argue that there is some $i \in [k]$ and a connected component $Z_i$ of $G - D_i$ such that $D_j \subseteq Z_i$ for all $i \neq j \in [k]$.
 For $i \in [k]$ and $Z$ a connected component of $G - D_i$ define $s(i,Z) \coloneqq |\{j \in [k] \mid D_j \subseteq Z\}|$.
 Pick $i \in [k]$ and $Z_i$ a connected component of $G - D_i$ such that $s(i,Z_i)$ is maximal.
 Suppose towards a contradiction that $s(i,Z_i) < k-1$.
 Then, using Lemma \ref{la:d-sets-do-not-split}, there is a second component $Z_i'$ of the graph $G - D_i$ such that $s(i,Z_i') \geq 1$.
 Suppose that $D_j \subseteq Z_i'$ and let $Z_j$ be the connected component of $G - D_j$ such that $D_i \subseteq Z_j$.
 Then $Z_i \subseteq Z_j$ and hence, $s(j,Z_j) \geq s(i,Z) + 1$.
 This contradicts the maximality of $s(i,Z_i)$.
 Hence, $s(i,Z_i) = k-1$ which means that $D_j \subseteq Z_i$ for all $i \neq j \in [k]$.
 
 For two vertices $v \in D$ and $u \in V(G) \setminus D$ we say that \emph{$v$ is directly reachable from $u$} if there is a path $u = u_1,\dots,u_m = v$ from $u$ to $v$ such that $u_\mu \notin D$ for all $\mu \in [m-1]$.
 For $j \in [k]$ we say that \emph{$D_j$ is directly reachable from $u$} if there is some $v \in D_j$ such that $v$ is directly reachable from $u$.
 Finally, define
 \[d(u) \coloneqq |\{j \in [k] \mid D_j \text{ is directly reachable from } u\}|.\]
 Observe that $d(u) \geq 1$ for all $u \in V(G) \setminus D$ because $G$ is connected.
 Let $U \coloneqq \{u \in V(G) \setminus D \mid d(u) = 1\}$.
 \begin{claim}
  \label{cl:u-invariant}
  $U$ is $\chi$-invariant.
 \end{claim}
 \begin{claimproof}
  Let $C_\sim \subseteq \{\chi(v,w) \mid v,w \in V(G), v \neq w\}$ be the set of colors defined in Corollary \ref{cor:d-sets-invariant} such that $D_1,\dots,D_k$ are precisely the connected components of $G[C_\sim]$.
  Consider the graph $G/C_\sim$.
  By Lemma \ref{la:factor-graph-2-wl} the coloring $\chi/C_\sim$ is $2$-stable on $G/C_\sim$.
  Moreover, $\CD \coloneqq \{D_1,\dots,D_k\}$ is $(\chi/C_\sim)$-invariant by Corollary \ref{cor:d-sets-invariant}.
  Now, $U$ contains all vertices that can reach only one vertex of $\CD$ without visiting another vertex of $\CD$.
  This property is detected by the $2$-dimensional Weisfeiler-Leman algorithm which implies that $\{\{u\} \mid u \in U\}$ is $(\chi/C_\sim)$-invariant.
  So $U$ is $\chi$-invariant.
 \end{claimproof}

 Now let $\widetilde{G}$ be the graph obtained from $G$ by turning each set $D_i$, $i \in [k]$, into a clique.
 Formally, $V(\widetilde{G}) \coloneqq V(G)$ and
 \[E(\widetilde{G}) \coloneqq E(G) \cup \{vw \mid v \neq w \wedge \exists i \in [k]\colon v,w \in D_i\}.\]
 By Lemma \ref{la:d-sets-do-not-split} the connected components of $G - D_j$ are the same as the connected components of the graph $\widetilde{G} - D_j$ for all $j \in [k]$.
 Also, Corollary \ref{cor:d-sets-invariant} and Claim \ref{cl:u-invariant} imply that $\chi|_{(V(G) \setminus U)^2}$ is $2$-stable on the graph $\widetilde{G} - U$.
 
 Now $(\widetilde{G} - U) - D_i$ is connected.
 We claim that $(\widetilde{G} - U) - D_j$ is connected for all $j \in [k]$.
 Consider the graph $G^* \coloneqq (G - U)/C_\sim$ where $C_\sim$ denotes the set of colors from Corollary \ref{cor:d-sets-invariant}.
 Observe that $G^*$ is the graph obtained from $\widetilde{G} - U$ by contracting each clique $D_i$, $i \in [k]$, to a single vertex.
 Since $(\widetilde{G} - U) - D_i$ is connected, we conclude that the vertex $D_i \in V(G^*)$ is not a cut vertex of $G^*$.
 Hence, by Corollary \ref{cor:d-sets-invariant}, Lemma \ref{la:factor-graph-2-wl} and \ref{la:2-wl-cut-vertices} $D_j \in V(G^*)$ is not a cut vertex of $G^*$ for all $j \in [k]$.
 In other words, $(\widetilde{G} - U) - D_j$ is connected for all $j \in [k]$.
 This implies the first part of the lemma.
 
 Also, $R = V(G) \setminus (D \cup U)$.
 Hence, $R$ is $\chi$-invariant by Corollary \ref{cor:d-sets-invariant} and Claim \ref{cl:u-invariant}.
\end{proof}

By Lemma \ref{la:no-edges-between-d-sets} it holds that $E_G(D_i,D_j) = \emptyset$ for all distinct $i,j \in [k]$.
It follows that $N_G(D(v)) \cap Z(v) \subseteq R$ for all $v \in X$.
In particular, $N_G(Z(v)) = N_G(R) \cap D(v)$ for all $v \in X$.
Hence, $|N_G(R) \cap D(v)| < h$ by Theorem \ref{thm:small-separator-for-t-cr-bounded-closure}.
For $i \in [k]$ define
\begin{equation}
 \label{eq:def-s-sets}
 S_i \coloneqq N_G(R) \cap D_i
\end{equation}
and let $S \coloneqq \bigcup_{i \in [k]} S_i$.

\begin{observation}
 \label{obs:s-invariant}
 For all $i \in [k]$ it holds that $|S_i| < h$.
 Also, $S$ is $\chi$-invariant.
\end{observation}

\begin{proof}
 As already argued above, the first part follows from Theorem \ref{thm:small-separator-for-t-cr-bounded-closure}.
 Moreover, $S$ is $\chi$-invariant because $D$ and $R$ are $\chi$-invariant by Corollary \ref{cor:d-sets-invariant} and Lemma \ref{la:root-set}.
\end{proof}

Now let
\[p \coloneqq \min_{i \neq j \in [k]} \min_{v \in D_i,w \in D_j} \dist_G(v,w).\]
Observe that $p$ is indeed a natural number since $G$ is connected.
Also note that $p \geq 3$ by Lemma \ref{la:no-edges-between-d-sets} and \ref{la:no-common-neighbors-between-d-sets}.

Fix some $i \neq j \in [k]$ and $v \in D_i$, $w \in D_j$ such that $\dist_G(v,w) = p$.
Let $v=u_0,\dots,u_p=w$ be a path from $v$ to $w$ to length $p$.
Observe that $u_\mu \in R$ for all $\mu \in [p-1]$.
Moreover, let
\[\bar c \coloneqq (\chi(u_0,u_0),\chi(u_0,u_1),\chi(u_1,u_1),\chi(u_1,u_2),\dots,\chi(u_{p-1},u_{p-1}),\chi(u_{p-1},u_p),\chi(u_p,u_p))\]
be the sequence of vertex- and arc-colors appearing along the path.
A path $w_0,\dots,w_\ell$ is a \emph{$\bar c$-path} if
\[\bar c = (\chi(w_0,w_0),\chi(w_0,w_1),\chi(w_1,w_1),\chi(w_1,w_2),\dots,\chi(w_{\ell-1},w_{\ell-1}),\chi(w_{\ell-1},w_\ell),\chi(w_\ell,w_\ell)).\]
Note that every $\bar c$-path has length exactly $p$.
We define the graph $F$ with vertex set $V(F) \coloneqq \{D_{i'} \mid i' \in [k]\}$ and edge set
\[E(F) \coloneqq \{D_{i'}D_{j'} \mid (\chi/C_\sim)(D_{i'},D_{j'}) = (\chi/C_\sim)(D_{i},D_{j})\}\]
where $C_\sim$ is the set of colors obtained in Corollary \ref{cor:d-sets-invariant}.
Observe that the graph $F$ contains at least one edge.

\begin{remark}
 The graph $F$ depends on the initial choice of the pair $(i,j)$.
 However, the reader is encouraged to think of $F$ being defined via some color $c$ in the image of $\chi/C_\sim$ (which happens to be the color $(\chi/C_\sim)(D_{i},D_{j})$).
\end{remark}

We collect some basic properties of the graph $F$.

\begin{observation}
 \label{obs:f-regular}
 $F$ is regular, i.e., $\deg_F(D_i) = \deg_F(D_j)$ for all $i,j \in [k]$.
\end{observation}

\begin{proof}
 We have $(\chi/C_\sim)(D_i,D_i) = (\chi/C_\sim)(D_j,D_j)$ by Corollary \ref{cor:d-sets-invariant}.
 So $F$ is regular by Lemma \ref{la:factor-graph-2-wl}.
\end{proof}

\begin{lemma}
 \label{la:degree-bound-f}
 For all $i \in [k]$ it holds that $\deg_F(D_i) \geq 12\adeg h^3$.
\end{lemma}

\begin{proof}
 Suppose towards a contradiction that there is some $i \in [k]$ such that $\deg_F(D_i) \leq 12\adeg h^3$.
 Let $J \coloneqq \{j \in [k] \mid D_j \in N_F(D_i)\}$.
 Now pick arbitrary elements $j \in J$, $v \in S_i$, and a color $c \in \{\chi(v,w) \mid w \in S_j\}$.
 Let $W \coloneqq \{w \in V(G) \mid \chi(v,w) = c\}$.
 Then $W \cap S_j \neq \emptyset$.
 Also,
 \[W \subseteq \bigcup_{j \in J}S_j\]
 by Observation \ref{obs:s-invariant} and Lemma \ref{la:factor-graph-2-wl}.
 So $|W| \leq |J| \cdot h \leq 12\adeg h^4$ by Observation\ref{obs:s-invariant}.
 Now let $v' \in X \cap D_i$.
 Then $v \in \cl_t^{(G,\chi)}(v')$ and hence, $W \subseteq \cl_t^{(G,\chi)}(v')$ since $t \geq |W|$.
 But $\cl_t^{(G,\chi)}(v') = D_i$ by definition and $W \nsubseteq D_i$.
 A contradiction.
\end{proof}

\begin{lemma}
 \label{la:f-edge-to-c-path}
 For every $D_iD_j \in E(F)$ there are $v \in D_i$ and $w \in D_j$ such that there is a $\bar c$-path from $v$ to $w$ or there is a $\bar c$-path from $w$ to $v$.
\end{lemma}

\begin{proof}
 By definition of the graph $F$ there exists $D_iD_j \in E(F)$ and $v \in D_i$ and $w \in D_j$ such that there is a $\bar c$-path from $v$ to $w$.
 Now let $D_{i'}D_{j'} \in E(F)$ be another edge.
 By the definition of $F$, Corollary \ref{cor:d-sets-invariant} and Lemma \ref{la:factor-graph-2-wl} there are $v' \in D_{i'}$ and $w \in D_{j'}$ such that $\chi(v,w) = \chi(v',w')$ or $\chi(v,w) = \chi(w',v')$.
 Without loss of generality assume the former holds.
 By the properties of the $2$-dimensional Weisfeiler-Leman algorithm there is a \emph{$\bar c$-walk} from $v'$ to $w'$, i.e., a sequence of vertices $v' = w_0',\dots,w_\ell' = w'$ such that
 \[\bar c = (\chi(w_0,w_0),\chi(w_0,w_1),\chi(w_1,w_1),\chi(w_1,w_2),\dots,\chi(w_{\ell-1},w_{\ell-1}),\chi(w_{\ell-1},w_\ell),\chi(w_\ell,w_\ell)).\]
 (In comparison to a $\bar c$-path, a vertex may occur multiple times on the walk.)
 Since $p$ is the minimal distance between distinct sets $D_i$, $D_j$, $i,j \in [k]$, it follows that $w_0',\dots,w_\ell'$ is a path.
\end{proof}

The following lemma is the crucial step towards the proof of Lemma \ref{la:initial-color-main}.

\begin{lemma}
 \label{la:find-many-disjoint-paths}
 Let $d \coloneqq 4\adeg h^3$. There is a subgraph $\widehat{F} \subseteq F$ with $V(\widehat{F}) = V(F)$ such that
 \begin{enumerate}[label=(\Alph*)]
  \item\label{item:path-set-1} $\deg_{\widehat{F}}(D_i) \leq d$ for all $i \in [k]$,
  \item\label{item:path-set-2} $\sum_{i \in [k]} \deg_{\widehat{F}}(D_i) = 2\cdot |E(\widehat{F})| \geq \adeg h^3k$, and
  \item\label{item:path-set-3} for every $e = D_iD_j \in E(\widehat{F})$ there is a path $P_e$ from $D_i$ to $D_j$ of length $p$ such that all paths $P_e$, $e \in E(\widehat{F})$, are internally vertex-disjoint.
 \end{enumerate}
\end{lemma}

Before proceeding to the proof of Lemma \ref{la:find-many-disjoint-paths}, let us first prove Lemma \ref{la:initial-color-main} assuming Lemma \ref{la:find-many-disjoint-paths} holds true.

\begin{proof}[Proof of Lemma \ref{la:initial-color-main}]
 Let $\widehat{F} \subseteq F$ be the subgraph from Lemma \ref{la:find-many-disjoint-paths} and fix a set of paths $P_e$, $e \in E(\widehat{F})$, satisfying Property \ref{item:path-set-3}.
 By the length constraint, all internal vertices of a path $P_e$, $e \in E(\widehat{F})$, are contained in $R$ and both endvertices are contained in $S$.
 Now consider the graph $\widetilde{F}$ with vertex set $V(\widetilde{F}) \coloneqq S$ and $vw \in E(\widetilde{F})$ whenever there is a path $P_e$, $e \in E(\widehat{F})$, from $v$ to $w$.
 Clearly, $\widetilde{F}$ is a topological subgraph of $G$.
 Then $\sum_{v \in S} \deg_{\widetilde{F}}(v) = 2|E(\widetilde{F})| \geq 2|E(\widehat{F})| \geq \adeg h^3k$.
 On the other hand, $|S| < hk$ by Observation \ref{obs:s-invariant}.
 But this contradicts Theorem \ref{thm:average-degree-excluded-topological}.
 Hence, one of the Assumptions \ref{item:assumption-1} and \ref{item:assumption-2} is false.
\end{proof}

Hence, it remains to prove Lemma \ref{la:find-many-disjoint-paths}.
This is achieved in the next subsection.

\subsection{Constructing Disjoint Paths}
\label{subsec:finding-paths}

Consider some $i \neq j \in [k]$ such that $D_iD_j \in E(F)$ as well as some $v \in D_i$ and $w \in D_j$ such that $\dist_G(v,w) = p$.
Let $v = u_0,\dots,u_p=w$ be a path from $v$ to $w$ of length $p$.
Then $v \in S_i$, $w \in S_j$, and $u_\mu \in R$ for all $\mu \in [p-1]$.
Recall that $p \geq 3$ by Lemma \ref{la:no-edges-between-d-sets} and \ref{la:no-common-neighbors-between-d-sets}.

We define $r \coloneqq \lceil\frac{p-1}{2}\rceil$ and
\[L_i^{\leq r} \coloneqq \{u \in R \mid \min_{v' \in D_i} \dist(v',u) \leq r\}.\]
By the definition of the parameter $r$ we have that $L_i^{\leq r} \cap L_j^{\leq r} = \emptyset$ for all distinct $i,j \in [k]$,
but there exist $i,j \in [k]$ such that $N_G[L_i^{\leq r}] \cap N_G[L_j^{\leq r}] \neq \emptyset$.
Furthermore, for $1 \leq \mu \leq r$ let
\[L_i^{= \mu} \coloneqq \{u \in R \mid \min_{v' \in D_i} \dist(v',u) = \mu\}.\]
A visualization can be found in Figure \ref{fig:d-s-r-f-sets} (the set $V_M$ displayed in the figure is introduced later on).

\begin{figure}
 \centering
 \scalebox{0.9}{
 \begin{tikzpicture}
  \node at (7,5.2) {{\Large $\dots$}};
  \node at (7,1.6) {{\Large $\dots$}};
  
  \foreach[count=\p] \x/\i in {0.4/1,3.2/2,8.8/k-1,11.6/k}{
   \draw[rounded corners,line width=1.6pt,darkpastelgreen] (\x,0) rectangle (\x+2,2.8);
   \node at (\x+1,-0.4) {{\color{darkpastelgreen}$D_{\i}$}};
   \draw[rounded corners,line width=1.6pt,red,fill=red!10] ([shift=(200:0.8)]\x+1,3) arc (200:340:0.8);
   \node at (\x+1,1.9) {{\color{red}$S_{\i}$}};
   
   \node[smallvertex,fill=red] (v1\p) at (\x+0.7,2.55) {};
   \node[smallvertex,fill=red] (v2\p) at (\x+1.3,2.55) {};
   
   \draw[rounded corners,line width=1.6pt,fill,violet!10] (\x,3.6) rectangle (\x + 2,4.3);
   \draw[rounded corners,line width=1.6pt,fill,violet!25] (\x,4.3) rectangle (\x + 2,5.0);
   \draw[rounded corners,line width=1.6pt,fill,violet!40] (\x,5.0) rectangle (\x + 2,5.7);
   \draw[rounded corners,line width=1.6pt,fill,violet!55] (\x,5.7) rectangle (\x + 2,6.4);
   \draw[rounded corners,line width=1.6pt,violet,dashed] (\x,4.3) -- (\x + 2,4.3);
   \draw[rounded corners,line width=1.6pt,violet,dashed] (\x,5.0) -- (\x + 2,5.0);
   \draw[rounded corners,line width=1.6pt,violet,dashed] (\x,5.7) -- (\x + 2,5.7);
   \draw[rounded corners,line width=1.6pt,violet] (\x,3.6) rectangle (\x + 2,6.4);
   \node at (\x+1,6.8) {{\color{violet}$L_{\i}^{\leq r}$}};
   
   \foreach \q/\a in {1/1,2/1,3/1,4/2,5/2,6/2}{
    \node[smallvertex,fill=violet!60] (w\q\p) at (\x+0.28*\q+0.02,3.95) {};
    \draw (v\a\p) edge (w\q\p);
   }
   \foreach \q in {1,2,3,4}{
    \node[smallvertex,fill=violet] (u\q\p) at (\x + \q*0.4,6.05) {};
   }
  }
  
  \node at (-0.6,3.95) {{\color{violet}$L_{i}^{=1}$}};
  \node at (-0.6,5.10) {{\Large $\vdots$}};
  \node at (-0.6,6.05) {{\color{violet}$L_{i}^{=r}$}};
  
  \draw[rounded corners,line width=0.8pt] (0,3.2) rectangle (14,9.2);
  \node at (7,9.6) {$R$};
  \draw[rounded corners,line width=1.6pt,orange,fill=orange!10] (1.4,7.6) rectangle (12.6,8.8);
  \node at (7.0,8.2) {{\color{orange}$V_M$}};
  \node at (4.2,8.0) {{\Large $\dots$}};
  \node at (9.8,8.0) {{\Large $\dots$}};
  
  \node[smallvertex,fill=orange] (x1) at (2.8,8,0) {};
  \node[smallvertex,fill=orange] (x2) at (5.6,8,0) {};
  \node[smallvertex,fill=orange] (x3) at (8.4,8,0) {};
  \node[smallvertex,fill=orange] (x4) at (11.2,8.0) {};
  \draw (x1) edge (u31);
  \draw (x1) edge (u41);
  \draw (x1) edge (u12);
  \draw (x1) edge (u22);
  \draw (x2) edge (u11);
  \draw (x2) edge (u21);
  \draw (x2) edge (u13);
  \draw (x2) edge (u23);
  \draw (x3) edge (u32);
  \draw (x3) edge (u42);
  \draw (x3) edge (u34);
  \draw (x3) edge (u44);
  \draw (x4) edge (u33);
  \draw (x4) edge (u43);
  \draw (x4) edge (u14);
  \draw (x4) edge (u24);
 \end{tikzpicture}
 }
 \caption{Visualization of the interaction between the closure sets.}
 \label{fig:d-s-r-f-sets}
\end{figure}

\begin{observation}
 \label{obs:l-sets-invariant}
 For every $\mu \in [r]$ the set $L^{=\mu} \coloneqq \bigcup_{i \in [k]} L_i^{=\mu}$ is $\chi$-invariant.
 Moreover, there is a set of colors
 \begin{equation}
  C^{\leq r} \subseteq \{\chi(v,w) \mid v,w \in V(G), v \neq w\}
 \end{equation}
 such that $L_1^{\leq r},\dots,L_k^{\leq r}$ are precisely the connected components of $G[C^{\leq r}]$.
 Also,
 \begin{equation}
  (\chi/C^{\leq r})(L_i^{\leq r},L_i^{\leq r}) = (\chi/C^{\leq r})(L_j^{\leq r},L_j^{\leq r})
 \end{equation}
 for all $i,j \in [k]$.
\end{observation}

\begin{proof}
 The sets $D$ and $R$ are $\chi$-invariant by Corollary \ref{cor:d-sets-invariant} and Lemma \ref{la:root-set}.
 Define $L^{=0} \coloneqq D$.
 Then
 \[L^{=\mu} = \Big(N_G(L^{=\mu-1}) \cap R\Big) \setminus \bigcup_{\mu' < \mu} L^{=\mu'}.\]
 Hence, $L^{=\mu}$ is $\chi$-invariant for all $\mu \in [r]$ by induction.
 For the second part observe that $L_1^{\leq r} \cup D_1,\dots,L_k^{\leq r} \cup D_k$ are the connected components of the graph $G[C]$ where
 \[C \coloneqq \{\chi(v,w) \mid vw \in E(G), v,w \in \bigcup_{i \in [k]} L_i^{\leq r} \cup D_i\} \cup C_\sim\]
 where $C_\sim$ is the set of colors from Corollary \ref{cor:d-sets-invariant}.
 In combination with Corollary \ref{cor:d-sets-invariant} and Lemma \ref{la:factor-graph-2-wl} this also implies the third statement.
\end{proof}

We split the problem of constructing the desired paths into two parts: constructing paths within the sets $L_i^{\leq r}$ and building the ``middle part'' of constant size (the ``middle part'' either consists of a single edge or a single vertex and two incident edges).
We start by considering the paths within the sets $L_i^{\leq r}$.
For these parts, we are interested in which vertices from the set $L_i^{=r}$ can be extended to a path that is disjoint from a set of existing paths on $L_i^{\leq r}$.
Actually, for the overall argument to work out, we may even modify the set of existing paths on $L_i^{\leq r}$ as long as we do not change the ``interface'' in the set $L_i^{=r}$, i.e., the endpoints of the paths have to remain the same.

Fix $i \in [k]$.
Let $P_1,\dots,P_\ell$ be a set of $\ell$ vertex-disjoint paths of length $r-1$ from $L_i^{=1}$ to $L_i^{=r}$.
For every $\mu \in [r]$ we define $L_i^\mu(P_1,\dots,P_\ell)$ to be the set of vertices lying on the paths $P_1,\dots,P_\ell$ that are contained in $L_i^{= \mu}$.
Observe that $|L_i^\mu(P_1,\dots,P_\ell)| = \ell$ for all $\mu \in [r]$.
Now we define the \emph{expansion set of $P_1,\dots,P_\ell$} to be the set
\begin{align*}
 \Exp_i(P_1,\dots,P_\ell) \coloneqq \Big\{v \in L_i^{=r} \;\Big|\; &\text{there exist $\ell+1$ vertex-disjoint paths $P_1',\dots,P_{\ell+1}'$}\\
                                                                   &\text{of length $r-1$ from $L_i^{=1}$ to $L_i^{=r}$ such that}\\
                                                                   &L_i^{r}(P_1',\dots,P_{\ell+1}') = L_i^{r}(P_1,\dots,P_\ell) \cup \{v\}\Big\}.
\end{align*}

Recall that $t \geq t(h)$ is a given parameter for Theorem \ref{thm:initial-color-via-wl} and the sets $D_i = \cl_t^{(G,\chi)}(v_i)$ are defined as the $t$-closure of elements $v_i \in D_i$ (see \eqref{eq:def-d-v}).

\begin{lemma}
 \label{la:expansion-set}
 Let $i \in [k]$ and let $P_1,\dots,P_\ell$ be a set of $\ell < 12\adeg^2h^5$ vertex-disjoint paths of length $r-1$ from $L_i^{=1}$ to $L_i^{=r}$.
 Also, let $c \in C_V(G,\chi)$.
 Then
 \[|\Exp_i(P_1,\dots,P_\ell) \cap V_c| \geq \left(1 - \frac{24\adeg^2h^5}{t}\right) |L_i^{=r} \cap V_c|.\]
\end{lemma}

\begin{proof}
 Throughout the proof, let us fix an arbitrary element $x \in X$ such that $D_i = D(x)$ and define $\lambda \coloneqq \tColRef{t}{G,\chi,x}$ to be the $t$-CR-stable coloring after individualizing $x$.
 Observe that $D_i = \{w \in V(G) \mid |[w]_\lambda| = 1\}$ and the sets $F_i^{=\mu}$ are $\lambda$-invariant for all $\mu \in [r]$.
 
 Let $u_r \in L_i^{=r} \cap V_c$ and let $u_1,\dots,u_r$ be a path from $L_i^{=1}$ to $L_i^{=r} \cap V_c$.
 Also let $c_\mu \coloneqq \lambda(u_\mu)$ for all $\mu \in [r]$.
 We argue that
 \begin{equation}
  \label{eq:expansion-set}
  |\Exp_i(P_1,\dots,P_\ell) \cap \lambda^{-1}(c_r)| \geq \left(1 - \frac{24\adeg^2h^5}{t}\right) |\lambda^{-1}(c_r)|.
 \end{equation}
 Clearly, this implies the lemma since
 \begin{align*}
  |\Exp_i(P_1,\dots,P_\ell) \cap V_c| &=    \sum_{c' \in \lambda^{-1}(F_i^{=r} \cap V_c)} |\Exp_i(P_1,\dots,P_\ell) \cap \lambda^{-1}(c')|\\
                                      &\geq \sum_{c' \in \lambda^{-1}(F_i^{=r} \cap V_c)} \left(1 - \frac{24\adeg^2h^5}{t}\right) |\lambda^{-1}(c')|\\
                                      &=    \left(1 - \frac{24\adeg^2h^5}{t}\right) \sum_{c \in \lambda^{-1}(F_i^{=r} \cap V_c)} |\lambda^{-1}(c')|\\
                                      &=    \left(1 - \frac{24\adeg^2h^5}{t}\right) |L_i^{=r} \cap V_c|.
 \end{align*}
 To prove Equation \eqref{eq:expansion-set} we use an alternating-paths argument and define the following directed graph $H$ with vertex set $V(H) \coloneqq F_i^{\leq r}$ and edge set $E(H) \coloneqq E_{\fw} \cup E_{\bw}$.
 The \emph{forward edges} are defined as
 \[E_{\fw} \coloneqq \{(v,w) \mid vw \in E(G) \setminus \bigcup_{j \in [\ell]}E(P_j), v \in L_i^{=\mu}, w \in L_i^{=\mu+1} \text{ for some } \mu \in [r-1]\}.\]
 The \emph{backward edges} are defined as
 \[E_{\bw} \coloneqq \{(v,w) \mid vw \in \bigcup_{j \in [\ell]}E(P_j), v \in L_i^{=\mu+1}, w \in L_i^{=\mu} \text{ for some } \mu \in [r-1]\}.\]
 We consider directed paths that start in $L_i^{=1} \setminus L_i^1(P_1,\dots,P_\ell)$.
 A directed path $v_1,\ldots,v_q$ in $H$ is \emph{admissible} if
 \begin{enumerate}
  \item $v_1 \in L_i^{=1} \setminus L_i^1(P_1,\dots,P_\ell)$,
  \item $(v_\eta,v_{\eta+1}) \in E(H)$, and
  \item if $(v_\eta,v_{\eta+1})\in E_{\fw}$ and $v_{\eta+1} \in \bigcup_{j\in[\ell]} V(P_j)$ then $\eta\leq q-2$ and $(v_{\eta+1},v_{\eta+2}) \in E_{\bw}$
 \end{enumerate}
 for all $\eta \in [q-1]$. Let
 \[A\coloneqq\{v\in V(G) \mid \text{there is an admissible path $v_1,\ldots,v_q$ such that } v_q=v\}.\]
 Also let $A_\mu \coloneqq A \cap F_i^{=\mu}$ for all $\mu \in [r]$.
 \begin{claim}
  \label{cl:extension-set-contains-admissible}
  $A_r \subseteq \Exp_i(P_1,\dots,P_\ell)$.
 \end{claim}
 \begin{claimproof}
  Let $v \in A_r$ and let $v_1,\ldots,v_q$ be an admissible path of minimal length $q$ such that $v = v_q$.
  Let $P_{\ell+1}$ be the corresponding path graph with $V(P_{\ell+1}) \coloneqq \{v_1,\ldots,v_q\}$ and $E(P_{\ell+1}) \coloneqq \{v_iv_{i+1}\mid i\in[q-1]\}$.
  Consider the graph $P$ with vertex set $V(P) \coloneqq \bigcup_{j \in [\ell+1]} V(P_j)$ and edge set
  \[E(P) \coloneqq \left(\bigcup_{j \in [\ell]} E(P_j) \setminus E(P_{\ell+1})\right) \cup \left(E(P_{\ell+1})\setminus\bigcup_{j \in [\ell]} E(P_j)\right).\]
  Then $P$ is the disjoint union of $(\ell+1)$ many paths $P_1',\ldots,P_{\ell+1}'$ (and possibly isolated vertices) from $L_i^{=1}$ to $L_i^{=r}$ such that $L_i^r(P_1',\dots,P_{\ell+1}') = L_i^r(P_1,\dots,P_\ell) \cup \{v\}$.
  To see this, observe that if there is a vertex $w \in V(P_{\ell+1}) \cap V(P_j)$ for some $j \in [\ell]$, then there is also a common adjacent edge $ww' \in E(P_{\ell+1})\cap E(P_j)$ by the definition of an admissible path.
  This implies that $v \in \Exp_i(P_1,\dots,P_\ell)$.
 \end{claimproof}
 
 By the claim, it suffices to provide a lower bound on the size of the set $A_r \cap \lambda^{-1}(c_r)$.
 Towards this end, we analyze the structure of the set $A$.
 For $j \in [\ell]$ let $u_{\mu,j}$ be the unique vertex in the set $V(P_j) \cap F_i^{=\mu}$, $\mu \in [r]$.
 First observe that, if $u_{\mu,j} \in A$, then also $u_{\mu',j} \in A$ for all $\mu' < \mu$ since
 all vertices $u_{\mu',j}$ are reachable with backward edges in $E_{\bw}$.
 
 We call a vertex $b \in \bigcup_{j \in [\ell]} V(P_j) \setminus A$ a \emph{blocking vertex} if
 there is a vertex $v \in \bigcup_{j \in [\ell]} V(P_j) \cap A$ such that $(b,v) \in E_{\bw}$.
 In other words, the vertex $u_{\mu,j}$ is a blocking vertex if $u_{\mu,j} \notin A$ and $u_{\mu-1,j} \in A$
 (and therefore $u_{\mu',j} \in A$ for all $\mu'<\mu$).
 Let $B$ be the set of blocking vertices.
 By the above observation, $|B \cap V(P_j)| \leq 1$ for all $j \in [\ell]$.
 Hence, $|B| \leq \ell$.
 Let $\ell_\mu \coloneqq |B \cap F_i^{=\mu}|$ be the number of blocking vertices on level $\mu$, $\mu \in [r]$.
 
 \begin{claim}
  \label{cl:blocking-vertices}
  $\displaystyle|A_\mu \cap \lambda^{-1}(c_\mu)| \geq \left(\frac{|A_1 \cap \lambda^{-1}(c_1)|}{|\lambda^{-1}(c_1)|} - \frac{\ell_1+\ldots+\ell_\mu}{t}\right)|\lambda^{-1}(c_\mu)|$ for all $\mu \in [r]$.
 \end{claim}
 \begin{claimproof}
  The claim is proved by induction on $\mu \in [r]$.
  The base case $\mu=1$ is immediately clear.
  
  For the inductive step assume that $\mu \geq 1$.
  Since $G[\lambda^{-1}(c_\mu),\lambda^{-1}(c_{\mu+1})]$ is a non-empty biregular graph,
  for each subset $S \subseteq \lambda^{-1}(c_\mu)$ it holds that
  \[\frac{|N_G(S)\cap\lambda^{-1}(c_{\mu+1})|}{|\lambda^{-1}(c_{\mu+1})|}\geq\frac{|S|}{|\lambda^{-1}(c_\mu)|}\]
  as argued in the preliminaries.
  We first argue that
  \[N\big(A_{\mu} \cap \lambda^{-1}(c_\mu)\big) \;\cap\; \lambda^{-1}(c_{\mu+1}) \;\;\;\subseteq\;\;\; A_{\mu+1} \cup B.\]
  Let $v \in A_{\mu} \cap \lambda^{-1}(c_\mu)$ and $w \in N(v) \cap \lambda^{-1}(c_{\mu+1})$.
  If $w \in \bigcup_{j\in[\ell]}V(P_j)$ then $w \in B$ or $w \in A$.
  Otherwise $w \in V(G)\setminus\bigcup_{j\in\ell]}V(P_j)$ and $(v,w) \in E_{\fw}$ which means $w \in A$.
  This shows the inclusion and therefore
  \[    \frac{|A_{\mu + 1} \cap \lambda^{-1}(c_{\mu+1})|}{|\lambda^{-1}(c_{\mu+1})|}
   \geq \frac{|N(A_{\mu} \cap \lambda^{-1}(c_{\mu})) \cap \lambda^{-1}(c_{\mu+1})|-\ell_{\mu+1}}{|\lambda^{-1}(\mu_{i+1})|}
   \geq \frac{|A_{\mu} \cap \lambda^{-1}(c_{\mu})|}{|\lambda^{-1}(c_\mu)|} - \frac{\ell_{\mu+1}}{t}.\]
  By the induction hypothesis,
  \[\frac{|A_{\mu} \cap \lambda^{-1}(c_\mu)|}{|\lambda^{-1}(c_\mu)|} \geq \frac{|A_1 \cap \lambda^{-1}(c_1)|}{|\lambda^{-1}(c_1)|} - \frac{\ell_1+\dots+\ell_\mu}{t}.\]
  In combination this means
  \[\frac{|A_{\mu+1} \cap \lambda^{-1}(c_{\mu+1})|}{|\lambda^{-1}(c_{\mu+1})|} \geq \frac{|A_1 \cap \lambda^{-1}(c_1)|}{|\lambda^{-1}(c_1)|} - \frac{\ell_1+\dots+\ell_\mu+\ell_{\mu+1}}{t}.\]
 \end{claimproof}

 Now, we can prove Equation \eqref{eq:expansion-set}.
 We already observed in Claim \ref{cl:extension-set-contains-admissible} that
 \[A_{r} \subseteq \Exp_i(P_1,\dots,P_\ell).\]
 Moreover, it holds that $|A_{1} \cap \lambda^{-1}(c_1)| = |\lambda^{-1}(c_1)| - \ell$.
 Combining this with Claim \ref{cl:blocking-vertices}, we obtain
 \begin{align*}
  |\Exp_i(P_1,\dots,P_\ell) \cap \lambda^{-1}(c_r)| &\geq |A_{r} \cap \lambda^{-1}(c_r)|\\
                                                    &\geq \left(\frac{|A_1 \cap \lambda^{-1}(c_1)|}{|\lambda^{-1}(c_1)|} - \frac{\ell}{t}\right)|\lambda^{-1}(c_r)|\\
                                                    &\geq \left(\frac{|\lambda^{-1}(c_1)| - 12\adeg^2h^5}{|\lambda^{-1}(c_1)|} - \frac{12\adeg^2h^5}{t}\right)|\lambda^{-1}(c_r)|\\
                                                    &\geq \left(\frac{t - 12\adeg^2h^5}{t} - \frac{12\adeg^2h^5}{t}\right)|\lambda^{-1}(c_r)|\\
                                                    &=    \left(1 - \frac{24\adeg^2h^5}{t}\right)|\lambda^{-1}(c_r)|.
 \end{align*}
\end{proof}

Building on the previous lemma, the critical step becomes the construction of the middle part of the paths $P_e$ that need to be constructed in order to prove Lemma \ref{la:find-many-disjoint-paths}.
Here, we distinguish between two cases depending on the parity of the path length $p$.

\subsubsection{Paths of Odd Length}
\label{subsec:paths-odd}

We first provide a proof for Lemma \ref{la:find-many-disjoint-paths} in the simpler case in which $p$ is odd.
The basic idea is to construct the paths one-by-one, i.e., initially we define $E(\widehat{F})$ to be empty.
In each iteration, the set of edges (as well as the corresponding set of paths) is extended by one until Property \ref{item:path-set-2} is satisfied while always maintaining Properties \ref{item:path-set-1} and \ref{item:path-set-3}.
Observe that Properties \ref{item:path-set-1} and \ref{item:path-set-3} are satisfied initially.

Hence, let us fix a subgraph $\widehat{F} \subseteq F$ which satisfies Properties \ref{item:path-set-1} and \ref{item:path-set-3}, but violates Property \ref{item:path-set-2}.
Let $P_e$, $e \in E(\widehat{F})$, be the corresponding set of paths.
We argue how to extend $\widehat{F}$ by a single edge while maintaining Properties \ref{item:path-set-1} and \ref{item:path-set-3}.

Consider the color $c_M \coloneqq \bar c_{2r+2}$ of the \emph{middle edge} of a $\bar c$-path as well as the color of its two incident vertices $c_L \coloneqq \bar c_{2r+1}$ and $c_R \coloneqq \bar c_{2r+3}$.
Let $E_M \coloneqq \{vw \in E(G) \mid \chi(v,w) = c_M\}$.

For each $D_iD_j \in E(F)$ we define the \emph{witness set} $W(D_iD_j) \coloneqq \{vw \in E_M \mid v \in L_i^{\leq r}, w \in L_j^{\leq r}\}$.
We remark that each $e \in W(D_iD_j)$ appears on a $\bar c$-path from some $v \in D_i$ to some $w \in D_j$, or on a $\bar c$-path from some $w \in D_j$ to some $v \in D_i$.

A partition $\CP$ of a set $A$ is an \emph{equipartition} if $|P| = |P'|$ for all $P,P' \in \CP$.
Observe that, for an equipartition $\CP$, we have that $|A| = |\CP| \cdot |P|$ for all $P \in \CP$.

\begin{lemma}
 \label{la:odd-path-witness-edge-partition}
 The sets $W(D_iD_j)$, $D_iD_j \in E(F)$, form an equipartition of the set $E_M$.
\end{lemma}

\begin{proof}
 Consider the set $C^{\leq r} \subseteq \{\chi(v,w) \mid v,w \in V(G), v \neq w\}$ defined in Observation \ref{obs:l-sets-invariant}
 such that $L_1^{\leq r},\dots,L_k^{\leq r}$ are precisely the connected components of $G[C^{\leq r}]$.
 By Observation \ref{obs:l-sets-invariant},
 \[\{\!\{\chi(v,v) \mid v \in L_i^{\leq r}\}\!\} = \{\!\{\chi(v,v) \mid v \in L_j^{\leq r}\}\!\}\]
 for all $i,j \in [k]$.
 Also,
 \[c_M \in \{\!\{\chi(v,w) \mid v \in L_i^{\leq r}, w \in L_j^{\leq r} \}\!\}\]
 for some $D_iD_j \in E(F)$ by definition of the graph $F$.
 Now, Lemma \ref{la:factor-graph-2-wl} implies that
 \[\{\!\{\chi(v,w) \mid v \in L_i^{\leq r}, w \in L_j^{\leq r} \}\!\} = \{\!\{\chi(v,w) \mid v \in L_{i'}^{\leq r}, w \in L_{j'}^{\leq r} \}\!\}\]
 for all $D_iD_j, D_{i'}D_{j'} \in E(G)$ and $c_M$ only appears in sets $\{\!\{\chi(v,w) \mid v \in L_i^{\leq r}, w \in L_j^{\leq r} \}\!\}$ for $D_iD_j \in E(G)$.
\end{proof}

Now let $i \in [k]$.
Let $P_1^i,\dots,P_{d_i}^i$ denote the paths which are obtained from intersecting the paths $P_e$, $e \in E(\widehat{F})$, with the set $L_i^{\leq r}$.
Observe that $d_i = \deg_{\widehat{F}}(D_i)$, i.e., the number $d_i$ of paths $P_j^i$ obtained this way equals the degree of $D_i$ in the graph $\widehat{F}$.

The basic idea to extend the graph $\widehat{F}$ by a single edge is to apply a counting argument.
Consider the set $E_M$.
We argue that there is some $e \in E_M$ which can be extended to a $\bar c$-path which is internally-vertex disjoint from all paths $P_e$, $e \in E(\widehat{F})$, and provides a witness for extending $\widehat{F}$ by a single edge.
More precisely, in light of Lemma \ref{la:expansion-set}, it suffices to prove the following lemma.
Recall the definition of the parameter $d$ from Lemma \ref{la:find-many-disjoint-paths}.

\begin{lemma}
 \label{la:odd-path-extension-witness}
 There is an edge $vw \in E_M$ such that $vw \in W(D_iD_j)$ and
 \begin{enumerate}[label=(\alph*)]
  \item\label{item:odd-path-extension-witness-1} $\deg_{\widehat{F}}(D_i) < d$ and $\deg_{\widehat{F}}(D_j) < d$,
  \item\label{item:odd-path-extension-witness-2} $D_iD_j \notin E(\widehat{F})$, and
  \item\label{item:odd-path-extension-witness-3} $v \in \Exp_i(P_1^i,\dots,P_{d_i}^i)$ and $w \in \Exp_j(P_1^j,\dots,P_{d_j}^j)$.
 \end{enumerate}
\end{lemma}

\begin{proof}
 We first provide upper bounds on the number edges in $E_M$ violating one of the first two properties.
 Let
 \[\widehat{E}_M^1 \coloneqq \bigcup_{D_i\colon \deg_{\widehat{F}}(D_i) = d}\bigcup_{D_j \in N_F(D_i)} W(D_iD_j)\]
 denote the set of edges violating Property \ref{item:odd-path-extension-witness-1}.
 Let $U \coloneqq \{i \in [k] \mid \deg_{\widehat{F}}(D_i) = d\}$.
 Then $|E(\widehat{F})| \geq \frac{d}{2}|U| = 2\adeg h^3|U|$.
 Since $|E(\widehat{F})| < \frac{1}{2}\adeg h^3k$ it follows that $|U| < \frac{k}{4}$.
 By Observation \ref{obs:f-regular} there is a number $d_F$ such that $\deg_F(D_i) = d_F$ for all $i \in [k]$
 Hence,
 \[|\{D_iD_j \in E(F) \mid \deg_{\widehat{F}}(D_i) = d \vee \deg_{\widehat{F}}(D_j) = d\}| \leq d_F \cdot |U| < \frac{d_Fk}{4} = \frac{|E(F)|}{2}.\]
 So $|\widehat{E}_M^1| \leq \frac{1}{2}|E_M|$ by Lemma \ref{la:odd-path-witness-edge-partition}.
 Next, let
 \[\widehat{E}_M^2 \coloneqq \bigcup_{D_iD_j \in E(\widehat{F})} W(D_iD_j)\]
 denote the set of edges violating Property \ref{item:odd-path-extension-witness-2}.
 We have that $\frac{|E(\widehat{F})|}{|E(F)|} < \frac{1}{12}$ by Lemma \ref{la:degree-bound-f} and the fact that $\widehat{F}$ violates Property \ref{item:path-set-2}.
 So $|\widehat{E}_M^2| \leq \frac{1}{12}|E_M|$ by Lemma \ref{la:odd-path-witness-edge-partition}.
 
 Now let us analyze Property \ref{item:odd-path-extension-witness-3}.
 By Lemma \ref{la:expansion-set} it holds that
 \[|\Exp_i(P_1^i,\dots,P_{d_i}^i) \cap V_{c_L}| \geq \left(1 - \frac{1}{6}\right) |L_i^{=r} \cap V_{c_L}|\]
 and 
 \[|\Exp_i(P_1^i,\dots,P_{d_i}^i) \cap V_{c_R}| \geq \left(1 - \frac{1}{6}\right) |L_i^{=r} \cap V_{c_R}|\]
 for all $i \in [k] \setminus U$.
 Hence,
 \begin{align*}
         &|\{vw \in E_M \setminus \widehat{E}_M^1 \mid vw \in W(D_iD_j), v \in \Exp_i(P_1^i,\dots,P_{d_i}^i), w \in \Exp_j(P_1^j,\dots,P_{d_j}^j)\}|\\
  \geq\; &\frac{2}{3}(|E_M| - |\widehat{E}_M^1|).
 \end{align*}
 So, the number of edges satisfying Property \ref{item:odd-path-extension-witness-1}, \ref{item:odd-path-extension-witness-2}, and \ref{item:odd-path-extension-witness-3} it at least
 \[\frac{2}{3}\left(|E_M| - |\widehat{E}_M^1|\right) - |\widehat{E}_M^2| \geq \frac{2}{3}\left(|E_M| - \frac{1}{2}|E_M|\right) - \frac{1}{12}|E_M| = \frac{1}{4}|E_M| > 0.\]
\end{proof}

\begin{proof}[Proof of Lemma \ref{la:find-many-disjoint-paths} for $p$ odd]
 Let $\widehat{F}$ be a maximal subgraph of $F$ that satisfies Property \ref{item:path-set-1} and \ref{item:path-set-3}.
 Suppose towards a contraction that Property \ref{item:path-set-2} is violated, i.e., $2\cdot|E(\widehat{F})| < \adeg h^3k$.
 Let $P_e$, $e \in E(\widehat{F})$, be the corresponding set of paths guaranteed by Property \ref{item:path-set-3}.
 For $i \in [k]$ let $P_1^i,\dots,P_{d_i}^i$ denote the paths which are obtained from intersecting the paths $P_e$, $e \in E(\widehat{F})$, with the set $L_i^{\leq r}$.
 Observe that $d_i = \deg_{\widehat{F}}(D_i)$ and the path $P_j^i$ has length $r-1$ for all $i \in [k]$ and $j \in [d_i]$.

 By Lemma \ref{la:odd-path-extension-witness} there is an edge $vw \in E_M$ such that $vw \in W(D_iD_j)$ satisfying Property \ref{item:odd-path-extension-witness-1}, \ref{item:odd-path-extension-witness-2}, and \ref{item:odd-path-extension-witness-3}.
 Let $\widehat{F} + D_iD_j$ denote the graph obtained from $\widehat{F}$ by adding the edge $D_iD_j$.
 By Property \ref{item:odd-path-extension-witness-1} the graph $\widehat{F} + D_iD_j$ satisfies Property \ref{item:path-set-1}.
 By Property \ref{item:odd-path-extension-witness-3} and the definition of an extension set there are vertex-disjoint paths $Q_1^i,\dots,Q_{d_i+1}^i$ from $L_i^{=1}$ to $L_i^{=r}$ of length $r-1$ such that
 \[L_i^{=r}(Q_1^i,\dots,Q_{d_i+1}^i) = L_i^{=r}(P_1^i,\dots,P_{d_i}^i) \cup \{v\}\]
 and vertex-disjoint paths $Q_1^j,\dots,Q_{d_j+1}^j$ from $L_j^{=1}$ to $L_j^{=r}$ of length $r-1$ such that
 \[L_j^{=r}(Q_1^j,\dots,Q_{d_j+1}^j) = L_j^{=r}(P_1^j,\dots,P_{d_j}^j) \cup \{w\}.\]
 This gives a set of paths $Q_e$, $e \in E(\widehat{F} + D_iD_j)$, witnessing Property \ref{item:path-set-3} for the graph $\widehat{F} + D_iD_j$.
 Hence, by the maximality of $\widehat{F}$, it holds that $D_iD_j \in E(\widehat{F})$ contracting Property \ref{item:odd-path-extension-witness-2}.
\end{proof}

\subsubsection{Paths of Even Length}
\label{subsec:paths-even}

Next, let us analyse the case in which $p$ is even.
This case is more complicated than the previous case since, for $p$ even, the center of a $\bar c$-path is a vertex (instead of an edge for $p$ odd).
Consider the color $c_M \coloneqq \bar c_{2r+3}$ of the \emph{middle vertex} of a $\bar c$-path as well as the color of its two adjacent vertices $c_L \coloneqq \bar c_{2r+1}$ and $c_R \coloneqq \bar c_{2r+5}$.
Also let $V_M \coloneqq \{v \in V(G) \mid \chi(v,v) = c_M\}$.

To be more precise, a main advantage of the previous case is that, given a middle edge $e \in E_M$, there is a unique edge $D_iD_j \in E(F)$ such that $e$ appears on a witnessing $\bar c$-path from $D_i$ to $D_j$ (see Lemma \ref{la:odd-path-witness-edge-partition}).
For $p$ even, this is not true anymore, i.e., for a middle vertex $v \in V_M$, there may be a large number of edges $D_iD_j \in E(F)$ such that $v$ appears on a witnessing $\bar c$-path from $D_i$ to $D_j$.
In principle, $|V_M|$ may even be significantly smaller than $|E(F)|$.
Hence, to be able to find a sufficient number of vertex-disjoint paths to ensure Property \ref{item:path-set-3}, we first have to ensure that $|V_M| \geq \frac{1}{2}\adeg h^3k$.
Actually, to formulate a counting argument similar to the previous case, we require a slightly larger bound obtained in the next lemma. 

\begin{lemma}
 \label{la:size-v-m}
 $|V_M| \geq 3\adeg h^{3}k$.
\end{lemma}

\begin{proof}
 Let $\widetilde{B} = ([k],V_M,\widetilde{E})$ be the bipartite graph with edge set
 \[\widetilde{E} \coloneqq \{vi \mid v \in V_M,i \in [k], \text{ and there is $w \in F_i^{=r}$ such that $\chi(w,v) = \bar c_{2r+2}$}\}.\]
 Let $E_L \coloneqq \{vw \in E(G) \mid \chi(w,v) = \bar c_{2r+2}\} \subseteq E(G)$.
 For each $vi \in \widetilde{E}$ define $W(vi) \coloneqq \{vw \in E_L \mid w \in F_i^{=r}\}$ to be the set of edges witnessing that $vi \in \widetilde{E}$.
 \begin{claim}
  \label{cl:witness-equipartition}
  The sets $W(vi)$, $vi \in \widetilde{E}$, form an equipartition of the set $E_L$.
 \end{claim}
 \begin{claimproof}
  First observe that $|W(vi)| \geq 1$ for all $vi \in \widetilde{E}$ and each $e \in E_L$ appears in some set $W(vi)$, $vi \in \widetilde{E}$, by Observation \ref{obs:l-sets-invariant}.
  
  Let $C^{\leq r}$ be the set of colors from Observation \ref{obs:l-sets-invariant} such that $L_1^{\leq r},\dots,L_k^{\leq r}$ are precisely the connected components of $G[C^{\leq r}]$.
  By Lemma \ref{la:factor-graph-2-wl} the coloring $\chi/C^{\leq r}$ is $2$-stable on the graph $G/C^{\leq r}$.
  Now, we have that the number of elements in $W(vi)$ equals the number of occurrences of the color $\bar c_{2r+2}$ in the set $(\chi/C^{\leq r})(v,L_i^{\leq r})$.
  Since any two of such multisets are either equal or disjoint by Lemma \ref{la:factor-graph-2-wl}, the claim follows.
 \end{claimproof}
 
 \begin{claim}
  \label{cl:degree-bound-tildeb}
  $\deg_{\widetilde{B}}(i) \geq 36\adeg^2h^5$ for all $i \in [k]$.
 \end{claim}
 \begin{claimproof}
  Let $i \in [k]$ and pick $x \in X \cap D_i$ (recall the definition of the set $X$ from Equation \eqref{eq:def-x}).
  Also, let $\lambda \coloneqq \tColRef{t}{G,\chi,x}$ be the $t$-CR-stable coloring after individualizing $x$.
  Then $D_i = \{w \in V(G) \mid |[w]_\lambda| = 1\}$ and $L_i^{= r}$ is $\lambda$-invariant.
  Since $V_M \cap \cl_t^{(G,\chi)}(x) = \emptyset$ it follows that $|\{v \in V_M \mid \exists w \in L_i^{=r} \colon \chi(w,v) = \bar c_{2r+2}\}| > t$.
  Since $t \geq 36\adeg^2h^5$ (see Equation \eqref{eq:def-t}), it follows that $\deg_{\widetilde{B}}(i) \geq 36\adeg^2h^5$.
 \end{claimproof}
 
 \begin{claim}
  \label{cl:biregular-tildeb}
  $\widetilde{B}$ is biregular.
 \end{claim}
 \begin{claimproof}
  This follows from Observation \ref{obs:l-sets-invariant} and Lemma \ref{la:factor-graph-2-wl}.
 \end{claimproof}
 
 We shall prove that there is a subgraph $\widehat{B} \subseteq \widetilde{B}$ such that
 \begin{enumerate}[label=(\Roman*)]
  \item\label{item:bipartite-path-set-1} $\deg_{\widehat{B}}(i) \leq 12\adeg^2h^5$ for all $i \in [k]$,
  \item\label{item:bipartite-path-set-2} $|E(\widehat{B})| \geq 3\adeg^2h^5k$, and
  \item\label{item:bipartite-path-set-3} for every $e = vi \in E(\widehat{B})$ there is a path $P_e$ from $D_i$ to $v$ of length $r+1$ such that all paths $P_e$, $e \in E(\widehat{B})$, are internally vertex-disjoint.
 \end{enumerate}
 As before, we initialize $\widehat{B}$ with the same vertex set $V(\widehat{B}) \coloneqq V(\widetilde{B})$ and empty edge set.
 Observe that this satisfies Properties \ref{item:bipartite-path-set-1} and \ref{item:bipartite-path-set-3}.
 We iteratively add edges until Property \ref{item:bipartite-path-set-2} is satisfied while maintaining Properties \ref{item:bipartite-path-set-1} and \ref{item:bipartite-path-set-3}.
 
 So let $\widehat{B}$ be the current subgraph and suppose $|E(\widehat{B})| < 3\adeg^2h^5k$.
 We argue that one can add a single edge to $\widehat{B}$ while maintaining Properties \ref{item:bipartite-path-set-1} and \ref{item:bipartite-path-set-3}.
 Let
 \[E_L^1 \coloneqq \bigcup_{vi \in E(\widehat{B})} W(vi)\]
 be the set of edges from $E_L$ that witness existing edges from $\widehat{B}$.
 It follows from Claim \ref{cl:degree-bound-tildeb} that $|E(\widetilde{B})| \geq 36\adeg^2h^5k$.
 Hence, we have that
 \[|E_L^1| = \frac{|E(\widehat{B})|}{|E(\widetilde{B})|} \cdot |E_L| \leq \frac{1}{12}|E_L|\] by Claim \ref{cl:witness-equipartition}.
 Next, let
 \[E_L^2 \coloneqq \bigcup_{vi \in E(\widehat{B})\colon \deg_{\widehat{B}}(i) = 12\adeg^2h^5} W(vi).\]
 Also, let $U \coloneqq \{i \in [k] \mid \deg_{\widehat{B}}(i) = 12\adeg^2h^5\}$.
 Since $E(\widehat{B}) < 3\adeg^2h^5k$ and $|U| \cdot 12\adeg^2h^5 \geq |E(\widehat{B})|$ it follows that $|U| < \frac{k}{4}$.
 Hence, $|\{vi \in E(\widetilde{B}) \mid i \in U\}| < \frac{1}{4}|E(\widetilde{B})|$ using Claim \ref{cl:biregular-tildeb}.
 So
 \[|E_L^2| = \frac{|\{vi \in E(\widetilde{B}) \mid i \in U\}|}{|E(\widetilde{B})|}\cdot |E_L| < \frac{1}{4} |E_L|\]
 by Claim \ref{cl:witness-equipartition}.
 
 For $i \in [k]$, let $P_1^i,\dots,P_{d_i}^i$ denote the result of intersecting the paths $P_e$, $e \in E(\widehat{B})$, with the set $L_i^{\leq r}$.
 Observe that $d_i = \deg_{\widehat{B}}(i)$.
 Also let $V_L \coloneqq \{v \in V(G) \mid \chi(v,v) = c_L\}$.
 For $i \in [k]$ such that $\deg(i) < 12a^2h^5$ let $U_i \coloneqq (V_L \cap L_i^{=r}) \setminus \Exp_i(P_1^i,\dots,P_{d_i}^i)$.
 By Lemma \ref{la:expansion-set} we have that
 \[|U_i| \leq \frac{1}{6}|V_L \cap L_i^{=r}|.\]
 Let
 \[E_L^3 \coloneqq \{wv \in E_L \mid w \in L_{i}^{=r}, \deg_{\widehat{B}}(i) < 12\adeg^2h^5, w \in U_i\}.\]
 We have that
 \[\left|\bigcup_{i \in [k] \setminus U} U_i\right| = \sum_{i \in [k]\setminus U} |U_i| \leq \sum_{i \in [k]\setminus U} \frac{1}{6}|V_L \cap L_i^{=r}| = \sum_{i \in [k]\setminus U} \frac{1}{6}\cdot \frac{|V_L|}{k} \leq \frac{1}{6} |V_L|\]
 using Observation \ref{obs:l-sets-invariant}.
 Moreover,
 \[E_L^3 = \{vw \in E_L \mid w \in \bigcup_{i \in [k] \setminus U} U_i\}.\]
 Since the graph $G[\bar c_{2r+2}]$ (i.e., the graph induced by $E_L$) is biregular, we conclude that $|E_L^3| \leq \frac{1}{6}|E_L|$.
 
 In total, this means $|E_L^1| + |E_L^2| + |E_L^3| < |E_L|$ and there is some edge $vw \in E_L \setminus (E_L^1 \cup E_L^2 \cup E_L^3)$.
 Let $i \in [k]$ such that $vw \in W(vi)$.
 Then $vi \notin E(\widehat{B})$ and $\deg_{\widehat{B}}(i) < 12\adeg^2h^5$, and $w \in \Exp_i(P_1^i,\dots,P_{d_i}^i)$.
 Hence, by the definition of the expansion set, we can add the edge $vi$ to the graph $\widehat{B}$ while maintaining Properties \ref{item:bipartite-path-set-1} and \ref{item:bipartite-path-set-3}.
 
 Repeating this argument until Property \ref{item:bipartite-path-set-2} is satisfied, we obtain a graph $\widehat{B} \subseteq \widetilde{B}$ satisfying Properties \ref{item:bipartite-path-set-1}, \ref{item:bipartite-path-set-2} and \ref{item:bipartite-path-set-3}.
 
 \medskip
 
 Now consider the set of paths $P_e$, $e \in E(\widehat{B})$, from Property \ref{item:bipartite-path-set-3}.
 By the length constraint on the paths, for $vi = e \in E(\widehat{B})$, it holds that $P_e$ is a path from $S_i$ to $v$ with all internal vertices contained in $R$ (cf.\ Equation \eqref{eq:def-s-sets}).
 We construct a graph $B^*$ with vertex set $V(B^*) \coloneqq V_M \cup \bigcup_{i \in [k]} S_i$ and edges $vw$ whenever there is a path $P_e$, $e \in E(\widehat{B})$, from $v$ to $w$.
 Then $B^*$ is a topological subgraph of $G$,
 \[|V(B^*)| = |V_M| + |S| < |V_M| + kh\]
 and
 \[|E(B^*)| = |E(\widehat{B})| \geq 3\adeg^2h^5k\]
 using Property \ref{item:bipartite-path-set-2}.
 Since $B^*$ is a topological subgraph of $G$ and $G$ has no topological subgraph isomorphic to $K_h$ by Assumption \ref{item:assumption-2},
 \[2|E(B^*)| \leq \adeg h^2|V(B^*)|\]
 by Theorem \ref{thm:average-degree-excluded-topological}.
 Together, this implies that
 \[6\adeg^2h^5k \leq \adeg h^2 \cdot |V(B^*)| < \adeg h^2(|V_M| + kh).\]
 Thus,
 \[|V_M| > 3\adeg h^3k.\]
\end{proof}

To build a counting argument similar to the previous case, we also require the following simple auxiliary lemma.

\begin{lemma}
 \label{la:path-in-color-path}
 Let $G$ be a graph and let $X_1,\dots,X_m$ be distinct color classes of a $1$-stable coloring $\lambda$ such that $E_G(X_i,X_{i+1}) \neq \emptyset$.
 Also let $B_i \subseteq X_i$ for all $i \in [m]$ and define $b \coloneqq \sum_{i \in [m]}\frac{|B_i|}{|X_i|}$.
 Suppose that $b < 1$.
 Then there exists a path $u_1,\dots,u_m$ such that $u_i \in X_i \setminus B_i$ for all $i \in [m]$.
\end{lemma}

\begin{proof}
 For $\ell \leq m$ define
 \[R_\ell \coloneqq \{u_\ell \in X_\ell \setminus B_\ell\mid \text{there is a path $u_1,\dots,u_\ell$ such that $u_i \in X_i \setminus B_i$ for all $i \in [\ell]$}\}.\]
 We prove by induction on $\ell \leq m$ that
 \[|X_\ell \setminus R_\ell| \leq |X_\ell| \sum_{i \in [\ell]}\frac{|B_i|}{|X_i|}.\]
 Observe that this implies that $|R_m| > 0$ which gives the desired path.
 
 The base step is trivial since $X_1 \setminus R_1 = B_1$.
 For the inductive step suppose $\ell \geq 1$.
 Then
 \[N_G(R_\ell) \cap X_{\ell+1} \subseteq R_{\ell+1} \cup B_{\ell+1}.\]
 Hence,
 \[|R_{\ell+1}| \geq |X_{\ell+1} \cap N_G(R_{\ell})| - |B_{\ell+1}| \geq |R_\ell| \frac{|X_{\ell+1}|}{|X_\ell|} - |B_{\ell+1}|.\]
 This means that
 \begin{align*}
  |X_{\ell+1} \setminus R_{\ell+1}| &\leq |X_{\ell+1}| - |R_\ell| \frac{|X_{\ell+1}|}{|X_\ell|} + |B_{\ell+1}|\\
                                    &\leq |X_{\ell+1}|\left( 1 - \frac{|R_\ell|}{|X_\ell|} + \frac{|B_{\ell+1}|}{|X_{\ell+1}|}\right)\\
                                    &=    |X_{\ell+1}|\left( \frac{|X_\ell \setminus R_\ell|}{|X_\ell|} + \frac{|B_{\ell+1}|}{|X_{\ell+1}|}\right)\\
                                    &\leq |X_{\ell+1}|\left( \sum_{i \in [\ell]}\frac{|B_i|}{|X_i|} + \frac{|B_{\ell+1}|}{|X_{\ell+1}|}\right)\\
                                    &=    |X_{\ell+1}|\sum_{i \in [\ell+1]}\frac{|B_i|}{|X_i|}.
 \end{align*}
\end{proof}

Now, in order to construct the desired set of paths, we proceed similar to the case that $p$ is odd.
The paths are constructed one-by-one and in each iteration, the set of paths is extended until Property \ref{item:path-set-2} is satisfied while always maintaining Properties \ref{item:path-set-1} and \ref{item:path-set-3}.
So fix a subgraph $\widehat{F} \subseteq F$ which satisfies Properties \ref{item:path-set-1} and \ref{item:path-set-3}, but violates Property \ref{item:path-set-2} (initially, the edge set is empty).
Let $P_e$, $e \in E(\widehat{F})$, be the corresponding set of paths.

For $i \in [k]$ let $P_1^i,\dots,P_{d_i}^i$ denote the result of intersecting the paths $P_e$, $e \in E(\widehat{F})$, with the set $L_i^{\leq r}$.
Observe that $d_i = \deg_{\widehat{F}}(D_i)$.

Let $V_L \coloneqq \{v \in V(G) \mid \chi(v,v) = c_L\}$ and $V_R \coloneqq \{v \in V(G) \mid \chi(v,v) = c_R\}$.
Moreover, define
\begin{align*}
 B_L \coloneqq\;\;\; &\{v \in V_L \mid v \in L_i^{=r} \setminus \Exp_i(P_1^i,\dots,P_{d_i}^i)\}\\
              \cup\; &\{v \in V_L \mid v \in L_i^{=r} \wedge \deg_{\widehat{F}}(D_i) = d\}
\end{align*}
and
\begin{align*}
 B_R \coloneqq\;\;\; &\{v \in V_R \mid v \in L_i^{=r} \setminus \Exp_i(P_1^i,\dots,P_{d_i}^i)\}\\
              \cup\; &\{v \in V_R \mid v \in L_i^{=r} \wedge \deg_{\widehat{F}}(D_i) = d\}
\end{align*}
Also define
\[X_M \coloneqq \{(u,v,w) \mid \chi(u,v) = \bar c_{2r+2}, \chi(v,w) = \bar c_{2r+4}, u \in L_i^{=r}, w \in L_j^{=r}, D_iD_j \in E(F)\}\]
and let
\begin{align*}
 B_M \coloneqq\;\;\; &\{(u,v,w) \in X_M \mid v \in V_M((P_e)_{e \in E(\widehat{F})})\}\\
              \cup\; &\{(u,v,w) \in X_M \mid u \in L_i^{=r}, w \in L_j^{=r}, D_iD_j \in E(\widehat{F})\}
\end{align*}
Consider graph $\widetilde{B} \coloneqq (V_L \uplus X_M \uplus V_R, \widetilde{E})$ where
\[\widetilde{E} \coloneqq \{u(u,v,w),w(u,v,w) \mid (u,v,w) \in X_M\}.\]
Also define $\widetilde{\lambda}\colon V(\widetilde{B}) \rightarrow \{1,2,3\}$ via
\[\widetilde{\lambda}(v) \coloneqq \begin{cases}
                                    1 &\text{if } v \in V_L\\
                                    2 &\text{if } v \in X_M\\
                                    3 &\text{if } v \in V_R
                                   \end{cases}\]
For $v \in V_M$ define $X_M(v) \coloneqq \{(u,w) \mid (u,v,w) \in X_M\}$.

\begin{lemma}
 \label{la:x-m-blocks}
 The coloring $\widetilde{\lambda}$ is $1$-stable on the graph $\widetilde{B}$.
 Also, $|X_M(v)| = |X_M(v')|$ for all $v,v' \in V_M$.
\end{lemma}

\begin{proof}
 Clearly, for each $(u,v,w) \in X_M$, it holds that $|N_{\widetilde{B}}((u,v,w)) \cap \widetilde{\lambda}^-1(c)| = 1$ for both $c \in \{1,3\}$.
 So consider $u,u' \in V_L$ such that $u \in F_i^{=r}$ and $u \in F_{i'}^{=r}$.
 Since $\chi(u,u) = \chi(u',u')$ we conclude that
 \[|\{v \in V_M \mid \chi(u,v) = \bar c_{2r+2}\}| = |\{v' \in V_M \mid \chi(u',v') = \bar c_{2r+2}\}|.\]
 Also, given vertices $v,v'$ such that $\chi(u,v) = \bar c_{2r+2}$ and $\chi(u',v') = \bar c_{2r+2}$ it holds that
 \[|\{w \in V_R \mid \chi(v,w) = \bar c_{2r+4}, \chi(u,w) \in C\}| = |\{w \in V_R \mid \chi(v',w') = \bar c_{2r+4}, \chi(u',w') \in C\}|\]
 for every set of colors $C \subseteq \{\chi(v'',w'') \mid v'',w'' \in V(G), v'' \neq w''\}$.
 By the definition of the graph $F$, Corollary \ref{cor:d-sets-invariant}, Observation \ref{obs:l-sets-invariant} and Lemma \ref{la:factor-graph-2-wl} there is a set of colors $C \subseteq \{\chi(v'',w'') \mid v'',w'' \in V(G), v'' \neq w''\}$
 such that $(u,w) \in C$ if and only if there is an edge $D_iD_j \in E(F)$ such that $u \in L_i^{=r}$ and $w \in L_j^{=r}$.
 By the definition of the set $X_M$ this implies that $\deg_{\widetilde{B}}(u) = \deg_{\widetilde{B}}(u')$.
 By symmetry, it also holds that $\deg_{\widetilde{B}}(w) = \deg_{\widetilde{B}}(w')$ for all $w,w' \in V_R$.
 Finally, the same arguments also imply that $|X_M(v)| = |X_M(v')|$ for all $v,v' \in V_M$.
\end{proof}

In order to extend the graph $\widehat{F}$ it suffices to find a path $v_L,v_M,v_R$ in the graph $\widetilde{B}$ such that $v_L \in V_L\setminus B_L$, $v_M \in X_M \setminus B_M$, and $v_R \in V_R \setminus B_R$.
By Lemma \ref{la:path-in-color-path} it suffices to argue that $\frac{|B_L|}{|V_L|} + \frac{|B_M|}{|X_M|} + \frac{|B_R|}{|V_R|} < 1$.
Towards this end, we provide upper bounds on the three addends.

\begin{lemma}
 \label{la:even-path-left-right-bound}
 $|B_L| < \frac{3}{8}|V_L|$ and $|B_R| < \frac{3}{8}|V_R|$.
\end{lemma}

\begin{proof}
 By symmetry, it suffices to prove the first inequality.
 Let $U \coloneqq \{i \in [k] \mid \deg_{\widehat{F}}(D_i) = d\}$.
 Then $|E(\widehat{F})| \geq \frac{d}{2}|U| = 2\adeg h^3|U|$.
 Since $|E(\widehat{F})| < \frac{1}{2}\adeg h^3k$ it follows that $|U| < \frac{k}{4}$.
 
 Also, $|V_L \cap F_i^{=r}| = |V_L \cap F_j^{=r}|$ for all $i,j \in [k]$ by Observation \ref{obs:l-sets-invariant} and Lemma \ref{la:factor-graph-2-wl}.
 Now let $i \in [k] \setminus U$.
 Then
 \[|\Exp_i(P_1^i,\dots,P_{d_i}^i) \cap V_L| \geq \frac{5}{6}\cdot|F_i^{=r} \cap V_L| = \frac{5}{6}\cdot\frac{|V_L|}{k}\]
 by Lemma \ref{la:expansion-set}.
 So overall,
 \[|B_L| \leq |U| \cdot \frac{|V_L|}{k} + (k - |U|) \cdot \frac{1}{6}\cdot\frac{|V_L|}{k} < \frac{1}{4}\cdot\frac{|V_L|}{k} + \frac{3}{4}\cdot\frac{1}{6}\cdot\frac{|V_L|}{k} = \frac{3}{8}|V_L|.\]
\end{proof}

Hence, it remains to bound $\frac{|B_M|}{|X_M|}$.
For $D_iD_j \in E(F)$ we define $W(D_iD_j) \coloneqq \{(u,v,w) \in X_M \mid u \in F_i^{=r},w \in F_j^{=r}\}$.

\begin{lemma}
 \label{la:even-path-witness-partition}
 The sets $W(D_iD_j)$, $D_iD_j \in E(F)$, form an equipartition of the set $X_M$.
\end{lemma}

\begin{proof}
 By definition, for each element $(u,v,w) \in X_M$, there is a unique $D_iD_j \in E(F)$ such that $(u,v,w) \in W(D_iD_j)$.
 Hence, it only remains to prove that all sets $W(D_iD_j)$ have the same size.
 
 Let $C^{\leq r} \subseteq \{\chi(v,w) \mid v,w \in V(G), v \neq w\}$ be the set from Observation \ref{obs:l-sets-invariant} such that $L_1^{\leq r},\dots,L_k^{\leq r}$ are precisely the connected components of $G[C^{\leq r}]$.
 Moreover, let $C_\sim \subseteq \{\chi(v,w) \mid v,w \in V(G), v \neq w\}$ be the set from Corollary \ref{cor:d-sets-invariant} such that $D_1,\dots,D_k$ are precisely the connected components of $G[C_\sim]$.
 Finally, define
 \[C \coloneqq C_\sim \cup C^{\leq r} \cup \{\chi(v',w') \mid v'w' \in E(G), v' \in D_i, w' \in F_i^{\leq r}\}.\]
 Then $D_1 \cup L_1^{\leq r},\dots,D_k \cup L_k^{\leq r}$ are precisely the connected components of $G[C]$.
 
 Consider the graph $G/C$.
 By Lemma \ref{la:factor-graph-2-wl} the coloring $\chi/C$ is $2$-stable on the graph $G/C$.
 Moreover, by the definition of the graph $F$ and Lemma \ref{la:factor-graph-2-wl}, there is a color set $C_F$ such that $(\chi/C)(D_i \cup L_i^{\leq r},D_j \cup L_j^{\leq r}) = C_F$ if and only if $D_iD_j \in E(F)$.
 
 By Lemma \ref{la:factor-graph-2-wl} there is a unique color set $C_L$ in the image of $\chi/C$ such that $\bar c_{2r+2} \in C_L$.
 Similarly, there is a unique color set $C_R$ in the image of $\chi/C$ such that $\bar c_{2r+4} \in C_R$.
 Let $n_L$ and $n_R$ be the number of appearances of $\bar c_{2r+2}$ and $\bar c_{2r+4}$ in $C_L$ and $C_R$, respectively.
 
 Also, by the properties of the $2$-dimensional Weisfeiler-Leman algorithm, there is a number $n_M$ such that, for each $D_iD_j \in E(F)$, it holds that
 \[n_M = \big|\big\{v \in V_M \;\big|\; (\chi/C)(D_i \cup L_i^{\leq r},\{v\}) = C_L, (\chi/C)(\{v\},D_j \cup L_j^{\leq r}) = C_R\big\}\big|.\]
 By definition of the numbers $n_L$, $n_M$ and $n_R$ we get that
 \[W(D_iD_j) = n_L \cdot n_M \cdot n_R\]
 for all edges $D_iD_j \in E(F)$.
\end{proof}

\begin{lemma}
 \label{la:even-path-middle-bound}
 $|B_M| < \frac{1}{4}|X_M|$.
\end{lemma}

\begin{proof}
 Recall that $|E(\widehat{F})| < \frac{1}{2}\adeg h^3k$.
 On the other hand, $|E(F)| \geq 6\adeg h^3k$ by Lemma \ref{la:degree-bound-f}.
 This means that
\[|\{(u,v,w) \in X_M \mid u \in L_i^{=r}, w \in L_j^{=r}, D_iD_j \in E(\widehat{F})\}| < \frac{1}{12}|X_M|\]
 using Lemma \ref{la:even-path-witness-partition}.
 Moreover,
 \[|\{(u,v,w) \in X_M \mid v \in V_M((P_e)_{e \in E(\widehat{F})})\}| = \frac{|V_M((P_e)_{e \in E(\widehat{F})})|}{|V_M|}|X_M| = \frac{|E(\widehat{F})|}{|V_M|}|X_M|\]
 by Lemma \ref{la:x-m-blocks}.
 Also $\frac{|E(\widehat{F})|}{|V_M|} < \frac{1}{6}$ by Lemma \ref{la:size-v-m}.
 Together, this means that
 \[|B_M| < \frac{1}{12}|X_M| + \frac{1}{6}|X_M| = \frac{1}{4}|X_M|.\]
\end{proof}

\begin{proof}[Proof of Lemma \ref{la:find-many-disjoint-paths} for $p$ even]
 Let $\widehat{F}$ be a maximal subgraph of $F$ that satisfies Property \ref{item:path-set-1} and \ref{item:path-set-3}.
 Suppose towards a contraction that Property \ref{item:path-set-2} is violated, i.e., $2\cdot|E(\widehat{F})| < \adeg h^3k$.
 Let $P_e$, $e \in E(\widehat{F})$, be the corresponding set of paths guaranteed by Property \ref{item:path-set-3}.
 For $i \in [k]$ let $P_1^i,\dots,P_{d_i}^i$ denote the paths which are obtained from intersecting the paths $P_e$, $e \in E(\widehat{F})$, with the set $L_i^{\leq r}$.
 Observe that $d_i = \deg_{\widehat{F}}(D_i)$ and the path $P_j^i$ has length $r-1$ for all $i \in [k]$ and $j \in [d_i]$.

 By Lemma \ref{la:path-in-color-path}, \ref{la:even-path-left-right-bound} and Lemma \ref{la:even-path-middle-bound}, there is a path $v_L,v_M,v_R$ in the graph $\widetilde{B}$ such that $v_L \in V_L\setminus B_L$, $v_M \in X_M \setminus B_M$, and $v_R \in V_R \setminus B_R$.
 Suppose that $v_M = (u,v,w)$.
 Observe that $u = v_L$ and $w = v_R$ by the definition of the graph $\widetilde{B}$.
 Also suppose $u \in F_i^{=r}$ and $w \in F_j^{=r}$.
 Let $\widehat{F} + D_iD_j$ denote the graph obtained from $\widehat{F}$ by adding the edge $D_iD_j$.
 
 By the definition of the sets $B_L$ and $B_R$ we conclude that $\deg_{\widehat{F}}(D_i) < d$ and $\deg_{\widehat{F}}(D_j) < d$.
 Hence, $\widehat{F} + D_iD_j$ satisfies Property \ref{item:path-set-1}.
 Also, we conclude that $u \in \Exp_i(P_1^i,\dots,P_{d_i}^i)$ and $w \in \Exp_j(P_1^j,\dots,P_{d_j}^j)$.
 Thus, by definition of an extension set, there are vertex-disjoint paths $Q_1^i,\dots,Q_{d_i+1}^i$ from $L_i^{=1}$ to $L_i^{=r}$ of length $r-1$ such that
 \[L_i^{=r}(Q_1^i,\dots,Q_{d_i+1}^i) = L_i^{=r}(P_1^i,\dots,P_{d_i}^i) \cup \{u\}\]
 and vertex-disjoint paths $Q_1^j,\dots,Q_{d_j+1}^j$ from $L_j^{=1}$ to $L_j^{=r}$ of length $r-1$ such that
 \[L_j^{=r}(Q_1^j,\dots,Q_{d_j+1}^j) = L_j^{=r}(P_1^j,\dots,P_{d_j}^j) \cup \{w\}.\]
 Also, by the definition of the set $B_M$, the vertex $v$ is not occupied by any of the paths $P_e$, $e \in E(\widehat{F})$.
 This gives a set of paths $Q_e$, $e \in E(\widehat{F} + D_iD_j)$, witnessing Property \ref{item:path-set-3} for the graph $\widehat{F} + D_iD_j$.
 
 Also, by the definition of the set $B_M$, we get that $D_iD_j \notin E(\widehat{F})$.
 But this contradicts the maximality of the graph $\widehat{F}$.
\end{proof}

This completes the proof of Lemma \ref{la:find-many-disjoint-paths} and thus, we have also shown Theorem \ref{thm:initial-color-via-wl}.

\section{Group-Theoretic Techniques for Isomorphism Testing}
\label{sec:group-machinery}

Having established the necessary combinatorial tools, we can now turn to assembling the main algorithm.
Towards this end, we require several group-theoretic tools.
All the tools are based on \cite{Neuen22,Wiebking20} building extensions of Babai's algorithm to test isomorphism of hypergraphs and further abstract objects.
Before we can introduce those tools, we first cover some basic terminology on group theory.

Let us also remark that, from this point onwards, the isomorphism algorithm for graphs excluding $K_h$ as a topological subgraph is identical to the algorithm from \cite{GroheNW23} replacing the tools for computing the initial sets $X_1$ and $X_2$.

\subsection{Basics}

For a general background on group theory we refer to \cite{Rotman99}, whereas background on permutation groups can be found in \cite{DixonM96}.

\paragraph{Permutation groups.}

A \emph{permutation group} acting on a set $\Omega$ is a subgroup $\Gamma \leq \Sym(\Omega)$ of the symmetric group.
The size of the permutation domain $\Omega$ is called the \emph{degree} of $\Gamma$.
If $\Omega = [n]$, then we also write $S_n$ instead of $\Sym(\Omega)$.
For $\gamma \in \Gamma$ and $\alpha \in \Omega$ we denote by $\alpha^{\gamma}$ the image of $\alpha$ under the permutation $\gamma$.
The set $\alpha^{\Gamma} = \{\alpha^{\gamma} \mid \gamma \in \Gamma\}$ is the \emph{orbit} of $\alpha$.

For $\alpha \in \Omega$ the group $\Gamma_\alpha = \{\gamma \in \Gamma \mid \alpha^{\gamma} = \alpha\} \leq \Gamma$ is the \emph{stabilizer} of $\alpha$ in $\Gamma$.
The \emph{pointwise stabilizer} of a set $A \subseteq \Omega$ is the subgroup $\Gamma_{(A)} = \{\gamma \in \Gamma \mid\forall \alpha \in A\colon \alpha^{\gamma}= \alpha \}$.
For $A \subseteq \Omega$ and $\gamma \in \Gamma$ let $A^{\gamma} = \{\alpha^{\gamma} \mid \alpha \in A\}$.
The set $A$ is \emph{$\Gamma$-invariant} if $A^{\gamma} = A$ for all $\gamma \in \Gamma$.

For $A \subseteq \Omega$ and a bijection $\theta\colon \Omega \rightarrow \Omega'$ we denote by $\theta[A]$ the restriction of $\theta$ to the domain $A$.
For a $\Gamma$-invariant set $A \subseteq \Omega$, we denote by $\Gamma[A] \coloneqq \{\gamma[A] \mid \gamma \in \Gamma\}$ the induced action of $\Gamma$ on $A$, i.e., the group obtained from $\Gamma$ by restricting all permutations to $A$.
More generally, for every set $\Lambda$ of bijections with domain $\Omega$, we denote by $\Lambda[A] \coloneqq \{\theta[A] \mid \theta \in \Lambda\}$.

Let $\Gamma \leq \Sym(\Omega)$ and $\Gamma' \leq \Sym(\Omega')$.
A \emph{homomorphism from $\Gamma$ to $\Gamma'$} is a mapping $\varphi\colon \Gamma \rightarrow \Gamma'$ such that $\varphi(\gamma)\varphi(\delta) = \varphi(\gamma\delta)$ for all $\gamma,\delta \in \Gamma$.
For $\gamma \in \Gamma$ we denote by $\gamma^{\varphi}$ the $\varphi$-image of $\gamma$.
Similarly, for $\Delta \leq \Gamma$ we denote by $\Delta^{\varphi}$ the $\varphi$-image of $\Delta$ (note that $\Delta^{\varphi}$ is a subgroup of $\Gamma'$).

\paragraph{Algorithms for permutation groups.}

Next, let us review some basic facts about algorithms for permutation groups.
More details can be found in \cite{Seress03}.

In order to perform computational tasks for permutation groups efficiently the groups are represented by generating sets of small size (i.e., polynomial in the size of the permutation domain).
Indeed, most algorithms are based on so-called strong generating sets,
which can be chosen of size quadratic in the size of the permutation domain of the group and can be computed in polynomial time given an arbitrary generating set (see, e.g., \cite{Seress03}).

\begin{theorem}[cf.\ \cite{Seress03}] 
 \label{thm:permutation-group-library}
 Let $\Gamma \leq \Sym(\Omega)$ and let $S$ be a generating set for $\Gamma$.
 Then the following tasks can be performed in time polynomial in $|\Omega|$ and $|S|$:
 \begin{enumerate}
  \item compute the order of $\Gamma$,
  \item given $\gamma \in \Sym(\Omega)$, test whether $\gamma \in \Gamma$,
  \item compute the orbits of $\Gamma$, and
  \item given $A \subseteq \Omega$, compute a generating set for $\Gamma_{(A)}$.
 \end{enumerate}
\end{theorem}

\paragraph{Groups with restricted composition factors.}

In this work, we shall be interested in a particular subclass of permutation groups, namely groups with restricted composition factors.
Let $\Gamma$ be a group.
A \emph{subnormal series} is a sequence of subgroups $\Gamma = \Gamma_0 \trianglerighteq \Gamma_1 \trianglerighteq \dots \trianglerighteq \Gamma_k = \{\id\}$.
The length of the series is $k$ and the groups $\Gamma_{i-1} / \Gamma_{i}$ are the factor groups of the series, $i \in [k]$.
A \emph{composition series} is a strictly decreasing subnormal series of maximal length. 
For every finite group $\Gamma$ all composition series have the same family (considered as a multiset) of factor groups (cf.\ \cite{Rotman99}).
A \emph{composition factor} of a finite group $\Gamma$ is a factor group of a composition series of $\Gamma$.

\begin{definition}
 For $d \geq 2$ let $\mgamma_d$ denote the class of all finite groups $\Gamma$ for which every composition factor of $\Gamma$ is isomorphic to a subgroup of $S_d$.
\end{definition}

Let us point out the fact that there are two similar classes of groups usually referred by $\Gamma_d$ in the literature.
The first is the class denoted by $\mgamma_d$ here originally introduced by Luks \cite{Luks82}, while the second one, for example used in \cite{BabaiCP82}, in particular allows composition factors that are simple groups of Lie type of dimension at most $d$.

\begin{lemma}[Luks \cite{Luks82}]
 \label{la:gamma-d-closure}
 Let $\Gamma \in \mgamma_d$. Then
 \begin{enumerate}
  \item $\Delta \in \mgamma_d$ for every subgroup $\Delta \leq \Gamma$, and
  \item $\Gamma^{\varphi} \in \mgamma_d$ for every homomorphism $\varphi\colon \Gamma \rightarrow \Delta$.
 \end{enumerate}
\end{lemma}

\subsection{Hypergraph Isomorphism}

Two hypergraphs $\CH_1 = (V_1,\CE_1)$ and $\CH_2 = (V_2,\CE_2)$ are isomorphic if there is a bijection $\varphi\colon V_1 \rightarrow V_2$ such that $E \in \CE_1$ if and only if $E^{\varphi} \in \CE_2$ for all $E \in 2^{V_1}$
(where $E^\varphi\coloneqq\{\varphi(v)\mid v\in E\}$ and $2^{V_1}$ denotes the power set of $V_1$).
We write $\varphi\colon \CH_1 \cong \CH_2$ to denote that $\varphi$ is an isomorphism from $\CH_1$ to $\CH_2$.
Consistent with previous notation, we denote by $\Iso(\CH_1,\CH_2)$ the set of isomorphisms from $\CH_1$ to $\CH_2$.
More generally, for $\Gamma \leq \Sym(V_1)$ and a bijection $\theta\colon V_1 \rightarrow V_2$, we define
\[\Iso_{\Gamma\theta}(\CH_1,\CH_2) \coloneqq \{\varphi \in \Gamma\theta \mid \varphi\colon \CH_1 \cong \CH_2\}.\]
The set $\Iso_{\Gamma\theta}(\CH_1,\CH_2)$ is either empty, or it is a coset of $\Aut_\Gamma(\CH_1) \coloneqq \Iso_\Gamma(\CH_1,\CH_1)$, i.e., we have $\Iso_{\Gamma\theta}(\CH_1,\CH_2) = \Aut_\Gamma(\CH_1)\varphi$ where $\varphi \in \Iso_{\Gamma\theta}(\CH_1,\CH_2)$ is an arbitrary isomorphism.
As a result, the set $\Iso_{\Gamma\theta}(\CH_1,\CH_2)$ can be represented efficiently by a generating set for $\Aut_\Gamma(\CH_1)$ and a single isomorphism $\varphi \in \Iso_{\Gamma\theta}(\CH_1,\CH_2)$.
In the remainder of this work, all sets of isomorphisms are represented in this way.

\begin{theorem}[{\cite[Theorem 1.1]{Neuen22}}]
 \label{thm:hypergraph-isomorphism-gamma-d}
 Let $\CH_1 = (V_1,\CE_1)$ and $\CH_2 = (V_2,\CE_2)$ be two hypergraphs and let $\Gamma \leq \Sym(V_1)$ be a $\mgamma_d$-group and $\theta\colon V_1 \rightarrow V_2$ a bijection.
 Then $\Iso_{\Gamma\theta}(\CH_1,\CH_2)$ can be computed in time $(n+m)^{\CO((\log d)^{c})}$ for some absolute constant $c$ where $n \coloneqq |V_1|$ and $m \coloneqq |\CE_1|$.
\end{theorem}

\subsection{Coset-Labeled Hypergraphs}

Actually, for the applications in this paper, the Hypergraph Isomorphism Problem itself turns out to be insufficient.
Instead, we require a generalization of the problem that is, for example, motivated by graph decomposition approaches to graph isomorphism testing (see, e.g., \cite{GroheNSW20, Wiebking20}).
Let $G_1$ and $G_2$ be two graphs and suppose that an algorithm has already computed sets $D_1 \subseteq V(G_1)$ and $D_2 \subseteq V(G_2)$ in an isomorphism-invariant way, i.e., each isomorphism from $G_1$ to $G_2$ also maps $D_1$ to $D_2$.
Moreover, assume that $G_1 - D_1$ is not connected and let $Z_{i,1},\dots,Z_{i,k}$ be the connected components of $G_i - D_i$ (without loss of generality $G_1 - D_1$ and $G_2 - D_2$ have the same number of connected components, otherwise the graphs are non-isomorphic).
Also, let $S_{i,j} \coloneqq  N_{G_i}(Z_{i,j})$ for all $j \in [k]$ and $i \in \{1,2\}$.
A natural strategy for an algorithm is to recursively compute representations for $\Iso(G_1[Z_{1,j_1} \cup S_{1,j_1}], G_2[Z_{2,j_2} \cup S_{2,j_2}])$ for all $j_1,j_2 \in [k]$.
Then, in a second step, the algorithm needs to compute all isomorphisms $\varphi\colon G_1[D_1] \cong G_2[D_2]$ such that there is a bijection $\sigma\colon [k] \rightarrow [k]$ satisfying
\begin{enumerate}[label=(\roman*)]
 \item\label{item:decomposition-strategy-1} $(S_{1,j})^{\varphi} = S_{2,\sigma(j)}$, and
 \item\label{item:decomposition-strategy-2} the restriction $\varphi[S_{1,j}]$ extends to an isomorphism from $G_1[Z_{1,j} \cup S_{1,j}]$ to  $G_2[Z_{2,\sigma(j)} \cup S_{2,\sigma(j)}]$ (in the natural way)
\end{enumerate}
for all $j \in [k]$.

Let us first discuss a simplified case where $S_{1,j_1} \neq S_{1,j_2}$ for all distinct $j_1,j_2 \in [k]$.
Without loss of generality assume the same holds for the second graph, i.e., $S_{2,j_1} \neq S_{2,j_2}$ for all distinct $j_1,j_2 \in [k]$ (otherwise the graph are non-isomorphic).
In this situation the first property naturally translates to an instance of the Hypergraph Isomorphism Problem (in particular, the bijection $\sigma$ is unique for any given bijection $\varphi$).
However, for the second property, we also need to be able to put restrictions on how two hyperedges can be mapped to each other.
Towards this end, we consider hypergraphs with coset-labeled hyperedges where each hyperedge is additionally labeled by a coset.

A \emph{labeling} of a set $V$ is a bijection $\rho\colon V\to\{1,\ldots,|V|\}$.
A \emph{labeling coset} of a set $V$ is a set $\Lambda$ consisting of labelings such that $\Lambda = \Delta\rho \coloneqq \{\delta\rho \mid \delta \in \Delta\}$ for some group $\Delta \leq \Sym(V)$ and some labeling $\rho \colon V \to \{1,\ldots,|V|\}$.
Observe that each labeling coset $\Delta\rho$ can also be written as $\rho\Theta\coloneqq\{\rho\theta\mid\theta\in\Theta\}$ where $\Theta\coloneqq\rho^{-1}\Delta\rho \leq S_{|V|}$.
  
\begin{definition}[Coset-Labeled Hypergraph]
 A \emph{coset-labeled hypergraph} is a tuple $\CH = (V,\CE,\mathfrak{p})$
 where $V$ is a finite set of vertices, $\mathcal{E} \subseteq 2^{V}$ is a set of hyperedges,
 and $\Fp$ is a function that associates with each $E \in \mathcal{E}$ a pair $\Fp(E) =(\rho\Theta,c)$
 consisting of a labeling coset of $E$ and a color $c$ (which is usually a natural number).
 
 Two coset-labeled hypergraphs $\CH_1 = (V_1,\mathcal{E}_1,\Fp_1)$ and $\CH_2 = (V_2,\mathcal{E}_2,\Fp_2)$ are \emph{isomorphic} if there is a bijection $\varphi\colon V_1 \rightarrow V_2$ such that
 \begin{enumerate}
  \item $E \in \CE_1$ if and only if $E^{\varphi} \in \CE_2$ for all $E \in 2^{V_1}$, and
  \item for all $E \in\CE_1$ with $\Fp_1(E) = (\rho_1\Theta_1,c_1)$
    and  $\Fp_2(E^{\varphi}) = (\rho_2\Theta_2,c_2)$ we have $c_1 = c_2$ and
   \begin{equation}\label{coset:1}
    \varphi[E]^{-1}\rho_1\Theta_1=\rho_2\Theta_2.
   \end{equation}
 \end{enumerate}
 In this case, $\varphi$ is an \emph{isomorphism} from $\CH_1$ to $\CH_2$, denoted by $\varphi\colon \CH_1 \cong \CH_2$.
 Observe that (\ref{coset:1}) is equivalent to $c_1 = c_2$, $\Theta_1=\Theta_2$ and $\varphi[E] \in \rho_1\Theta_1\rho_2^{-1}$.
 For $\Gamma \leq \Sym(V_1)$ and a bijection $\theta\colon V_1 \rightarrow V_2$ let
 \begin{equation*}
  \Iso_{\Gamma\theta}(\CH_1,\CH_2) \coloneqq \{\varphi \in \Gamma\theta \mid \varphi\colon\CH_1 \cong \CH_2\}.
 \end{equation*}
\end{definition}

Note that, for two coset-labeled hypergraphs $\CH_1$ and $\CH_2$, the set of isomorphisms $\Iso(\CH_1,\CH_2)$ forms a coset of $\Aut(\CH_1)$ and therefore, it again admits a compact representation.
Indeed, this is a crucial feature of the above definition that again allows the application of group-theoretic techniques.

The next theorem is an immediate consequence of \cite[Theorem 6.6.7]{Neuen19} and Theorem \ref{thm:hypergraph-isomorphism-gamma-d}.

\begin{theorem}
 \label{thm:coset-labeled-hypergraphs-gamma-d}
 Let $\CH_1 = (V_1,\CE_1,\Fp_1)$ and $\CH_2 = (V_2,\CE_2,\Fp_2)$ be two coset-labeled hypergraphs
 such that for all $E\in \CE_1\cup\CE_2$ it holds $|E|\leq d$.
 Also let $\Gamma \leq \Sym(V_1)$ be a $\mgamma_d$-group and $\theta\colon V_1 \rightarrow V_2$ a bijection.
 
 Then $\Iso_{\Gamma\theta}(\CH_1,\CH_2)$ can be computed in time $(n+m)^{\CO((\log d)^{c})}$ for some absolute constant $c$ where $n \coloneqq |V_1|$ and $m \coloneqq |\mathcal{E}_1|$.
\end{theorem}

Observe that Theorem \ref{thm:coset-labeled-hypergraphs-gamma-d} also covers the standard isomorphism problem for hypergraphs (assuming all hyperedges are small) by setting $\Fp(E) \coloneqq (\Sym(E)\rho_E,0)$ where $\rho_E$ is an arbitrary labeling of $E$.
On the other hand, choosing suitable cosets encoding sets of isomorphisms between connected component of $G_1 - D_1$ and $G_2 - D_2$, we can also formulate the restrictions encountered in Property \ref{item:decomposition-strategy-2}.

\subsection{Multiple-Labeling-Cosets}

The last theorem covers the problem discussed in the beginning of the previous subsection assuming that all separators of the first graph are distinct, i.e., $S_{1,j_1} \neq S_{1,j_2}$ for all distinct $j_1,j_2 \in [k]$.
In this subsection, we consider the case in which $S_{1,j_1} = S_{1,j_2}$ for all $j_1,j_2 \in [k]$.
In order to handle the case of identical separators, we build on a framework considered in \cite{SchweitzerW19,Wiebking20}.
(The mixed case in which some, but not all, separators coincide can be handled by combining both techniques.)

\begin{definition}[Multiple-Labeling-Coset]
 A \emph{multiple-labeling-coset} is a tuple $\CX = (V,L,\Fp)$ where $V$ is a finite set, $L = \{\rho_1\Theta_1,\dots,\rho_t\Theta_t\}$ is a set of labeling cosets $\rho_i\Theta_i$, $i \in [t]$, of the set $V$, and $\Fp\colon L \to C$ is a coloring that assigns each labeling coset $\rho\Theta \in L$ a color $\Fp(\rho\Theta) = c$ (as before, $c$ is usually a natural number).
 
 Two multiple-labeling-cosets $\CX_1 = (V_1,L_1,\Fp_1)$ and $\CX_2 = (V_2,L_2,\Fp_2)$ are \emph{isomorphic} if there is a bijection $\varphi\colon V_1 \to V_2$ such that
 \begin{equation}\label{multcoset}
   \big(\;\rho\Theta \in L_1 \;\;\wedge\;\; \Fp_1(\rho\Theta) = c\;\big) \;\;\;\;\Leftrightarrow\;\;\;\; \big(\;\varphi^{-1}\rho\Theta \in L_2 \;\;\wedge\;\; \Fp_2(\varphi^{-1}\rho\Theta) = c\;\big)
 \end{equation}
 for all labeling cosets $\rho\Theta$ of $V$ and colors $c \in C$.
 In this case, $\varphi$ is an \emph{isomorphism} from $\CX_1$ to $\CX_2$, denoted by $\varphi\colon\CX_1\cong\CX_2$.
 Observe that (\ref{multcoset}) is equivalent to $|L_1| = |L_2|$ and for each $\rho_1\Theta_1\in L_1$ there is a $\rho_2\Theta_2\in L_2$
 such that $\Fp_1(\rho_1\Theta_1)=\Fp_2(\rho_2\Theta_2)$ and $\Theta_1=\Theta_2$ and $\varphi\in\rho_1\Theta_1\rho_2^{-1}$.
 Let
 \begin{equation*}
  \Iso(\CX_1,\CX_2)\coloneqq\{\varphi\colon V_1\to V_2\mid\varphi\colon\CX_1\cong\CX_2\}
 \end{equation*}
\end{definition}

Again, the set of isomorphisms $\Iso(\CX_1,\CX_2)$ forms a coset of $\Aut(\CX_1) \coloneqq \Iso(\CX_1,\CX_1)$ and therefore, it again admits a compact representation.

\begin{theorem}[{\cite[Corollary 8]{Wiebking20}}]
 \label{thm:isomorphism-multiple-cosets}
 Let $\CX_1=(V_1,L_1,\Fp_1)$ and $\CX_2=(V_2,L_2,\Fp_2)$ be two multiple-labeling cosets.
 Then $\Iso(\CX_1,\CX_2)$ can be computed  in time $(n+m)^{\CO((\log n)^{c})}$ for some absolute constant $c$ where $n \coloneqq |V_1|$ and $m \coloneqq |L_1|$.
\end{theorem}

\subsection{Isomorphism Tests for $t$-CR-bounded Graphs}

Finally, we require isomorphism tests for $t$-CR-bounded graphs as defined in Section \ref{sec:t-cr}.
More precisely, building on the tools already established above, it suffices to find a ``bounding group'' that is contained in the class $\mgamma_t$ and is guaranteed to contain all isomorphisms between two graphs.
Also, for technical reasons, we argue how to find such a group building on the $t$-closure of certain sets $X_i$ for which we assume such a ``bounding group'' is already given.

\begin{theorem}[{\cite[Lemma 5.2]{Neuen22}}]
 \label{thm:compute-isomorphisms-t-closure}
 Let $G_1,G_2$ be two graphs and let $X_1 \subseteq V(G_1)$ and $X_2 \subseteq V(G_2)$.
 Also, let $\Gamma \leq \Sym(X_1)$ be a $\mgamma_t$-group and $\theta\colon X_1 \rightarrow X_2$ a bijection.
 Moreover, let $D_i \coloneqq \cl_t^{G_i}(X_i)$ for $i \in \{1,2\}$
 and define $\Gamma'\theta'\coloneqq\{\varphi \in \Iso((G_1,X_1),(G_2,X_2)) \mid \varphi[X_1] \in \Gamma\theta\}[D_1]$.
 
 Then $\Gamma' \in \mgamma_t$.
 Moreover, there is an algorithm computing a $\mgamma_t$-group $\Delta \leq \Sym(D_1)$ and a bijection $\delta\colon D_1 \rightarrow D_2$ such that
 \[\Gamma'\theta'\subseteq \Delta\delta\]
 in time $n^{\CO((\log t)^{c})}$ for some absolute constant $c$ where $n\coloneqq |V(G_1)|$.
\end{theorem}

\section{Isomorphism Test for Graphs Excluding a Topological Subgraph}
\label{sec:algorithm-isomorphism}

Building on the tools presented in the previous sections, we can now provide an isomorphism test for graphs excluding $K_h$ as a topological subgraph running in time $n^{\polylog(h)}$.
For the algorithm we shall follow the high-level description provided in Section \ref{sec:t-cr}.

\begin{theorem}
 \label{thm:main}
 There is an algorithm that, given a number $h \in \NN$ and two connected vertex-colored graphs $G_1$ and $G_2$ with $n$ vertices, either correctly concludes that $G_1$ has a topological subgraph isomorphic to $K_h$,
 or decides whether $G_1$ is isomorphic to $G_2$ in time $n^{\CO((\log h)^{c})}$ for some absolute constant $c$.
\end{theorem}

\begin{proof}
 We present a recursive algorithm that, given two vertex-colored graphs $(G_1,\chi_1)$ and $(G_2,\chi_2)$ and a color $c_0$ such that for $S_i \coloneqq \chi_i^{-1}(c_0)$ it holds that $|S_i|<h$ for $i=1,2$,
 either correctly concludes that $G_1$ has a topological subgraph isomorphic to $K_h$ or computes a representation for $\Iso((G_1,\chi_1),(G_2,\chi_2))[S_1]$.
 The color $c_0$ does not have to be in the range of the $\chi_i$.
 Thus initializing it with a color $c_0$ not in the range, we have $|S_i| = 0 < h$, in which case the algorithm simply decides whether $\Iso((G_1,\chi_1),(G_2,\chi_2)) \neq \emptyset$, that is, decides whether $(G_1,\chi_1)$ and $(G_2,\chi_2)$ are isomorphic.
 (For $S_1 = S_2 = \emptyset$, we define $\Iso((G_1,\chi_1),(G_2,\chi_2))[S_1]$ to contain the empty mapping if $(G_1,\chi_1)$ and $(G_2,\chi_2)$ are isomorphic, in the other case $\Iso((G_1,\chi_1),(G_2,\chi_2))[S_1]$ is empty.)
  
 So let $(G_1,\chi_1)$ and $(G_2,\chi_2)$ be the vertex-colored input graphs, and let $c_0$ be a color such that $|S_i| < h$.
 Let $t \coloneqq t(h) = 144\adeg^2h^5$ as defined in Equation \eqref{eq:def-t}.
 The algorithm first applies Corollary \ref{cor:initial-color} to the graphs $G_i$ and the parameter $t$.
 This results in a pair-colorings $\chi_i'$ and colors $c_i \in \{\chi_i'(v,v) \mid v \in V(G_i)\}$, or the algorithm correctly concludes that one of the graphs has a topological subgraph isomorphic to $K_h$.
 
 If a topological subgraph is detected in both input graphs, then, in particular, $G_1$ has a topological subgraph isomorphic to $K_h$.
 If $c_1 \neq c_2$ or a topological subgraph is detected in only one of the graphs, then the input graphs are non-isomorphic.
 So suppose that $c \coloneqq c_1 = c_2$.
 Let $X_i \coloneqq \{v \in V(G) \mid \chi_i'(v,v) = c\}$.
 Then $X_1^{\varphi} = X_2$ for every $\varphi \in \Iso(G_1,G_2)$.
 
 Also let $D_i \coloneqq \cl_t^{G_i}(X_i)$.
 Observe that $D_1^{\varphi} = D_2$ for every $\varphi \in \Iso(G_1,G_2)$.
 Also observe that $S_i \subseteq D_i$ since $|S_i| = |\chi_i^{-1}(c_0)| < h \leq t$.
 Now let $Z_{i,1},\ldots,Z_{i,k}$ be the vertex sets of the connected components of $G_i - D_i$ and define $\CZ_i \coloneqq \{Z_{i,1},\ldots,Z_{i,k}\}$
 (if the number of connected components differs in the two graphs then they are non-isomorphic).
 
 If $k = 1$ and $|D_i| < h$, the algorithm proceeds as follows.
 First, the coloring $\chi_i$, $i \in \{1,2\}$, is updated to take membership in the set $D_i$ into account, i.e., $\chi_i(v)$ is replaced by $\chi_i(v) \coloneqq (\chi_i(v),1)$ if $v \in D_i$ and $\chi_i(v) \coloneqq (\chi_i(v),0)$ if $v \in V(G_i) \setminus D_i$.
 Afterwards, the algorithm computes a set $X_i^1$ according to the above procedure with respect to the input graph $G_i^1\coloneqq G_i - D_i$.
 Let $D_i^1\coloneqq \cl_t^{(G_i,\chi_i)}(X_i^1)$ be the closure of $X_i^1$ in the graph $G_i$ (rather than $G_i^1$).
 Then $D_i^1\supseteq D_i$ since $|D_i| < h \leq t$.
 Moreover,$D_i^1 \supsetneq D_i$ since $X_i^1 \subseteq D_i^1$ and $\emptyset\neq X_i^1\subseteq V(G_i^1)$.
 This procedure is repeated until $|D_i^{j^*}| \geq h$, or $k \geq 2$, or $k = 0$ for some $j^* \geq 1$.
 
 So without loss of generality suppose that $|D_i| \geq h$, $k \geq 2$ or $ k = 0$.
 If $k = 0$ and $|D_i| < h$, then $V(G_i) = D_i$ and therefore $|V(G_i)| < h$, and the statement of the theorem can directly be obtained from Babai's quasipolynomial time isomorphism test \cite{Babai16} since both graphs have size at most $h-1$.
 
 In the following, suppose that $k \geq 2$ or $|D_i| \geq h$.
 Let $S_{i,j} \coloneqq N_{G_i}(Z_{i,j})$ for all $j \in [k]$ and $i \in \{1,2\}$.
 We have $|S_{i,j}| < h$ by Theorem \ref{thm:small-separator-for-t-cr-bounded-closure}.
 Finally, define $H_{i,j} \coloneqq G[Z_{i,j} \cup S_{i,j}]$ and $\chi_{i,j}^{H} \colon V(H_{i,j}) \rightarrow C \times \{0,1\}$ to be the vertex-coloring defined by
 \[\chi_{i,j}^{H}(v) \coloneqq \begin{cases}
                                (\chi_i(v),1) &\text{if } v \in Z_{i,j},\\
                                (\chi_i(v),0) &\text{otherwise}
                               \end{cases}\]
 for all $j \in [k]$ and both $i \in \{1,2\}$.
 Observe that $|V(H_{i,j})| < |V(G_i)|$ for all $j \in [k]$ and $i \in \{1,2\}$.
 Hence, we can compute isomorphisms between the graphs $H_{i,j}$ using recursion.
 For each pair $j_1,j_2 \in [k]$ and $i_1,i_2 \in \{1,2\}$ we compute the set of isomorphisms
 \[\Phi_{j_1,j_2}^{i_1,i_2} \coloneqq \Iso((H_{i_1,j_1},\chi_{i_1,j_1}),(H_{i_2,j_2},\chi_{i_2,j_2}))[S_{i_1,j_1}]\]
 recursively (see also Figure \ref{fig:graph-decomposition}).
 
 \begin{figure}
  \centering
  \scalebox{0.9}{
  \begin{tikzpicture}[scale=0.6,use Hobby shortcut]
   
   \draw[thick,fill=yellow, fill opacity=0.5] (0,1.8) ellipse (5cm and 4cm);
   
   \begin{scope}
   \clip (0,1.8) ellipse (5cm and 4cm);
   \draw[thick, fill=blue, fill opacity=0.3] (0,5.6) ellipse (2.5cm and 2cm);
   \end{scope}
   \node at (0,4.6) {$S_i$};
   
   \draw[thick, fill=red, fill opacity=0.3] (1.7,3) ellipse (1.8cm and 1.2cm);
   \node at (1.7,3) {$X_i$};
   
   \node at (-1.5,2) {$D_i$};

   \begin{scope}[xshift=-3.5cm,yshift=0.5cm]
   \begin{scope}[rotate=-90]
   \gurke{}
   \node at (0,-4) {$Z_{i,1}$};
   \end{scope}
   \begin{scope}[rotate=-50]
   \gurke{}
   \node at (0,-4) {$Z_{i,2}$};
   \randA{}
   \end{scope}
   \begin{scope}[rotate=-10]
   \gurke{}
   \node at (0,-4) {$Z_{i,3}$};
   \end{scope}
   \node at (0,0) {$S_{i,1}$};
   \end{scope}
   
   \begin{scope}[rotate=0,yshift=-1cm]
   \gurke{}
   \node at (0,-4) {$Z_{i,4}$};
   \node at (0,0) {$S_{i,4}$};
   \randB{}
   \end{scope}

   \begin{scope}[xshift=3.5cm,yshift=0.5cm]
   \begin{scope}[rotate=90]
   \gurke{}
   \node at (0,-4) {$Z_{i,7}$};
   \end{scope}
   \begin{scope}[rotate=50]
   \gurke{}
   \node at (0,-4) {$Z_{i,6}$};
   \randA{}
   \end{scope}
   \begin{scope}[rotate=10]
   \gurke{}
   \node at (0,-4) {$Z_{i,5}$};
   \end{scope}
   \node at (0,0) {$S_{i,5}$};
   \end{scope}
   
   \node at (2.5,-6) {$H_{i,P_2}$};
   \node at (-7,2.7) {$H_{i,P_1}$};
   \node at (7,2.7) {$H_{i,P_3}$};
   
  \end{tikzpicture}
  }
  \caption{Visualization of the graph decomposition\protect\footnotemark.}
  \label{fig:graph-decomposition}
 \end{figure}
 \footnotetext{The visualization was designed by Daniel Wiebking and originally appeared in \cite{GroheNW23}.}
 
 \medskip
 For both $i \in \{1,2\}$ we define an equivalence relation $\sim_i$ on $[k]$ via $j_1 \sim_i j_2$ if and only if $S_{i,j_1} = S_{i,j_2}$ for $j_1,j_2 \in [k]$.
 Let $\CP_i\coloneqq\{P_{i,1},\dots,P_{i,p}\}$ be the corresponding partition into equivalence classes.
 For each $P_i \in \CP_i$ let $S_{i,P_i} \coloneqq S_{i,j}$ for some $j \in P_i$.
 Observe that, by definition, $S_{i,P_i}$ does not depend on the choice of $j \in P_i$.
 Also, define $H_{i,P_i} \coloneqq G_i[(\bigcup_{j \in P_i} Z_{i,j}) \cup S_{i,P_i}]$ and let $\chi_{i,P_i}^{H}$ be the coloring defined by
 \[\chi_{i,P_i}^{H}(v) \coloneqq \begin{cases}
                                  (\chi_i(v),0) &\text{if } v \in S_{i,P_i},\\
                                  (\chi_i(v),1) &\text{otherwise}
                                 \end{cases}\]
 for all $v \in V(H_{P_i})$.
 For each $i_1,i_2 \in \{1,2\}$ and $P_1 \in \CP_{i_1}$ and $P_2 \in \CP_{i_2}$ the algorithm computes
 \[\Phi_{P_1,P_2}^{i_1,i_2} \coloneqq \Iso((H_{i_1,P_1},\chi_{i_1,P_1}^{H}),(H_{i_2,P_2},\chi_{i_2,P_2}^{H}))[S_{i_1,P_1}]\]
 as follows.
 Without loss of generality assume that $P_1 \in \CP_1$ and $P_2 \in \CP_2$.
 We formulate the isomorphism problem between $(H_{1,P_1},\chi_{1,P_1}^H)$ and $(H_{2,P_2},\chi_{2,P_2}^H)$ as an instance of multiple-labeling-coset isomorphism.
 We define another equivalence relation $\simeq$ on $P_1 \uplus P_2$  via
 \[j_1\simeq j_2 \;\;\;\Leftrightarrow\;\;\; \Phi_{j_1,j_2}^{i_1,i_2} \neq \emptyset\]
 where $j_1 \in P_{i_1}$ and $j_2 \in P_{i_2}$.

 Again, we partition $P_1 \uplus P_2 = Q_1\cup\ldots\cup Q_q$ into the equivalence classes of the relation $\simeq$.
 For each equivalence class $Q_j$ we fix one representative $j^* \in Q_j$ and pick $i^{*} \in \{1,2\}$ such that $j^{*} \in P_{i^{*}}$.
 Let $\lambda_{j^*} \colon S_{i^{*},P_{i^{*}}}\to[|S_{P_{i^{*}}}|]$ be an arbitrary bijection.
 
 Let $i \in \{1,2\},j_i\in P_i\cap Q_j$ and define $\rho_{j_i}\Gamma_{j_i} \coloneqq\Phi_{j_i,j^*}^{i,i^{*}}\lambda_{j^*}$.
 Let $\CX_{i,P_i} \coloneqq (S_{i,P_i},L_{i,P_i},\Fp_{i,P_i})$ where
 \[L_{i,P_i} \coloneqq \{\rho_{j_i}\Gamma_{j_i}\mid j_i \in P_i\}\] and
 \[\Fp_{i,P_i}(\rho_{j_i}\Gamma_{j_i})\coloneqq\{\!\!\{j\mid j_i'\in P_i\cap Q_j\text{ and }\rho_{j_i}\Gamma_{j_i} = \rho_{j_i'}\Gamma_{j_i'}\}\!\!\}\]
 (for each $j_i'$ such that $\rho_{j_i}\Gamma_{j_i} = \rho_{j_i'}\Gamma_{j_i'}$
 the element $j$ is added to the multiset where $j_i' \in P_i \cap Q_j$).
 \begin{claim}
  $\Phi_{P_1,P_2}^{1,2} = \Iso(\CX_{1,P_1},\CX_{2,P_2})$.
 \end{claim}
 \begin{claimproof}
  Let $\varphi \in \Iso((H_{1,P_1},\chi_{1,P_1}^{H}),(H_{2,P_2},\chi_{2,P_2}^{H}))$ and let $\sigma \colon P_1 \rightarrow P_2$ be the unique bijection such that $Z_{1,j}^{\varphi} = Z_{2,\sigma(j)}$ for all $j \in P_1$.
  Let $j_1 \in P_1$ and consider the labeling coset $\rho_{j_1}\Gamma_{j_1} \in L_{1,P_1}$.
  Let $j_2 \coloneqq \sigma(j_1)$.
  Then $j_1 \simeq j_2$ since $\varphi[Z_{1,j_1}] \in \Iso((H_{1,j_1},\chi_{1,j_1}),(H_{2,j_2},\chi_{2,j_2}))$.
  Let $j^{*}=j_1^*=j_2^*$ be the representative from the equivalence class containing $j_1$ and $j_2$ and pick $i^{*} \in \{1,2\}$ such that $j^{*} \in P_{i^{*}}$.
  Then
  \[\varphi[S_{1,j_1}] \Phi_{j_2,j^{*}}^{2,i^{*}} = \Phi_{j_1,j^{*}}^{1,i^{*}}.\]
  Since $\lambda_{j_1^*}=\lambda_{j_2^*}$, this implies that $(\varphi[S_{1,j_1}])^{-1}\rho_{j_1}\Gamma_{j_1} = \rho_{j_2}\Gamma_{j_2}$ and, since the above statement holds for all $j_1 \in P_1$,
  it also means that $\Fp_{1,P_1}(\rho_{j_1}\Gamma_{j_1}) = \Fp_{2,P_2}(\rho_{j_2}\Gamma_{j_2})$ (i.e., equality between labeling cosets is preserved by the mapping $\sigma$).
  
  For the backward direction let $\varphi \in \Iso(\CX_{1,P_1},\CX_{2,P_2})$.
  This means, there is a bijection $\sigma\colon P_1 \rightarrow P_2$ such that
  \begin{enumerate}[label = (\alph*)]
   \item $j_1 \simeq \sigma(j_1)$, and
   \item $\varphi^{-1}\rho_{j_1}\Gamma_{j_1} = \rho_{\sigma(j_1)}\Gamma_{\sigma(j_1)}$
  \end{enumerate}
  for all $j_1 \in P_1$.
  This means that, for every $j_1 \in P_1$, it holds that $\varphi \in \Phi_{j_1,\sigma(j_1)}^{1,2}$.
  But this implies that $\varphi \in \Phi_{P_1,P_2}^{1,2}$.
 \end{claimproof}
 
 Hence, $\Phi_{P_1,P_2}^{1,2}$ can be computed using Theorem \ref{thm:isomorphism-multiple-cosets}. 
 Next, the algorithm turns to computing the set $\Iso((G_1,\chi_1,v_1),(G_2,\chi_2,v_2))[D_1]$ from the sets $\Phi_{P_1,P_2}^{i_1,i_2}$, $i_1,i_2 \in \{1,2\}$ and $P_1 \in \CP_{i_1}$, $P_2 \in \CP_{i_2}$.
 
 \medskip
 
 Let $v_1 \in X_1$ be an arbitrary vertex.
 For all $v_2 \in X_2$ the algorithm computes a representation of all isomorphisms $\varphi \in \Iso((G_1,\chi_1),(G_2,\chi_2))[D_1]$ such that $\varphi(v_1) = v_2$ as described below.
 The output of the algorithm is the union of all these isomorphisms iterating over all $v_2 \in X_2$.
 Additionally, all mappings are restricted to $S_1$ (recall that $S_1 \subseteq D_1$).
 
 Let $D_i' \coloneqq \cl_t^{(G_i,\chi_i')}(v_i)$ for $v_i\in X_i$ (and recall that $D_i = \cl_t^{G_i}(X_i)$).
 The algorithm first computes a $\mgamma_t$-group $\Gamma \leq \Sym(D_1')$ and a bijection $\gamma\colon D_1' \rightarrow D_2'$ such that
 \[\Iso((G_1',\chi_1',v_1),(G_2',\chi_2',v_2))[D_1'] \subseteq \Gamma\gamma\]
 using Theorem \ref{thm:compute-isomorphisms-t-closure}.
 Observe that $D_i \subseteq D_i'$.
 Let
 \[\Delta\delta \coloneqq \{\gamma' \in \Gamma\gamma \mid D_1^{\gamma'} = D_2\}[D_1].\]
 A representation for $\Delta\delta$ can be computed using Theorem \ref{thm:hypergraph-isomorphism-gamma-d}.
 
 To compute the set of isomorphisms, we now formulate the isomorphism problem between $(G_1,\chi_1,v_1)$ and $(G_2,\chi_2,v_2)$ as an instance of coset-labeled hypergraph isomorphism.
 Let $\CH_i \coloneqq (D_i,\CE_i,\Fp_i)$ where
 \[\CE_i \coloneqq E(G_i[D_i]) \cup \{S_{i,P_i} \mid P_i \in \CP_i\} \cup \{\{v\} \mid v \in D_i\}.\]
 The function $\Fp_i$ is defined separately for all three parts of the set $\CE_i$ (if an element occurs in more than one set of the union, the colors defined with respect to the single sets are combined by concatenating them in a tuple).

 For an edge $vw \in E(G_i[D_i])$ we define $\Fp_i(vw) \coloneqq (\rho_{v,w}\Sym([2]),0)$ where $\rho_{v,w}\colon \{v,w\} \rightarrow \{1,2\}$ with $\rho_{v,w}(v) = 1$ and $\rho_{v,w}(w) = 2$.
 
 In order to define $\Fp_i$ for sets $S_{i,P_i}$, $P_i \in \CP_i$, we first define an equivalence relation $\approx$ on the disjoint union $\CP_1 \uplus \CP_2$ where $P \approx Q$ if $\Iso((H_{i_1,P},\chi_{i_1,P}^{H}),(H_{i_2,Q},\chi_{i_2,Q}^{H})) \neq \emptyset$ for $P \in \CP_{i_1}$ and $Q \in \CP_{i_2}$.
 Let $\CQ_1,\dots,\CQ_r$ be the equivalence classes.
 For each equivalence class $\CQ_j$ we fix one representative $Q_j^* \in \CQ_j$ and
 pick $i^* \in \{1,2\}$ such that $Q_j^* \in \CP_{i^{*}}$.
 Let $\rho_{Q_j^*} \colon S_{Q_j^*} \rightarrow [| S_{Q_j^*}|]$ be an arbitrary bijection.
 Let $i \in \{1,2\}$, $P_i \in \CP_i \cap \CQ_j$ and define
 \[\rho_{i,P_i}\Gamma_{i,P_i}\coloneqq \Iso((H_{i,P_i},\chi_{i,P_i}^{H}),(H_{i^{*},Q_j^*},\chi_{i^{*},Q_j^*}^{H}))[S_{i,P_i}]\rho_{Q_j^*}.\]
 Now, for $P_i \in \CP_i \cap \CQ_j$, we define
 \[\Fp_i(S_{i,P_i}) \coloneqq (\rho_{i,P_i}\Gamma_{i,P_i},j).\]
 (Intuitively speaking, each separator $S_{i,P_i}$ is associated with a color $j$ and a labeling coset $\rho_{i,P_i}\Gamma_{i,P_i}$.
 The color $j$ encodes the isomorphism type of the graph $H_{i,P_i}$ whereas the labeling coset determines which mappings between separators extend to isomorphisms between the corresponding graphs below the separators.)
 
 Finally, for $v \in D_i$, we define $\Fp_i(v) \coloneqq (v \mapsto 1,\chi_i(v) + r)$ (recall that $r$ denotes the number of equivalence classes $\CQ_1,\dots,\CQ_r$).
 Then
 \[\Iso((G_1,\chi_1,v_1),(G_2,\chi_2,v_2))[D_1] = \Iso_{\Delta\delta}(\CH_1,\CH_2)\]
 which can be computed in the desired time by Theorem \ref{thm:coset-labeled-hypergraphs-gamma-d}.
 
 This completes the description of the algorithm.
 The correctness is immediate from the description of the algorithm.
 So it only remains to analyze the running time.
 
 First observe that the number of recursive calls the algorithm performs is at most quadratic in the number of vertices of the input graphs.
 Also, $|\CP_i|\leq n$ and $|S_{i,j}|<h$ for both $i \in \{1,2\}$ and all $j\in[k]$.
 Hence, the computation of all sets $\Phi_{P_1,P_2}^{i_1,i_2}$, $P_1 \in \CP_{i_1}$ and $P_2 \in \CP_{i_2}$, requires time $n^{\CO((\log h)^{c})}$ by Theorem \ref{thm:isomorphism-multiple-cosets}.
 Next, the algorithm iterates over all vertices $v_2 \in X_2$ and computes isomorphisms between coset-labeled hypergraphs using Theorem \ref{thm:coset-labeled-hypergraphs-gamma-d}.
 In total, the algorithm from Theorem \ref{thm:coset-labeled-hypergraphs-gamma-d} is applied $|X_2| \leq n$ times and a single execution requires time $n^{\CO((\log h)^{c})}$.
 Overall, this gives the desired bound on the running time.
\end{proof}

We remark that, by standard reduction techniques, there is also an algorithm computing a representation for the set $\Iso(G_1,G_2)$ in time $n^{\CO((\log h)^{c})}$ assuming $G_1$ excludes $K_h$ as a topological subgraph.

Moreover, the proof of the last theorem also reveals some insight into the structure of the automorphism group of a graph that excludes $K_h$ as a topological subgraph.

Let $G$ be a graph.
A \emph{tree decomposition} for $G$ is a pair $(T,\beta)$ where $T$ is a rooted tree and $\beta\colon V(T) \rightarrow 2^{V(G)}$ such that
\begin{enumerate}
 \item[(T.1)] for every $e \in E(G)$ there is some $t \in V(T)$ such that $e \subseteq \beta(t)$, and
 \item[(T.2)] for every $v \in V(G)$ the graph $T[\{t \in V(T) \mid v \in \beta(t)\}]$ is non-empty and connected.
\end{enumerate}
The \emph{adhesion-width} of $(T,\beta)$ is $\max_{t_1t_2 \in E(T)} |\beta(t_1) \cap \beta(t_2)|$.

Let $v \in V(G)$.
Also, recall that $(\Aut(G))_v = \{\varphi \in \Aut(G) \mid v^{\varphi} = v\}$ denotes the subgroup of the automorphism group of $G$ that stabilizes the vertex $v$.

\begin{theorem}
 \label{thm:structure-aut}
 Let $G$ be a graph that excludes $K_h$ as a topological subgraph.
 Then there is an iso\-mor\-phism-invariant tree decomposition $(T,\beta)$ of $G$ such that
 \begin{enumerate}
  \item the adhesion-width of $(T,\beta)$ is at most $h-1$, and
  \item for every $t \in V(T)$ there is some $v \in \beta(t)$ such that $(\Aut(G))_v[\beta(t)] \in \mgamma_d$ for $d \coloneqq 144\adeg^2h^5$.
 \end{enumerate}
\end{theorem}

The theorem readily follows from the same arguments used to prove Theorem \ref{thm:main}.
Indeed, consider the recursion tree $T$ of the algorithm from Theorem \ref{thm:main} on input $(G,G)$ where each node $t \in V(T)$ is associated with the corresponding set $\beta(t) \coloneqq D_1$.
For $t \in V(T)$ let $v \in X_1 \setminus S_1 \subseteq D_1$ (recall that $S_1 = \beta(s) \cap \beta(t)$ where $s$ is the unique parent node of $t$).
Then $D_1^{\gamma} = D_1$ for all $\gamma \in (\Aut(G))_v$ and $(\Aut(G))_v[\beta(t)] \in \mgamma_d$.
Finally, observe that $X_1 \setminus S_1 \neq \emptyset$ (in a situation where $X_1 \subseteq S_1$, it also holds that $D_1 \subseteq S_1$ and the algorithm from Theorem \ref{thm:main} would recompute a set $X_i^{1}$).

\section{Conclusion}

We presented an isomorphism test for all graphs excluding $K_h$ as a topological subgraph running in time $n^{\polylog(h)}$.
On the technical side, the main contribution towards this algorithm is a combinatorial statement which provides a suitable isomorphism-invariant initial set to apply the $t$-CR algorithm.
As a consequence, we also obtain restrictions on the structure of the automorphism groups of graphs excluding $K_h$ as a topological subgraph.

Overall, the presented result unifies and extends existing isomorphism tests with polylogarithmic parameter dependence in the exponent of the running time, and essentially completes the picture for such algorithms on sparse graph classes.
It is an interesting open question whether the techniques can be extended to graph parameters that include dense graphs.
As a specific example, can isomorphism of graphs of rank-width $k$ be tested in time $n^{\polylog(k)}$?

Of course, it also remains an important question for which graph parameters the isomorphism problem is fixed-parameter tractable.
However, here it is already open whether isomorphism testing parameterized by the maximum degree is fixed-parameter tractable.

\bibliographystyle{plainurl}
\small
\bibliography{literature}

\begin{thebibliography}{10}

\bibitem{Babai16}
L{\'{a}}szl{\'{o}} Babai.
\newblock Graph isomorphism in quasipolynomial time [extended abstract].
\newblock In Daniel Wichs and Yishay Mansour, editors, {\em Proceedings of the
  48th Annual {ACM} {SIGACT} Symposium on Theory of Computing, {STOC} 2016,
  Cambridge, MA, USA, June 18-21, 2016}, pages 684--697. {ACM}, 2016.
\newblock \href {https://doi.org/10.1145/2897518.2897542}
  {\path{doi:10.1145/2897518.2897542}}.

\bibitem{BabaiCP82}
L\'{a}szl\'{o} Babai, Peter~J. Cameron, and P\'{e}ter~P. P\'{a}lfy.
\newblock On the orders of primitive groups with restricted nonabelian
  composition factors.
\newblock {\em J. Algebra}, 79(1):161--168, 1982.
\newblock \href {https://doi.org/10.1016/0021-8693(82)90323-4}
  {\path{doi:10.1016/0021-8693(82)90323-4}}.

\bibitem{BabaiGM82}
L{\'{a}}szl{\'{o}} Babai, D.~Yu. Grigoryev, and David~M. Mount.
\newblock Isomorphism of graphs with bounded eigenvalue multiplicity.
\newblock In Harry~R. Lewis, Barbara~B. Simons, Walter~A. Burkhard, and
  Lawrence~H. Landweber, editors, {\em Proceedings of the 14th Annual {ACM}
  Symposium on Theory of Computing, May 5-7, 1982, San Francisco, California,
  {USA}}, pages 310--324. {ACM}, 1982.
\newblock \href {https://doi.org/10.1145/800070.802206}
  {\path{doi:10.1145/800070.802206}}.

\bibitem{BabaiKL83}
L{\'{a}}szl{\'{o}} Babai, William~M. Kantor, and Eugene~M. Luks.
\newblock Computational complexity and the classification of finite simple
  groups.
\newblock In {\em 24th Annual Symposium on Foundations of Computer Science,
  Tucson, Arizona, USA, 7-9 November 1983}, pages 162--171. {IEEE} Computer
  Society, 1983.
\newblock \href {https://doi.org/10.1109/SFCS.1983.10}
  {\path{doi:10.1109/SFCS.1983.10}}.

\bibitem{BerkholzBG17}
Christoph Berkholz, Paul~S. Bonsma, and Martin Grohe.
\newblock Tight lower and upper bounds for the complexity of canonical colour
  refinement.
\newblock {\em Theory Comput. Syst.}, 60(4):581--614, 2017.
\newblock \href {https://doi.org/10.1007/s00224-016-9686-0}
  {\path{doi:10.1007/s00224-016-9686-0}}.

\bibitem{Bodlaender90}
Hans~L. Bodlaender.
\newblock Polynomial algorithms for graph isomorphism and chromatic index on
  partial k-trees.
\newblock {\em J. Algorithms}, 11(4):631--643, 1990.
\newblock \href {https://doi.org/10.1016/0196-6774(90)90013-5}
  {\path{doi:10.1016/0196-6774(90)90013-5}}.

\bibitem{BollobasT98}
B{\'{e}}la Bollob{\'{a}}s and Andrew Thomason.
\newblock Proof of a conjecture of {M}ader, {E}rd{\H{o}}s and {H}ajnal on
  topological complete subgraphs.
\newblock {\em Eur. J. Comb.}, 19(8):883--887, 1998.
\newblock \href {https://doi.org/10.1006/eujc.1997.0188}
  {\path{doi:10.1006/eujc.1997.0188}}.

\bibitem{CaiFI92}
Jin{-}yi Cai, Martin F{\"{u}}rer, and Neil Immerman.
\newblock An optimal lower bound on the number of variables for graph
  identification.
\newblock {\em Comb.}, 12(4):389--410, 1992.
\newblock \href {https://doi.org/10.1007/BF01305232}
  {\path{doi:10.1007/BF01305232}}.

\bibitem{ChenP19}
Gang Chen and Ilia~N. Ponomarenko.
\newblock Lectures on coherent configurations.
\newblock \url{http://www.pdmi.ras.ru/~inp/ccNOTES.pdf}, 2019.

\bibitem{DixonM96}
John~D. Dixon and Brian Mortimer.
\newblock {\em Permutation Groups}, volume 163 of {\em Graduate Texts in
  Mathematics}.
\newblock Springer-Verlag, New York, 1996.
\newblock \href {https://doi.org/10.1007/978-1-4612-0731-3}
  {\path{doi:10.1007/978-1-4612-0731-3}}.

\bibitem{GroheM15}
Martin Grohe and D{\'{a}}niel Marx.
\newblock Structure theorem and isomorphism test for graphs with excluded
  topological subgraphs.
\newblock {\em {SIAM} J. Comput.}, 44(1):114--159, 2015.
\newblock \href {https://doi.org/10.1137/120892234}
  {\path{doi:10.1137/120892234}}.

\bibitem{GroheN21}
Martin Grohe and Daniel Neuen.
\newblock Recent advances on the graph isomorphism problem.
\newblock In Konrad~K. Dabrowski, Maximilien Gadouleau, Nicholas Georgiou,
  Matthew Johnson, George~B. Mertzios, and Daniël Paulusma, editors, {\em
  Surveys in Combinatorics 2021}, London Mathematical Society Lecture Note
  Series, page 187–234. Cambridge University Press, 2021.
\newblock \href {https://doi.org/10.1017/9781009036214.006}
  {\path{doi:10.1017/9781009036214.006}}.

\bibitem{GroheNS23}
Martin Grohe, Daniel Neuen, and Pascal Schweitzer.
\newblock A faster isomorphism test for graphs of small degree.
\newblock {\em {SIAM} J. Comput.}, 52(6):FOCS18--1--FOCS18--36, 2023.
\newblock \href {https://doi.org/10.1137/19m1245293}
  {\path{doi:10.1137/19m1245293}}.

\bibitem{GroheNSW20}
Martin Grohe, Daniel Neuen, Pascal Schweitzer, and Daniel Wiebking.
\newblock An improved isomorphism test for bounded-tree-width graphs.
\newblock {\em {ACM} Trans. Algorithms}, 16(3):34:1--34:31, 2020.
\newblock \href {https://doi.org/10.1145/3382082} {\path{doi:10.1145/3382082}}.

\bibitem{GroheNW23}
Martin Grohe, Daniel Neuen, and Daniel Wiebking.
\newblock Isomorphism testing for graphs excluding small minors.
\newblock {\em {SIAM} J. Comput.}, 52(1):238--272, 2023.
\newblock \href {https://doi.org/10.1137/21m1401930}
  {\path{doi:10.1137/21m1401930}}.

\bibitem{GroheS15}
Martin Grohe and Pascal Schweitzer.
\newblock Isomorphism testing for graphs of bounded rank width.
\newblock In Venkatesan Guruswami, editor, {\em {IEEE} 56th Annual Symposium on
  Foundations of Computer Science, {FOCS} 2015, Berkeley, CA, USA, 17-20
  October, 2015}, pages 1010--1029. {IEEE} Computer Society, 2015.
\newblock \href {https://doi.org/10.1109/FOCS.2015.66}
  {\path{doi:10.1109/FOCS.2015.66}}.

\bibitem{HopcroftT71}
John~E. Hopcroft and Robert~Endre Tarjan.
\newblock A v{\({^2}\)} algorithm for determining isomorphism of planar graphs.
\newblock {\em Inf. Process. Lett.}, 1(1):32--34, 1971.
\newblock \href {https://doi.org/10.1016/0020-0190(71)90019-6}
  {\path{doi:10.1016/0020-0190(71)90019-6}}.

\bibitem{ImmermanL90}
Neil Immerman and Eric Lander.
\newblock Describing graphs: A first-order approach to graph canonization.
\newblock In Alan~L. Selman, editor, {\em Complexity Theory Retrospective: In
  Honor of Juris Hartmanis on the Occasion of His Sixtieth Birthday, July 5,
  1988}, pages 59--81. Springer New York, New York, NY, 1990.
\newblock \href {https://doi.org/10.1007/978-1-4612-4478-3_5}
  {\path{doi:10.1007/978-1-4612-4478-3_5}}.

\bibitem{Karp72}
Richard~M. Karp.
\newblock Reducibility among combinatorial problems.
\newblock In Raymond~E. Miller and James~W. Thatcher, editors, {\em Proceedings
  of a symposium on the Complexity of Computer Computations, held March 20-22,
  1972, at the {IBM} Thomas J. Watson Research Center, Yorktown Heights, New
  York, {USA}}, The {IBM} Research Symposia Series, pages 85--103. Plenum
  Press, New York, 1972.
\newblock \href {https://doi.org/10.1007/978-1-4684-2001-2\_9}
  {\path{doi:10.1007/978-1-4684-2001-2\_9}}.

\bibitem{KieferN22}
Sandra Kiefer and Daniel Neuen.
\newblock The power of the weisfeiler-leman algorithm to decompose graphs.
\newblock {\em {SIAM} J. Discret. Math.}, 36(1):252--298, 2022.
\newblock \href {https://doi.org/10.1137/20m1314987}
  {\path{doi:10.1137/20m1314987}}.

\bibitem{KieferPS19}
Sandra Kiefer, Ilia Ponomarenko, and Pascal Schweitzer.
\newblock The {W}eisfeiler-{L}eman dimension of planar graphs is at most 3.
\newblock {\em J. {ACM}}, 66(6):44:1--44:31, 2019.
\newblock \href {https://doi.org/10.1145/3333003} {\path{doi:10.1145/3333003}}.

\bibitem{KomlosS96}
J\'{a}nos Koml\'{o}s and Endre Szemer\'{e}di.
\newblock Topological cliques in graphs. {II}.
\newblock {\em Combin. Probab. Comput.}, 5(1):79--90, 1996.
\newblock \href {https://doi.org/10.1017/S096354830000184X}
  {\path{doi:10.1017/S096354830000184X}}.

\bibitem{Luks82}
Eugene~M. Luks.
\newblock Isomorphism of graphs of bounded valence can be tested in polynomial
  time.
\newblock {\em J. Comput. Syst. Sci.}, 25(1):42--65, 1982.
\newblock \href {https://doi.org/10.1016/0022-0000(82)90009-5}
  {\path{doi:10.1016/0022-0000(82)90009-5}}.

\bibitem{Miller83a}
Gary~L. Miller.
\newblock Isomorphism of k-contractible graphs. {A} generalization of bounded
  valence and bounded genus.
\newblock {\em Inf. Control.}, 56(1/2):1--20, 1983.
\newblock \href {https://doi.org/10.1016/S0019-9958(83)80047-3}
  {\path{doi:10.1016/S0019-9958(83)80047-3}}.

\bibitem{Neuen16}
Daniel Neuen.
\newblock Graph isomorphism for unit square graphs.
\newblock In Piotr Sankowski and Christos~D. Zaroliagis, editors, {\em 24th
  Annual European Symposium on Algorithms, {ESA} 2016, August 22-24, 2016,
  Aarhus, Denmark}, volume~57 of {\em LIPIcs}, pages 70:1--70:17. Schloss
  Dagstuhl - Leibniz-Zentrum f{\"{u}}r Informatik, 2016.
\newblock \href {https://doi.org/10.4230/LIPIcs.ESA.2016.70}
  {\path{doi:10.4230/LIPIcs.ESA.2016.70}}.

\bibitem{Neuen19}
Daniel Neuen.
\newblock {\em The Power of Algorithmic Approaches to the Graph Isomorphism
  Problem}.
\newblock PhD thesis, {RWTH} Aachen University, Aachen, Germany, 2019.
\newblock \href {https://doi.org/10.18154/RWTH-2020-00160}
  {\path{doi:10.18154/RWTH-2020-00160}}.

\bibitem{Neuen22}
Daniel Neuen.
\newblock Hypergraph isomorphism for groups with restricted composition
  factors.
\newblock {\em {ACM} Trans. Algorithms}, 18(3):27:1--27:50, 2022.
\newblock \href {https://doi.org/10.1145/3527667} {\path{doi:10.1145/3527667}}.

\bibitem{Ponomarenko89}
Ilia~N. Ponomarenko.
\newblock The isomorphism problem for classes of graphs.
\newblock {\em Dokl. Akad. Nauk SSSR}, 304(3):552--556, 1989.

\bibitem{Ponomarenko91}
Ilia~N. Ponomarenko.
\newblock The isomorphism problem for classes of graphs closed under
  contraction.
\newblock {\em Journal of Soviet Mathematics}, 55(2):1621--1643, Jun 1991.
\newblock \href {https://doi.org/10.1007/BF01098279}
  {\path{doi:10.1007/BF01098279}}.

\bibitem{Rotman99}
Joseph~J. Rotman.
\newblock {\em An Introduction to the Theory of Groups}, volume 148 of {\em
  Graduate Texts in Mathematics}.
\newblock Springer-Verlag, New York, fourth edition, 1995.
\newblock \href {https://doi.org/10.1007/978-1-4612-4176-8}
  {\path{doi:10.1007/978-1-4612-4176-8}}.

\bibitem{SchweitzerW19}
Pascal Schweitzer and Daniel Wiebking.
\newblock A unifying method for the design of algorithms canonizing
  combinatorial objects.
\newblock In Moses Charikar and Edith Cohen, editors, {\em Proceedings of the
  51st Annual {ACM} {SIGACT} Symposium on Theory of Computing, {STOC} 2019,
  Phoenix, AZ, USA, June 23-26, 2019}, pages 1247--1258. {ACM}, 2019.
\newblock \href {https://doi.org/10.1145/3313276.3316338}
  {\path{doi:10.1145/3313276.3316338}}.

\bibitem{Seress03}
\'{A}kos Seress.
\newblock {\em Permutation Group Algorithms}, volume 152 of {\em Cambridge
  Tracts in Mathematics}.
\newblock Cambridge University Press, Cambridge, 2003.
\newblock \href {https://doi.org/10.1017/CBO9780511546549}
  {\path{doi:10.1017/CBO9780511546549}}.

\bibitem{Weisfeiler76}
Boris Weisfeiler.
\newblock {\em On Construction and Identification of Graphs}, volume 558 of
  {\em Lecture Notes in Mathematics}.
\newblock Springer-Verlag, 1976.

\bibitem{WeisfeilerL68}
Boris Weisfeiler and Andrei Leman.
\newblock The reduction of a graph to canonical form and the algebra which
  appears therein.
\newblock {\em NTI, Series 2}, 1968.
\newblock English translation by Grigory Ryabov available at
  \url{https://www.iti.zcu.cz/wl2018/pdf/wl_paper_translation.pdf}.

\bibitem{Wiebking20}
Daniel Wiebking.
\newblock Graph isomorphism in quasipolynomial time parameterized by treewidth.
\newblock In Artur Czumaj, Anuj Dawar, and Emanuela Merelli, editors, {\em 47th
  International Colloquium on Automata, Languages, and Programming, {ICALP}
  2020, July 8-11, 2020, Saarbr{\"{u}}cken, Germany (Virtual Conference)},
  volume 168 of {\em LIPIcs}, pages 103:1--103:16. Schloss Dagstuhl -
  Leibniz-Zentrum f{\"{u}}r Informatik, 2020.
\newblock \href {https://doi.org/10.4230/LIPIcs.ICALP.2020.103}
  {\path{doi:10.4230/LIPIcs.ICALP.2020.103}}.

\end{thebibliography}

\end{document}